\documentclass[journal,onecolumn]{IEEEtran}

\ifCLASSOPTIONcaptionsoff
  \usepackage[nomarkers]{endfloat}
 \let\MYoriglatexcaption\caption
 \renewcommand{\caption}[2][\relax]{\MYoriglatexcaption[#2]{#2}}
\fi

\usepackage[utf8]{inputenc}
\usepackage[T1]{fontenc}
\usepackage{url}
\usepackage{ifthen}
\usepackage{cite}
\usepackage{cancel}
\usepackage{pdfsync}
\usepackage{cite}
\usepackage{enumerate}
\usepackage{graphicx}
\usepackage[cmex10]{amsmath} % prevents amsmath from using a Type 3 font for math within footnotes
\usepackage{amssymb}
\usepackage{mathtools}
\usepackage{xspace}
\usepackage{caption}
\usepackage{subcaption} % added<<<<<<<<<<<
\usepackage[lined,linesnumbered,ruled,commentsnumbered]{algorithm2e} % for algorithm
\usepackage{colonequals}
\usepackage[mathscr]{euscript}
\usepackage{fixltx2e}    % for text subscript
\usepackage{amsthm}
\usepackage{mathrsfs}
\usepackage{hyperref}

\usepackage{float}
\usepackage{epsfig}
\usepackage{epstopdf}

\usepackage{tikz}
\usepackage{pgfplots}
\pgfplotsset{colormap/jet}
\usetikzlibrary{arrows,shapes,backgrounds,plotmarks,positioning}
\pgfplotsset{compat=newest}                       
\pgfplotsset{plot coordinates/math parser=false}

\newcommand{\X}{\mathcal{X}}
\newcommand{\Y}{\mathcal{Y}}
\newcommand{\Sen}{\mathcal{S}}
\newcommand{\M}{\mathcal{M}}
\newcommand{\R}{\mathcal{R}}

\newcommand{\PS}{P_{S}}
\newcommand{\PX}{P_{X}}
\newcommand{\PY}{P_{Y}}
\newcommand{\PSX}{P_{SX}}
\newcommand{\PSgX}{P_{S|X}}
\newcommand{\PXgS}{P_{X|S}}

\newcommand{\PXgY}{P_{X|Y}}

\newcommand{\PYgX}{P_{Y|X}}
\newcommand{\PSY}{P_{SY}}
\newcommand{\PSgY}{P_{S|Y}}

\newcommand{\PYgS}{P_{Y|S}}

\newcommand{\Q}{\mathbf{Q}}
\newcommand{\q}{\mathbf{q}}
\newcommand{\E}{\mathbb{E}}
\newcommand{\Real}{\mathbb{R}}
\newcommand{\RealP}{\mathbb{R}_+}
\DeclareMathOperator*{\argmax}{arg\,max}
\DeclareMathOperator*{\argmin}{arg\,min}
\newcommand{\eps}{\varepsilon}
\newcommand{\epsu}{\varepsilon_{u}}
\newcommand{\epsl}{\varepsilon_{l}}

\newcommand{\XL}{\X_{L}}
\newcommand{\XH}{\X_{H}}
\newcommand{\YL}{\Y_{L}}
\newcommand{\YH}{\Y_{H}}

\newcommand{\lifty}{l(s,y)}

\newcommand{\e}{\mathrm{e}}

\let\emptyset\varnothing

\newtheorem{Theorem}{Theorem}
\newtheorem{Definition}{Definition}
\newtheorem{Proposition}{Proposition}

\begin{document}

\title{On the Lift, Related Privacy Measures, and Applications to Privacy-Utility Tradeoffs}
\author{Mohammad~Amin~Zarrabian, Ni~Ding, and~Parastoo~Sadeghi.
\thanks{M.~A.~Zarrabian is with the College of Engineering, Computing and Cybernetics, Australian National University, Canberra, Australia, e-mail: mohammad.zarrabian@anu.edu.au.}% <-this % stops a space
\thanks{N.~Ding is with the  School of Computing and Information Systems, University of Melbourne, Melbourne, Australia, email: ni.ding@unimelb.edu.au.}% <-this % stops a space
\thanks{P. Sadeghi is with the School of Engineering and Information Technology, University of New South Wales, Canberra, Australia, email: p.sadeghi@unsw.edu.au.}
\thanks{Preliminary results of this work have been published in part at the 2022 IEEE International Conference on Acoustics, Speech and Signal Processing, and IEEE Information Theory Workshop.}}

\maketitle
\begin{abstract}
This paper investigates lift, the likelihood ratio between the posterior and prior belief about sensitive features in a dataset.
Maximum and minimum lifts over sensitive features quantify the adversary’s knowledge gain and should be bounded to protect privacy.
We demonstrate that max and min lifts have a distinct range of values and probability of appearance in the dataset, referred to as \emph{lift asymmetry}.
We propose asymmetric local information privacy (ALIP) as a compatible privacy notion with lift asymmetry, where different bounds can be applied to min and max lifts.
We use ALIP in the watchdog and optimal random response (ORR) mechanisms, the main methods to achieve lift-based privacy.
It is shown that ALIP enhances utility in these methods compared to existing local information privacy, which ensures the same (symmetric) bounds on both max and min lifts.
We propose subset merging for the watchdog mechanism to improve data utility and subset random response for the ORR to reduce complexity.
We then investigate the related lift-based measures, including $\ell_1$-norm, $\chi^2$-privacy criterion, and $\alpha$-lift.
We reveal that they can only restrict max-lift, resulting in significant min-lift leakage.
To overcome this problem, we propose corresponding lift-inverse measures to restrict the min-lift.
We apply these lift-based and lift-inverse measures in the watchdog mechanism.
We show that they can be considered as relaxations of ALIP, where a higher utility can be achieved by bounding only average max and min lifts.
\end{abstract}

% \begin{IEEEkeywords}
% Local Information Privacy; Local Differential Privacy; Watchdog Privacy Mechanism; Optimal Random Response.
% \end{IEEEkeywords}

% \fi

\section{Introduction}

    With the recent emergence of ``Big-Data'', generating, sharing, and analyzing data are proliferating via the advancement of communication systems and machine learning methods. 
    While sharing datasets is essential to achieve social and economic benefits, it may lead to the leakage of private information, which has raised great concern about the privacy preservation of individuals.
    The main approach to protect privacy is perturbing the data via a privacy mechanism. 
    Consider some raw data denoted by random variable $X$ and some sensitive features denoted by $S$, which are correlated via a joint distribution $\PSX\neq \PS\times\PX$.
    A privacy mechanism (characterized by the transition probability $\PYgX$) is applied to publish $Y$ as a sanitized version of $X$ to protect $S$.

    The design of a privacy mechanism depends on the privacy measure. 
    Differential privacy (DP) \cite{2006CalibNoiseDFP,2006DFP,2011DFP} is a widely used notion of privacy. 
    DP restricts the chance of inferring the individual's presence in a dataset. 
    It ensures that neighbored sensitive features $s$ and $s'$, which differ in only one entry, result in a similar output probability distribution, by restricting the ratio between posterior beliefs ${\PYgS(y|s)}/{\PYgS(y|s')}$ below a threshold $\e^{\eps}.$
    The neighborhood assumption is relaxed in the local differential privacy (LDP) \cite{2011Learnprivately, 2013LDPMiniMax, 2014ExtermalMechanism, 2014LDPrateDist, 2018DPHammDistr}, where the ratio between posterior beliefs is restricted below $\e^{\eps}$
    for \textit{any} two sensitive features $s$ and $s'$,  denoted by $\eps$-LDP. The quantity of $\eps$ is known as the  \textit{privacy budget}.
    DP and LDP are considered context-free privacy notions, i.e., they do not take into account the prior distribution $\PS$.
    In contrast, in information-theoretic (IT) privacy, also known as context-aware privacy \cite{2013UPTdatasets,2018CntxtAwareDataAggre}, it is assumed that the distribution of data or an estimation of it is available. 
    Some of the dominant IT privacy measures are mutual information (MI) \cite{2013UPTdatasets,2014PFInfBottneck,2020PUTandPF}, maximal leakage \cite{2016OperationLeak,2016MaxLeakcipher,2020MaxL}, $\alpha$-leakage \cite{2018TunableMsurInfleak}, and local information privacy (LIP) \cite{2012PrivStatisticInfere, 2018CntxtAwareDataAggre, 2019LIPBoundPrior, 2019LIPLaplcian, 2020LIPDataAggr, 2021ContextawareLIP, 2021LinearReductMeth, 2019Watchdog, 2020PropertiesWatchdog, 2022EnhanceUtilWatchdog, 2020PerfectObfusc, 2021DataSanitize}.
    A challenge is that while data perturbation restricts privacy leakage, it necessarily reduces data resolution and datasets' usefulness. 
    Therefore, a privacy mechanism is desired to deliver a satisfactory level of data utility. 
    Depending on the application, data utility is quantified either by measures of similarity between $X$ and $Y$, like f-divergence \cite{2014ExtermalMechanism} and MI \cite{2014ExtermalMechanism,2014PFInfBottneck}, or measures of dissimilarity and error, like Hamming distortion \cite{2014LDPrateDist,2018DPHammDistr} and mean square error \cite{2021ContextawareLIP}, respectively.
    This tension between privacy and utility is known as the privacy-utility tradeoff (PUT).

    In this paper, we consider lift, a pivotal element in IT privacy measures, which is the likelihood ratio between the posterior belief $\PSgY(s|y)$ and prior belief $\PS(s)$ about sensitive features in a dataset:
    \begin{equation}\label{eq:lift}
    l(s,y)=\frac{\PSgY(s|y)}{\PS(s)}=\frac{\PSY(s,y)}{\PS(s)\PY(y)}.
    \end{equation}
    The logarithm of the lift  $i(s,y)=\log\lifty$, which we call \emph{log-lift}, is the information density \cite{2019Watchdog}. 
    For each $y$, the more $\PSgY(s|y)$ differs from $\PS(s)$, the more the adversary gains knowledge about $s$.
    Consequently, both min-lift and max-lift, denoted by $\min_{s}l(s,y)$ and $\max_{s}l(s,y)$, respectively, quantify the highest privacy leakage for each $y$.
    In LIP, min-lift and max-lift are bounded below and above by thresholds $\e^{-\eps}$ and $\e^{\eps}$, respectively, to restrict the adversary's knowledge gain, denoted by $\eps$-LIP. 
    The main privacy mechanisms to achieve $\eps$-LIP are the watchdog mechanism \cite{2019Watchdog,2020PropertiesWatchdog}  and optimal random response (ORR) \cite{2021DataSanitize}. 
    Watchdog mechanism bi-partitions the alphabet of $X$ into low-risk and high-risk symbols, and only high-risk ones are randomized. 
    It has been proved in \cite{2020PropertiesWatchdog} that $X$-invariant randomization (e.g., merging all high-risk symbols) minimizes privacy leakage for the watchdog mechanism. 
    ORR is an optimal mechanism for $\eps$-LIP, which maximizes MI as the utility measure.   
    
%=================================================================

    \subsection{Contributions}
    
        We investigate lift and its related privacy notions like LIP.
        We demonstrate that min-lift and max-lift have distinct values and probability of appearance in the dataset. 
        More specifically, min-lifts have a broader range of values than max-lifts, while max-lifts have a higher likelihood $\PSY(s,y)$ of appearing in the dataset.
        We call this property \textit{lift asymmetry}. 
        However, $\eps$-LIP allocates symmetric privacy budgets to $\min_{s}i(s,y)$ and $\max_{s}i(s,y)$ ($-\eps$ and $\eps$, respectively), which is incompatible with the lift asymmetry. 
        Thus, we propose asymmetric-LIP (ALIP) as an amenable privacy notion to the lift properties, where asymmetric privacy budgets can be allocated to  $\min_{s}i(s,y)$ and $\max_{s}i(s,y)$, denoted by $-\epsl$ and $\epsu$, respectively. 
        We demonstrate that ALIP implies $\eps$-LDP and can result in better utility than LIP in the watchdog and ORR mechanisms. 
        Utility increases by relaxing the bound on the min-lift, which has a lower probability of appearance in the dataset.

        We propose two randomization methods to overcome the low utility of the watchdog mechanism and the high complexity of the ORR mechanism.
        %: 
        In the watchdog mechanism, $X$-invariant randomization perturbs all high-risk symbols together and deteriorates data resolution and utility. 
        On the other hand, ORR suffers from high complexity, which is exponential in the size of datasets.
        To overcome these problems, we propose  \textit{subset merging} and \textit{subset random response} (SRR) perturbation methods that make finer subsets of high-risk symbols and privatize each subset separately.
        Subset merging enhances utility in the watchdog mechanism by applying $X$-invariant randomization to disjoint subsets of high-risk symbols.
        Also, SRR  relaxes the complexity of ORR for large datasets by applying random response solutions on disjoint subsets of high-risk symbols, which results in near-optimal utility.

        Besides LIP, we also consider some recently proposed privacy measures, which we call \textit{lift-based} measures, including $\ell_{1}$-norm \cite{2022DataDsclsurell1Priv}, $\chi^2$-strong privacy \cite{2021StrongChi2}, and $\alpha$-lift \cite{2021Alpha-LiftWatchdog}.
        They have been proposed as the privacy notions stronger than their corresponding average leakages: the total variation distance \cite{2019PUTTotalDistnce}, $\chi^2$-divergence \cite{2018EstimEffcientPriv}, and Sibson MI \cite{2020MaxL,2021Alpha-LiftWatchdog}, respectively.
        We clarify that they only bound max-lift leakage and can cause significant min-lift leakage. 
        Therefore, we propose a corresponding modified version of these measures to restrict min-lift leakage, which we call \textit{lift-inverse} measures. 
        We apply lift-based and lift-inverse measures in the watchdog mechanism with subset merging randomization to investigate their PUT. 
        They result in higher utility than ALIP since they are functions of average lift over sensitive features, thus, can be considered as relaxations of the max and min lift. 
        %
%=================================================================
%=================================================================

\section{Preliminaries}
%=================================================================    
    \subsection{Notation}
        We use the following notation throughout the paper. Capital letters denote discrete random variables, corresponding capital calligraphic letters denote their finite supports, and lowercase letters denote any of their realizations.
        For example, a random variable $X$ has the support $\X,$ and its realization is $x \in \X$.
        For random variables $S$ and $X$, we use $\PSX$ to indicate their joint probability distribution, $\PSgX$ for the conditional distribution of $S$ given $X$, and $\PS$ and $\PX$ for the marginal distributions. 
        Bold capital and lowercase letters are used for matrices and vectors, respectively, and lowercase letters for the corresponding elements of the vectors, e.g., $\mathbf{v}=[v_1,v_2,\cdots,v_{n}]^{T}$.
        We also use $|\cdot|$ for the cardinality of a set, e.g., $|\X|$.
        We denote the natural logarithm by $\log$ and the  set of integers $\{1,2,\cdots,n\}$ by $[n]$. The indicator function is shown by $\mathbf{1}_{\{f\}}$, which is $1$ when $f$ is true and zero otherwise.
%=================================================================
    \subsection{System Model and Privacy Measures}\label{sec:ALIP and Watchdog}

        Consider some useful data intended for sharing and denoted by random variable $X$ with alphabet $\X$.
        It is correlated with some sensitive features $S$ with the alphabet $\Sen$ through a discrete joint distribution $\PSX$.
        To protect the sensitive features, a privacy mechanism is applied to generate a sanitized version of $X$, denoted by $Y$ with the alphabet $\Y$.
        We assume $\PS$ and $\PX$ have full support, and $P_{Y|X,S}(y|x,s)=\PYgX(y|x)$, which results in the Markov chain $S-X-Y$. 
        
        The main privacy measure is lift\footnote{Since we assume $\PS$ and $\PY$ have full supports, $\lifty$ is finite.}, given in \eqref{eq:lift}.
        Lift and its logarithm, log-lift, quantify multiplicative information gain on each sensitive feature $s \in \Sen$ via accessing $y \in \Y$.
        There are two cases: $\lifty>1 \Rightarrow \PSgY(s|y)>\PS(s)$ indicates the increment of the belief about $s$ after releasing $y$; $\lifty \leq 1 \Rightarrow \PSgY(s|y) \leq \PS(s)$ means releasing $y$ decreases the belief. 
        The more the posterior belief deviates from the prior belief, the more an adversary gains knowledge about $s$.
        Thus, for each $y \in \Y$, the $\max_{s}\lifty$ and $\min_{s}\lifty$ determine the highest knowledge gain of sensitive features, and they should be restricted  to protect privacy.
        We use the following notation for these quantities,
        \begin{equation}\label{def:min-max-lift}
            \Psi(y)\triangleq\min_{s \in \Sen}\lifty \quad  \text{and} \quad \Lambda(y)\triangleq \max_{s \in \Sen}\lifty.
        \end{equation}
        The lift has been applied in local information privacy \cite{2019Watchdog,2020PropertiesWatchdog,2021DataSanitize} to provide protection of sensitive features, which is defined as follows.
        \begin{Definition}\label{def:LIP}
            For  $\eps \in \RealP$, a privacy mechanism $\M:\X\rightarrow\Y$ is $\eps$-local information private  or $\eps$-LIP, with respect to $S$, if for all $y\in\Y$,
            \begin{equation}\label{eq:LIP}
                \e^{-\eps} \leq \Psi(y) \quad \text{and} \quad \Lambda(y) \leq \e^{\eps}.
            \end{equation}
        \end{Definition}
            
        Another instance-wise measure is local differential privacy \cite{2011Learnprivately,2013LDPMiniMax,2021DataSanitize},
        
        \begin{Definition}\label{def:LDP}
            For  $\eps \in \RealP$, a privacy mechanism $\mathcal{M}: \mathcal{X} \rightarrow \mathcal{Y}$ is $\eps$-local differential private  or $\eps$-{LDP}, with respect to $S$,  if for all $s, s^{\prime} \in \mathcal{S}$ and all $y\in\Y$,
            \begin{equation}\label{eq:LDP}
                \Gamma(y)=\sup_{s,s^{\prime}\in \Sen}\frac{\PYgS(y|s)}{\PYgS(y|s^{\prime})}=\frac{\Lambda(y)}{\Psi(y)} \leq \e^{\eps}.
            \end{equation}
        \end{Definition}  
        
%=================================================================
%=================================================================

\section{Asymmetric Local Information Privacy}\label{sec:ALIP}

    According to \eqref{eq:LIP}, LIP restricts the decrement of $\log\Psi(y)$ and increment of $\log\Lambda(y)$ by the symmetric bounds.
    However, we demonstrate that these metrics have a distinct range of values and probabilities of appearance in the dataset, $\PSY(s,y)$. 
    We have plotted the histogram of $\log\Psi(y)$ and $\log\Lambda(y)$ for $10^3$ randomly generated distributions in Figure~\ref{fig:density}, where $|\X|=17$ and $|\Sen|=5$. 
    In this figure, the range of $\log\Psi(y)$ is $[-12,-0.06]$, much larger than the range of $\log\Lambda(y)$, $[0.02,1.64]$.
    Moreover, the maximum probability of $\log\Psi(y)$ is much lower than the maximum probability of $\log\Lambda(y)$.
    We refer to these properties as \textit{lift asymmetry}.  
    Since high values of $|\log\Psi(y)|$ have a significantly lower probability\footnote{For example, in Figure~\ref{fig:density}, the probability of $|\log\Psi(y)| \geq 6$ is near zero.} than the $\log\Lambda(y)$, we can relax the min-lift privacy by allocating a higher privacy budget to it while applying a stricter bound for the max-lift.
    Thus, we propose asymmetric local information privacy (ALIP), where we consider different privacy budgets $\epsl$ and $\epsu$ for $|\log\Psi(y)|$ and $\log\Lambda(y)$, respectively. 

    \begin{figure}[h]
        \centering
        \scalebox{0.5}{\begin{tikzpicture}
\begin{axis}[%
width=7in,
height=4.32623792125in,
scale only axis,
xmin=-12,
xmax=2,
xlabel style={font=\color{white!15!black}, yshift=-20pt},
ymin=0,
ymax=0.08,
ylabel style={font=\color{white!15!black}},
ylabel={ \Huge  Probability density function},
axis background/.style={fill=white},
yticklabel style = {font=\huge},
xticklabel style = {font=\huge},
 clip=false,
xmajorgrids,
ymajorgrids
]
\addplot [color=blue, line width=2.0pt, forget plot]
  table[row sep=crcr]{%
-11.9793320701979	2.77777777777778e-05\\
-11.8982702154161	0\\
-11.8172083606343	0\\
-11.7361465058525	0\\
-11.6550846510707	0\\
-11.5740227962889	0\\
-11.4929609415071	0\\
-11.4118990867253	0\\
-11.3308372319435	0\\
-11.2497753771617	0\\
-11.1687135223799	0\\
-11.087651667598	0\\
-11.0065898128162	0\\
-10.9255279580344	0\\
-10.8444661032526	0\\
-10.7634042484708	0\\
-10.682342393689	0\\
-10.6012805389072	0\\
-10.5202186841254	0\\
-10.4391568293436	0\\
-10.3580949745618	0\\
-10.27703311978	0\\
-10.1959712649982	0\\
-10.1149094102163	0\\
-10.0338475554345	0\\
-9.95278570065273	0\\
-9.87172384587092	0\\
-9.79066199108911	2.77777777777778e-05\\
-9.7096001363073	0\\
-9.6285382815255	0\\
-9.54747642674369	0\\
-9.46641457196188	0\\
-9.38535271718007	0\\
-9.30429086239826	0\\
-9.22322900761645	2.77777777777778e-05\\
-9.14216715283465	0\\
-9.06110529805284	2.77777777777778e-05\\
-8.98004344327103	2.77777777777778e-05\\
-8.89898158848922	2.77777777777778e-05\\
-8.81791973370741	0\\
-8.7368578789256	2.77777777777778e-05\\
-8.6557960241438	0\\
-8.57473416936199	0\\
-8.49367231458018	0\\
-8.41261045979837	5.55555555555556e-05\\
-8.33154860501656	0\\
-8.25048675023475	0\\
-8.16942489545294	0.000111111111111111\\
-8.08836304067114	5.55555555555556e-05\\
-8.00730118588933	5.55555555555556e-05\\
-7.92623933110752	5.55555555555556e-05\\
-7.84517747632571	8.33333333333333e-05\\
-7.7641156215439	8.33333333333333e-05\\
-7.6830537667621	5.55555555555556e-05\\
-7.60199191198029	5.55555555555556e-05\\
-7.52093005719848	0\\
-7.43986820241667	5.55555555555556e-05\\
-7.35880634763486	8.33333333333333e-05\\
-7.27774449285305	2.77777777777778e-05\\
-7.19668263807125	2.77777777777778e-05\\
-7.11562078328944	2.77777777777778e-05\\
-7.03455892850763	5.55555555555556e-05\\
-6.95349707372582	0.000138888888888889\\
-6.87243521894401	0.000111111111111111\\
-6.7913733641622	0.000111111111111111\\
-6.7103115093804	8.33333333333333e-05\\
-6.62924965459859	0.000222222222222222\\
-6.54818779981678	0.000111111111111111\\
-6.46712594503497	0.000138888888888889\\
-6.38606409025316	0.000166666666666667\\
-6.30500223547136	0.000138888888888889\\
-6.22394038068955	0.000194444444444444\\
-6.14287852590774	0.00025\\
-6.06181667112593	0.000138888888888889\\
-5.98075481634412	0.000222222222222222\\
-5.89969296156231	0.000111111111111111\\
-5.81863110678051	0.000277777777777778\\
-5.7375692519987	0.000305555555555556\\
-5.65650739721689	0.000444444444444444\\
-5.57544554243508	0.000305555555555556\\
-5.49438368765327	0.000305555555555556\\
-5.41332183287146	0.000611111111111111\\
-5.33225997808966	0.000444444444444444\\
-5.25119812330785	0.000527777777777778\\
-5.17013626852604	0.000527777777777778\\
-5.08907441374423	0.0005\\
-5.00801255896242	0.000694444444444444\\
-4.92695070418061	0.000694444444444444\\
-4.84588884939881	0.000888888888888889\\
-4.764826994617	0.000972222222222222\\
-4.68376513983519	0.000944444444444444\\
-4.60270328505338	0.000888888888888889\\
-4.52164143027157	0.00122222222222222\\
-4.44057957548976	0.00125\\
-4.35951772070796	0.00105555555555556\\
-4.27845586592615	0.00127777777777778\\
-4.19739401114434	0.00158333333333333\\
-4.11633215636253	0.00144444444444444\\
-4.03527030158072	0.00186111111111111\\
-3.95420844679891	0.00183333333333333\\
-3.87314659201711	0.00186111111111111\\
-3.7920847372353	0.00241666666666667\\
-3.71102288245349	0.00261111111111111\\
-3.62996102767168	0.00241666666666667\\
-3.54889917288987	0.00261111111111111\\
-3.46783731810807	0.00283333333333333\\
-3.38677546332626	0.00294444444444444\\
-3.30571360854445	0.00338888888888889\\
-3.22465175376264	0.00363888888888889\\
-3.14358989898083	0.00436111111111111\\
-3.06252804419902	0.00422222222222222\\
-2.98146618941722	0.00452777777777778\\
-2.90040433463541	0.00494444444444444\\
-2.8193424798536	0.00519444444444444\\
-2.73828062507179	0.00513888888888889\\
-2.65721877028998	0.00602777777777778\\
-2.57615691550817	0.00658333333333333\\
-2.49509506072637	0.00666666666666667\\
-2.41403320594456	0.00855555555555556\\
-2.33297135116275	0.00830555555555556\\
-2.25190949638094	0.00866666666666667\\
-2.17084764159913	0.00997222222222222\\
-2.08978578681732	0.0104722222222222\\
-2.00872393203552	0.0114166666666667\\
-1.92766207725371	0.0118055555555556\\
-1.8466002224719	0.0119444444444444\\
-1.76553836769009	0.0133055555555556\\
-1.68447651290828	0.0144166666666667\\
-1.60341465812647	0.0153333333333333\\
-1.52235280334467	0.0168611111111111\\
-1.44129094856286	0.0169444444444444\\
-1.36022909378105	0.0181944444444444\\
-1.27916723899924	0.01925\\
-1.19810538421743	0.0208055555555556\\
-1.11704352943562	0.0203611111111111\\
-1.03598167465382	0.0194722222222222\\
-0.954919819872009	0.0200277777777778\\
-0.873857965090199	0.0204722222222222\\
-0.792796110308391	0.0206666666666667\\
-0.711734255526583	0.0183888888888889\\
-0.630672400744775	0.0172777777777778\\
-0.549610545962967	0.0151666666666667\\
-0.468548691181157	0.0123333333333333\\
-0.387486836399349	0.0103888888888889\\
-0.306424981617541	0.00672222222222222\\
-0.225363126835733	0.00377777777777778\\
-0.144301272053925	0.00169444444444444\\
-0.0632394172721167	0.000527777777777778\\
0.0178224375096931	0\\
0.0988842922915012	0.000305555555555556\\
0.179946147073309	0.00277777777777778\\
0.261008001855117	0.0113333333333333\\
0.342069856636925	0.0236666666666667\\
0.423131711418733	0.0436666666666667\\
0.504193566200541	0.0571111111111111\\
0.58525542098235	0.0708888888888889\\
0.666317275764158	0.0687777777777778\\
0.747379130545966	0.0613611111111111\\
0.828440985327774	0.0522777777777778\\
0.909502840109582	0.0378055555555556\\
0.99056469489139	0.0284166666666667\\
1.0716265496732	0.0185\\
1.15268840445501	0.011\\
1.23375025923682	0.00605555555555556\\
1.31481211401863	0.00319444444444444\\
1.39587396880043	0.00158333333333333\\
1.47693582358224	0.000638888888888889\\
1.55799767836405	0.000416666666666667\\
1.63905953314586	0.000222222222222222\\
};
\draw[decorate,decoration={brace,mirror,amplitude=10}]
([yshift=-20pt]axis cs:-12,0) --
node[below=10pt] {\Huge $\log\Psi(y)$}
([yshift=-20pt]axis cs: -0.05,0);
\draw[decorate,decoration={brace,mirror,amplitude=10}]
([yshift=-20pt]axis cs:0.05,0) --
node[below=10pt] {\Huge $\log\Lambda(y)$}
([yshift=-20pt]axis cs:2,0);
\end{axis}
\end{tikzpicture}%
}
        \caption{Histogram of $\log\Psi(y)=\min_{s}i(s,y)$ and $\log\Lambda(y)=\max_{s}i(s,y)$ for $10^3$ randomly generated distributions, where $|\X|=17$, $|\Sen|=5$.}
        \label{fig:density}
    \end{figure}
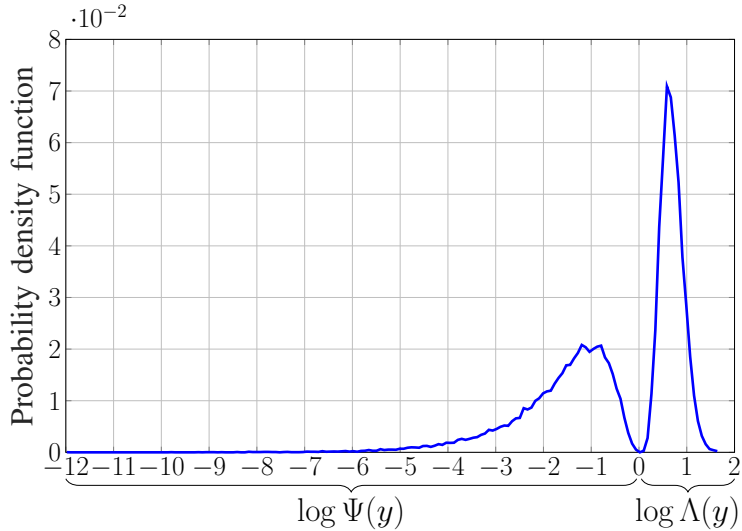
    
    This will result in the following notion of privacy, which is more compatible  with the lift asymmetry property.
    
    \begin{Definition}\label{def:ALIP}
        For $\epsl, \epsu \in \RealP$, a privacy mechanism $\M:\X\rightarrow\Y$ is $(\epsl,\epsu)$-asymmetric local information private, or $(\epsl,\epsu)$-ALIP, with respect to $S$, if for all $y \in \Y$,
        \begin{equation}\label{eq:ALIP}
            \e^{-\epsl} \leq  \Psi(y) \quad \text{and} \quad \Lambda(y) \leq \e^{\epsu}.
        \end{equation}
    \end{Definition}

    The following proposition indicates how $(\epsl,\epsu)$-ALIP restricts average privacy leakage measures and LDP.
    
    \begin{Proposition}\label{prop:ALIP-LDP-Average}
        If $(\epsl,\epsu)$-ALIP is satisfied, then,
        \begin{enumerate}
            \item \label{subprop:ALIP-MI}  $I(S;Y)\leq \epsu$,
            \item \label{subprop:ALIP-total}  $T(S;Y)\leq\frac{1}{2} (\e^{\epsu}-1)$ and $\chi^2(S;Y) \leq (\e^{\epsu}-1)^2$,
            \item \label{subprop:ALIP-Sibson}  $I_{\alpha}^{S}(S;Y) \leq \frac{\alpha}{\alpha-1}\epsu$ and $I_{\alpha}^{A}(S;Y) \leq \frac{\alpha}{\alpha-1}\epsu$,
            \item \label{subprop:ALIP-LDP}$\eps$-LDP is satisfied where $\eps=\epsl+\epsu$,
        \end{enumerate}
        where $T(S;Y)$ is the total variation distance, $\chi^2(S;Y)$ is $\chi^2$-divergence, $I_{\alpha}^{S}(S;Y)$ is Sibson MI, and  $I_{\alpha}^{A}(S;Y)$ is Arimoto MI.
    \end{Proposition}
    
    \begin{proof}
        The proof is given in Appendix \ref{apx:proof of ALIP-LDP-Average}.
    \end{proof}
    Propositions \ref{prop:ALIP-LDP-Average}-\ref{subprop:ALIP-MI} to \ref{prop:ALIP-LDP-Average}-\ref{subprop:ALIP-Sibson} demonstrate that average measures are bounded with the max-lift privacy budget. 
    In subsection \ref{subsec: ALIP numericall}, we show that ALIP can enhance utility via relaxing min-lift  $\epsl >\epsu$, where a smaller upper bound is allocated to the max-lift and average measures in Proposition \ref{prop:ALIP-LDP-Average}.
    Proposition \ref{prop:ALIP-LDP-Average}-\ref{subprop:ALIP-LDP} shows the relationship between $(\epsl,\epsu)$-ALIP and $\eps$-LDP.
    We introduce a variable $\lambda \in (0,1)$ to have a convenient representation of this relationship as follows:
    for an LDP privacy budget $\eps$, if we set $\epsl=\lambda\eps$ and $\epsu=(1-\lambda)\eps$, we have $\epsl+\epsu=\eps$.
    Thus, varying $\lambda$ gives rise to different $(\epsl,\epsu)$-ALIP scenarios within the same budget for $\eps$-LDP.
    If $\lambda<0.5$, we have relaxation on the max-lift privacy; if $\lambda > 0.5$, it implies relaxation on the min-lift privacy.
    When $\lambda=0.5$, we have the \textit{symmetric} case of $\frac{\eps}{2}$-LIP, where $\epsl=\epsu=\frac{\eps}{2}$.

%=================================================================

    \subsection{ALIP Privacy-Utility Tradeoff}\label{subsec: ALIP numericall}

        In this subsection, we propose a watchdog mechanism based on ALIP and LDP and an asymmetric ORR (AORR) mechanism for ALIP to perturb data and achieve privacy protection.
        We observe the PUT of ALIP and LDP, where the utility is measured by MI between $X$ and $Y$, $I(X;Y)$.
        %
%-----------------------------------------------------------------

        \subsubsection{Watchdog Mechanism} 
            Watchdog privacy mechanism bi-partitions $\X$ into low-risk and high-risk subsets denoted by $\XL$ and $\XH$, respectively, and only randomizes high-risk symbols.
            In the existing LIP, $\XL$ and $\XH$ are determined by symmetric bounds. We propose to use ALIP to obtain $\XL$ and $\XH$:
            \begin{gather}
                \XL \triangleq \{x\in \X: \e^{-\epsl} \leq \Psi(x) \quad and \quad   \Lambda(x) \leq \e^{\epsu} \} \quad and \quad
                \X_H=\X \setminus \XL \label{eq:ALIP XH}.
            \end{gather}
            For LDP, $\XL$ and $\XH$ are given by
            \begin{gather}
                \XL \triangleq \{x\in \X: \Gamma(x) \leq \e^{\eps} \} \ \quad and  \quad
                \X_H=\X \setminus \XL \label{eq:LDP XL,XH}.
            \end{gather}
            After obtaining $\XL$ and $\XH$, the privacy mechanism will be
            \begin{equation} \label{eq: watchdog mechanism}
                \mathcal{M} = \begin{cases}
                    \mathbf{1}_{\{x=y\}}, & x,y \in \XL=\YL,\\
                    r(y|x),		 			 & x \in \XH, y \in \YH,\\
                    0,		 			 & \textup{otherwise},
                \end{cases}
            \end{equation}
            where $\mathbf{1}_{\{x=y\}}$ indicates the publication of low-risk symbols without alteration, and  $r(y|x)$ is the randomization on high-risk symbols, where $\sum_{y\in\YH}r(y|x)=1$.
        
            An instance of $r(y|x)$ is the $X$-invariant randomization, if $r(y|x)=\R(y)$ for $x \in  \XH, y \in \YH$, and $\sum_{ y\in \YH}\R(y)=1$.
            An example of $\R(y)$ is the uniform randomization $\R(y)=\frac{1}{|\YH|}$ with the special case of \textit{complete merging}, where $|\YH|=1$, and all $x \in \XH$ are mapped to one super symbol $y^* \in \YH$.
            It has been proved in \cite{2020PropertiesWatchdog} for LIP that $X$-invariant randomization minimizes privacy leakage in $\XH$.
            Accordingly, if we apply ALIP in the watchdog mechanism, for $\XH \neq \emptyset$, the minimum leakages over $\XH$ are
            \begin{gather}
                \overline{\eps}_{u}:=\max_{s \in \Sen}i(s,\XH)=\max_{s \in \Sen}\log l(s,\XH)=\max_{s \in \Sen}\log\frac{P(\XH|s)}{P(\XH)}, \label{eq:max LIP leak}\\
                \overline{\eps}_{l}:=\left|\min_{s \in \Sen}i(s,\XH)\right|=\left|\min_{s \in \Sen}\log l(s,\XH)\right|= \left|\min_{s \in \Sen }\log\frac{P(\XH|s)}{P(\XH)}\right|, \label{eq:min LIP leak}
            \end{gather}
            where $\displaystyle P(\XH|s)=\sum_{x \in \XH}\PXgS(x|s)$ and $\displaystyle P(\XH)=\sum_{x \in \XH}\PX(x).$
           
            $X$-invariant randomization is also applicable for LDP and the following theorem shows that it minimizes LDP privacy leakage in $\XH$.
            \begin{Theorem}\label{prop:minimum leakage LDP}
                In the LDP watchdog mechanism where $\XL$ and $\XH$ are determined according to \eqref{eq:LDP XL,XH}, $X$-invariant randomization minimizes privacy leakage in $\XH$ measured by $\Gamma(y)$ in \eqref{eq:LDP}.
            \end{Theorem}
            \begin{proof}
                The proof is given in Appendix \ref{apx:proof of minimum leakage LDP}.
            \end{proof}
            
            In the watchdog mechanism with $X$-invariant randomization, the resulting utility measured by MI between $X$ and $Y$ is given by
            \begin{equation}\label{eq:MI watch}
                I(X;Y)=H(X)-\sum_{x \in \XH}\PX(x) \log \frac{P(\XH)}{\PX(x)}.
            \end{equation}
            In \cite{2020PropertiesWatchdog}, it has been verified that $I(X;Y)$ in \eqref{eq:MI watch} is monotonic in $\XH$: if $\XH^{'} \subset \XH$ then $I(X;Y) < I^{'}(X;Y)$, where $I^{'}(X;Y)$ is the resulting utility of $\XH^{'}$.

            \begin{Proposition}\label{prop:LDPutility}
                In the watchdog mechanism with $X$-invariant randomization, for a given LDP privacy budget $\eps$, $\lambda \in (0,1)$, and ALIP privacy budgets $\epsl=\lambda\eps, \epsu=(1-\lambda)\eps$, LDP results in higher utility than ALIP.
                \begin{proof}
                    Denote the high-risk subset for LDP by $\XH^{'}$ and for ALIP by $\XH$. Based on the remark following \eqref{eq:MI watch}, it is enough to prove that $\XH' \subseteq \XH$. We have
                    $$\XH^{'}=\{x \in \X: \frac{\Lambda(x)}{\Psi(x)} > \eps \} \quad \text{and} \quad \XH=\{ x \in \X: \Lambda(x) > \e^{(1-\lambda) \eps} \enspace or  \enspace \Psi(x) < \e^{-\lambda \eps}\}.$$
                    If $x \in \XH^{'}$, then either $\Lambda(x) > \e^{(1-\lambda)\eps}$ or  $\Lambda(x) \leq \e^{(1-\lambda)\eps}$. If the former holds, $ x \in \XH $. 
                    If the latter holds, since $\frac{\Lambda(x)}{\Psi(x)} > \e^{\eps}$, we get $ \displaystyle \Lambda(x) > \e^{(1-\lambda)\eps} \e^{\lambda \eps} \Psi(x)$. 
                    Because all quantities are positive, $\e^{\lambda \eps} \Psi(x) <1$; thus, $\Psi(x) < \e^{-\lambda\eps} \Rightarrow x \in \XH$. 
                \end{proof}
            \end{Proposition}

            Watchdog mechanism with $X$-invariant randomization is a powerful method with low complexity that can be easily applied to instance-wise measures. 
            However, it significantly degrades the utility \cite{2020PropertiesWatchdog} because $X$-invariant randomization obfuscates all high-risk symbols together to minimize privacy leakage, with the cost of deteriorating data resolution.
            In section \ref{sec:subset rand},  we propose subset merging randomization to enhance the utility of the watchdog mechanism.
            
 %-----------------------------------------------------------------       
 
        \subsubsection{Asymmetric Optimal Random Response (AORR)}
        
            ORR has been proposed in \cite{2021DataSanitize} as a localized instance-wise replacement of the privacy funnel \cite{2014PFInfBottneck}.
            It is the solution to the optimal utility problem subject to $\eps$-LIP or $\eps$-LDP constraints.
            For ALIP, we propose asymmetric optimal random response (AORR), which is defined as 
            \begin{align} \label{eq:AORR}
                &\max_{\substack{\PXgY,\PY}}I(X;Y)\\
                \nonumber \textrm{s.t.} \quad & S-X-Y\\
                \nonumber&\e^{-\epsl} \leq \Psi(y) \quad \text{and} \quad \Lambda(y)\leq \nonumber \e^{\epsu}, \hspace{5pt} \forall y \in 	\Y. 
            \end{align}
            Privacy constraints in this optimization problem form a closed, bounded, convex polytope \cite{2021DataSanitize}. 
            It has been proved that vertices of this polytope are the feasible candidates that maximize MI and satisfy privacy constraints \cite{ 2021DataSanitize, 2019DataDsclsurPtfPriv, 2014ExtermalMechanism}.
            However, the number of vertices grows exponentially in the dimension of the polyhedron, which is $|\X|(|\X|-1)$ for LDP and $\left(|\X|-1\right)$ for LIP. 
            This makes ORR computationally cumbersome for large $|\X|$. 
            Accordingly, \cite{2021DataSanitize} suggests some approaches with lower complexity than ORR to avoid vertex enumeration for the larger sizes of $\X$, but this comes at the cost of lower utility.

%-----------------------------------------------------------------   
        \subsubsection{Numerical Results}\label{subsubsec:Num watchdog}

            Here we demonstrate the privacy leakage and utility of AORR and the watchdog mechanism under ALIP. 
            For the utility, we use normalized MI (NMI) $$\text{NMI}=\frac{I(X;Y)}{H(X)} \in [0,1].$$ 
            It is clear that the maximum possible utility is obtained when $X$ is published without randomization, where $Y=X$ and $I(X;Y)=H(X)$. 
            Thus, $I(X;Y) \leq H(X)$ and NMI $\leq 1$.
            For the numerical simulations, we have randomly generated $10^{3}$ distributions for the watchdog mechanism and $100$ distributions for the AORR where $|\X|=17$ and $|\Sen|=5$. 
            Then, the mean values of NMI, $\log\left(\max_{y}\Lambda(y)\right)$ (max-lift leakage), and $ \left|\log\left(\min_{y}\Psi(y)\right)\right|$ (min-lift leakage) are shown versus the LDP privacy budget $\eps$, for the watchdog and AORR mechanisms in Figures \ref{Fig:Watchdog} and \ref{Fig:ORR}, respectively.
            For both mechanisms, $\eps$ varies from  $0.25$ to $8$ and three cases have been considered for $(\epsu,\epsl)$-ALIP, where $\lambda\in\{0.35,0.5,0.65\}$, $\epsl=\lambda\eps$, and $\epsu=(1-\lambda)\eps$. 

            In Figure~\ref{Fig:Watchdog}, we observe that in the watchdog mechanism, LDP provides higher utility and leakage than ALIP for all values of $\eps$ and $\lambda$, which confirms proposition \ref{prop:LDPutility}.
            Figures~\ref{fig:WatchdogUt} and \ref{fig:ORRutility} demonstrate that the min-lift relaxation, $\lambda=0.65$, enhances utility in the watchdog and AORR mechanisms for $\eps>1$. Note that in both figures, $\lambda=0.5$ refers to $\frac{\eps}{2}$-LIP. 
            On the other hand, $\lambda=0.35$ results in lower utility. 
            Generally, any value of $\lambda < 0.5$ reduces utility since it strictly bounds the min-lift while relaxing the max-lift.  
            As the min-lift has a wider range of values, achieving this strict bound enlarges the set $\XH$ and requires randomizing more symbols, which reduces utility.

            Another observation here is that AORR incurs significantly higher utility than the watchdog mechanism.  
            For instance, when $\lambda=0.5$ and $\eps=2$, the watchdog mechanism results in a utility of $0.52$, while AORR has a utility of $0.94$. 
            AORR finds the optimal utility, which due to PUT, necessarily results in the highest leakage subject to privacy constraints. 
            However, the watchdog mechanism is a non-optimal solution that minimizes leakage of high-risk symbols to provide strong privacy protection, which deteriorates utility.
            To solve this drawback of the watchdog mechanism, we propose a subset randomization method in the following section.

            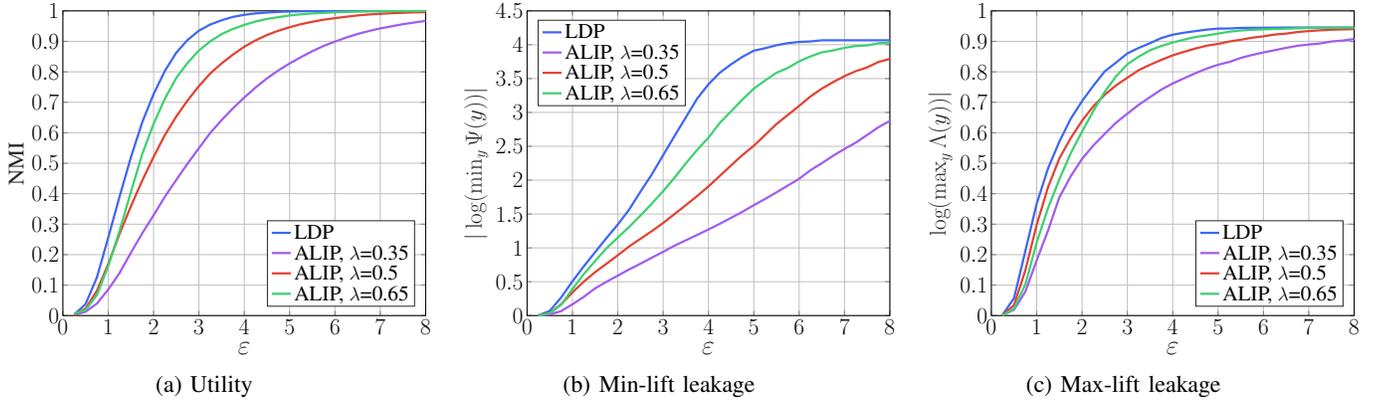
\begin{figure}
            \centering
                \begin{subfigure}{0.32\textwidth}
                \centering
                    \scalebox{0.42}{\definecolor{mycolor1}{rgb}{0.25, 0.4, 0.96}
\definecolor{mycolor2}{rgb}{0.63, 0.36, 0.94}
\definecolor{mycolor3}{rgb}{0.89, 0.26, 0.2}
\definecolor{mycolor4}{rgb}{0.24, 0.82, 0.44}

\begin{tikzpicture}
\begin{axis}[%
width=4.521in,
height=3.8in,
scale only axis,
xmin=0,
xmax=8,
xlabel style={font=\color{white!15!black}},
xlabel={\Huge $\eps$},
ymin=0,
ymax=1,
ylabel style={font=\color{white!15!black}},
ylabel={\huge NMI},
yticklabel style = {font=\huge},
xticklabel style = {font=\huge},
axis background/.style={fill=white},
xmajorgrids,
ymajorgrids,
legend style={at={(0.677,0.142)}, anchor=south west, legend cell align=left, align=left, draw=white!15!black, legend pos =south east,font=\fontsize{19}{19}\selectfont}
]
\addplot [color=mycolor1, line width=2pt]
  table[row sep=crcr]{%
0.25	0.00363595957333703\\
0.5	0.0374216802267113\\
0.75	0.125260516314092\\
1	0.255750255549129\\
1.25	0.390198520057209\\
1.5	0.519360515833202\\
1.75	0.63446470020284\\
2	0.726787656841567\\
2.25	0.803115874275826\\
2.5	0.862536358270269\\
2.75	0.904721569099746\\
3	0.935008373837039\\
3.25	0.95518938761849\\
3.5	0.969160109609324\\
3.75	0.980188464116254\\
4	0.98709413105778\\
4.25	0.991802104637987\\
4.5	0.994512384251872\\
4.75	0.996583895723379\\
5	0.998069602906152\\
5.25	0.998718008551077\\
5.5	0.999294258032745\\
5.75	0.999629893683911\\
6	0.999835365038007\\
6.25	0.999881764156754\\
6.5	0.999967885890903\\
6.75	0.999967885890903\\
7	0.999967885890903\\
7.25	0.999967885890903\\
7.5	0.999967885890903\\
7.75	0.999967885890903\\
8	0.999967885890903\\
};
\addlegendentry{LDP}

\addplot [color=mycolor2, line width=2pt]
  table[row sep=crcr]{%
0.25	0.00167990837919735\\
0.5	0.0124648803833893\\
0.75	0.0393621802388327\\
1	0.0854023038860948\\
1.25	0.139843094911915\\
1.5	0.206163488274976\\
1.75	0.269886248131067\\
2	0.329270556945718\\
2.25	0.390184017287841\\
2.5	0.445225738267401\\
2.75	0.498738210921823\\
3	0.549401065874796\\
3.25	0.59931741680842\\
3.5	0.641527014095103\\
3.75	0.680236696332396\\
4	0.7156139142274\\
4.25	0.748534783906187\\
4.5	0.778542460501192\\
4.75	0.804330427562624\\
5	0.827508023100588\\
5.25	0.848117199523057\\
5.5	0.867524304440143\\
5.75	0.885050024926977\\
6	0.899184320139297\\
6.25	0.912547533522322\\
6.5	0.924274335235935\\
6.75	0.933971937333256\\
7	0.942410796773631\\
7.25	0.949324593641984\\
7.5	0.956339808635771\\
7.75	0.961887098473526\\
8	0.967041163848539\\
};
\addlegendentry{$\text{ALIP, }\lambda\text{=0.35}$}

\addplot [color=mycolor3, line width=2pt]
  table[row sep=crcr]{%
0.25	0.00232803882362242\\
0.5	0.0228445524018419\\
0.75	0.0812568151430994\\
1	0.169120224550059\\
1.25	0.268307457355739\\
1.5	0.358593287128413\\
1.75	0.442917517264832\\
2	0.520642668288906\\
2.25	0.592641196205267\\
2.5	0.653266421118533\\
2.75	0.706044673963869\\
3	0.752968637445294\\
3.25	0.794591448052248\\
3.5	0.827508023100588\\
3.75	0.856432702256904\\
4	0.882394294136526\\
4.25	0.90316158336014\\
4.5	0.920764410121858\\
4.75	0.934848575307635\\
5	0.946264255415453\\
5.25	0.956339808635771\\
5.5	0.964193284647758\\
5.75	0.971014766681167\\
6	0.976069457000496\\
6.25	0.980810546838991\\
6.5	0.985090687840927\\
6.75	0.988128062664786\\
7	0.990479292002468\\
7.25	0.992303624082308\\
7.5	0.993781486097376\\
7.75	0.995167991759646\\
8	0.995810546316464\\
};
\addlegendentry{$\text{ALIP, }\lambda\text{=0.5}$}

\addplot [color=mycolor4, line width=2pt]
  table[row sep=crcr]{%
0.25	0.0064275173424969\\
0.5	0.0182855108827373\\
0.75	0.067311251522892\\
1	0.162248654051471\\
1.25	0.278333574807807\\
1.5	0.405568787267825\\
1.75	0.52999770487513\\
2	0.629128664355549\\
2.25	0.711887169492945\\
2.5	0.779754624789156\\
2.75	0.829508240636696\\
3	0.869928322507786\\
3.25	0.900022081301353\\
3.5	0.923778388767914\\
3.75	0.940689769634959\\
4	0.954043644382012\\
4.25	0.965060291512544\\
4.5	0.973575168460946\\
4.75	0.979301218045013\\
5	0.985090687840927\\
5.25	0.98875793706631\\
5.5	0.991580383819888\\
5.75	0.993511251606896\\
6	0.995292536108013\\
6.25	0.99645383731084\\
6.5	0.997470505083211\\
6.75	0.998141433432385\\
7	0.998755828147505\\
7.25	0.999118291077067\\
7.5	0.999225186917732\\
7.75	0.999576653858199\\
8	0.999729957430233\\
};
\addlegendentry{$\text{ALIP, }\lambda\text{=0.65}$}

\end{axis}
\end{tikzpicture}%
}
                    \caption{ \centering Utility \label{fig:WatchdogUt} }
                \end{subfigure}%
                ~
                \begin{subfigure}{0.32\textwidth}
                \centering
                    \scalebox{0.42}{\definecolor{mycolor1}{rgb}{0.25, 0.4, 0.96}
\definecolor{mycolor2}{rgb}{0.63, 0.36, 0.94}
\definecolor{mycolor3}{rgb}{0.89, 0.26, 0.2}
\definecolor{mycolor4}{rgb}{0.24, 0.82, 0.44}
\begin{tikzpicture}
\begin{axis}[%
width=4.521in,
height=3.8in,
scale only axis,
xmin=0,
xmax=8,
xlabel style={font=\color{white!15!black}},
xlabel={\Huge $\eps$},
ymin=0,
ymax=4.5,
ylabel style={font=\color{white!15!black}},
ylabel={\huge $|\log (\min_{y}\Psi(y))|$},
yticklabel style = {font=\huge},
xticklabel style = {font=\huge},
axis background/.style={fill=white},
xmajorgrids,
ymajorgrids,
legend style={legend cell align=left, align=left, draw=white!15!black, legend pos =north west, font=\fontsize{19}{19}\selectfont}
]
\addplot [color=mycolor1, line width=2.0pt]
  table[row sep=crcr]{%
0.25	0.00323550979259348\\
0.5	0.0681930678803089\\
0.75	0.266115053122653\\
1	0.505324515834168\\
1.25	0.726455248241475\\
1.5	0.936581828747015\\
1.75	1.14389488809416\\
2	1.34455599130745\\
2.25	1.5648906111244\\
2.5	1.82983994097763\\
2.75	2.08411154292122\\
3	2.36862207375918\\
3.25	2.65042198994283\\
3.5	2.92709414884952\\
3.75	3.20344922351204\\
4	3.41592457595669\\
4.25	3.58627886661454\\
4.5	3.71136773224988\\
4.75	3.82185217242897\\
5	3.91206040155766\\
5.25	3.94853323109333\\
5.5	3.99222444687754\\
5.75	4.02361040565654\\
6	4.04120206801038\\
6.25	4.04805255868705\\
6.5	4.06549879673044\\
6.75	4.06549879673044\\
7	4.06549879673044\\
7.25	4.06549879673044\\
7.5	4.06549879673044\\
7.75	4.06549879673044\\
8	4.06549879673044\\
};
\addlegendentry{LDP}

\addplot [color=mycolor2, line width=2.0pt]
  table[row sep=crcr]{%
0.25	0.00106849948736802\\
0.5	0.0158223556043248\\
0.75	0.0661409785421209\\
1	0.166518725775164\\
1.25	0.275762673849929\\
1.5	0.401851020504521\\
1.75	0.49834587270133\\
2	0.589676153225905\\
2.25	0.681212938389134\\
2.5	0.765728074298419\\
2.75	0.8549020415632\\
3	0.941127796284972\\
3.25	1.03193626921071\\
3.5	1.11147553528909\\
3.75	1.19059127135669\\
4	1.27164262346619\\
4.25	1.3558186739437\\
4.5	1.44385469149227\\
4.75	1.53278521490595\\
5	1.63031577410879\\
5.25	1.72253805364305\\
5.5	1.82035992519512\\
5.75	1.92010634711371\\
6	2.01987888437318\\
6.25	2.14748960690808\\
6.5	2.2496298614587\\
6.75	2.36131389386655\\
7	2.46074622032645\\
7.25	2.55565561409543\\
7.5	2.6673413443547\\
7.75	2.78302469928579\\
8	2.87500819219921\\
};
\addlegendentry{$\text{ALIP, }\lambda\text{=0.35}$}

\addplot [color=mycolor3, line width=2.0pt]
  table[row sep=crcr]{%
0.25	0.00177601476047174\\
0.5	0.0363716417748061\\
0.75	0.165147459612918\\
1	0.341649236749539\\
1.25	0.500577765399419\\
1.5	0.641133510029513\\
1.75	0.764995949427829\\
2	0.890507780282611\\
2.25	1.02102131459407\\
2.5	1.13339194648589\\
2.75	1.24731658933326\\
3	1.36657430269787\\
3.25	1.49840189644108\\
3.5	1.63031577410879\\
3.75	1.7647222455738\\
4	1.90775759597298\\
4.25	2.06268287187451\\
4.5	2.21867823853946\\
4.75	2.37170296914929\\
5	2.50800388772173\\
5.25	2.6673413443547\\
5.5	2.82921948101417\\
5.75	2.96308648574992\\
6	3.09683070241512\\
6.25	3.23645130633221\\
6.5	3.35162029095724\\
6.75	3.44856179210176\\
7	3.53591352257902\\
7.25	3.6098285046373\\
7.5	3.66534879996823\\
7.75	3.74568882083729\\
8	3.79055993857415\\
};
\addlegendentry{$\text{ALIP, }\lambda\text{=0.5}$}

\addplot [color=mycolor4, line width=2.0pt]
  table[row sep=crcr]{%
0.25	0.0014846703596063\\
0.5	0.0342842600649347\\
0.75	0.159791507833109\\
1	0.391589945092072\\
1.25	0.612043694050522\\
1.5	0.805086000068862\\
1.75	0.993752799197567\\
2	1.15176444972747\\
2.25	1.30972989939129\\
2.5	1.47680764250511\\
2.75	1.65316399375012\\
3	1.83745908491449\\
3.25	2.03397970868448\\
3.5	2.24748213257548\\
3.75	2.44750489898688\\
4	2.62769032327201\\
4.25	2.84573315026876\\
4.5	3.01685156172427\\
4.75	3.19136131982421\\
5	3.35162029095724\\
5.25	3.47184394544959\\
5.5	3.58591006663995\\
5.75	3.6569358153656\\
6	3.75501041494208\\
6.25	3.82766090462789\\
6.5	3.88733001701286\\
6.75	3.91586501424385\\
7	3.95436117072457\\
7.25	3.98198141622371\\
7.5	3.98931602003396\\
7.75	4.0219644363764\\
8	4.03133542964125\\
};
\addlegendentry{$\text{ALIP, }\lambda\text{=0.65}$}

\end{axis}
\end{tikzpicture}%
}
                    \caption{ \centering Min-lift leakage \label{fig:WatchdogL} }
                \end{subfigure}
                ~
                \begin{subfigure}{0.32\textwidth}
                \centering
                    \scalebox{0.42}{\definecolor{mycolor1}{rgb}{0.25, 0.4, 0.96}
\definecolor{mycolor2}{rgb}{0.63, 0.36, 0.94}
\definecolor{mycolor3}{rgb}{0.89, 0.26, 0.2}
\definecolor{mycolor4}{rgb}{0.24, 0.82, 0.44}
\begin{tikzpicture}

\begin{axis}[%
width=4.521in,
height=3.8in,
scale only axis,
xmin=0,
xmax=8,
xlabel style={font=\color{white!15!black}},
xlabel={\Huge $\eps$},
ymin=0,
ymax=1,
ylabel style={font=\color{white!15!black}},
ylabel={\huge $\log (\max_{y}\Lambda(y))|$},
yticklabel style = {font=\huge},
xticklabel style = {font=\huge},
axis background/.style={fill=white},
xmajorgrids,
ymajorgrids,
legend style={legend cell align=left, align=left, draw=white!15!black, legend pos =south east, font=\fontsize{19}{19}\selectfont}
]
\addplot [color=mycolor1, line width=2.0pt]
  table[row sep=crcr]{%
0.25	0.00315874809439262\\
0.5	0.057219924634615\\
0.75	0.212199312806219\\
1	0.367161035234292\\
1.25	0.483102092047976\\
1.5	0.570931402343619\\
1.75	0.645995361202699\\
2	0.703726500710685\\
2.25	0.754446865634714\\
2.5	0.801459689801233\\
2.75	0.8304260144019\\
3	0.860444658462739\\
3.25	0.878520601576536\\
3.5	0.893898425615071\\
3.75	0.910927185969493\\
4	0.921842295580934\\
4.25	0.92836628054703\\
4.5	0.933603764487825\\
4.75	0.93855390599651\\
5	0.941726018972076\\
5.25	0.942265933321042\\
5.5	0.94398703640604\\
5.75	0.94449771685048\\
6	0.944619873297232\\
6.25	0.944845320154716\\
6.5	0.945261746121502\\
6.75	0.945261746121502\\
7	0.945261746121502\\
7.25	0.945261746121502\\
7.5	0.945261746121502\\
7.75	0.945261746121502\\
8	0.945261746121502\\
};
\addlegendentry{LDP}

\addplot [color=mycolor2, line width=2.0pt]
  table[row sep=crcr]{%
0.25	0.00140923595734939\\
0.5	0.0192709134749189\\
0.75	0.0784682183938534\\
1	0.181497152284372\\
1.25	0.279167096513362\\
1.5	0.387988388153007\\
1.75	0.457107763122485\\
2	0.514180243282065\\
2.25	0.558411008245316\\
2.5	0.596289602914078\\
2.75	0.633138384372225\\
3	0.662749669276134\\
3.25	0.691742873615686\\
3.5	0.717102480471976\\
3.75	0.741793036255859\\
4	0.761688933673531\\
4.25	0.778754469202979\\
4.5	0.795416415444968\\
4.75	0.810589471240357\\
5	0.823194505680589\\
5.25	0.832441588207062\\
5.5	0.845481243966707\\
5.75	0.855325269616205\\
6	0.863489831115207\\
6.25	0.871015563229702\\
6.5	0.878662462513195\\
6.75	0.885437918109565\\
7	0.889596743365134\\
7.25	0.893060812682589\\
7.5	0.899474637657631\\
7.75	0.902645365522888\\
8	0.907827988287778\\
};
\addlegendentry{$\text{ALIP, }\lambda\text{=0.35}$}

\addplot [color=mycolor3, line width=2.0pt]
  table[row sep=crcr]{%
0.25	0.00172963950793934\\
0.5	0.0342255737384146\\
0.75	0.148822412345277\\
1	0.298411276134526\\
1.25	0.421188276013577\\
1.5	0.51571666417901\\
1.75	0.582454095144018\\
2	0.639339727720258\\
2.25	0.687036386268709\\
2.5	0.724534064696039\\
2.75	0.755638648617589\\
3	0.78061178656532\\
3.25	0.805424076314891\\
3.5	0.823194505680589\\
3.75	0.839331516789187\\
4	0.854042627749054\\
4.25	0.865664276150782\\
4.5	0.876417980375835\\
4.75	0.885824178049922\\
5	0.891567015589498\\
5.25	0.899474637657631\\
5.5	0.905323419609901\\
5.75	0.911947674401806\\
6	0.916371397795687\\
6.25	0.922409151943704\\
6.5	0.925391958791599\\
6.75	0.93058516303682\\
7	0.933702956675614\\
7.25	0.93619938175658\\
7.5	0.938154701021764\\
7.75	0.939152298722932\\
8	0.940039866809569\\
};
\addlegendentry{$\text{ALIP, }\lambda\text{=0.5}$}

\addplot [color=mycolor4, line width=2.0pt]
  table[row sep=crcr]{%
0.25	0.00104607415309945\\
0.5	0.0225638070521542\\
0.75	0.102376453485999\\
1	0.240237688270723\\
1.25	0.352619249523886\\
1.5	0.447516009657103\\
1.75	0.532527171252739\\
2	0.603761547633567\\
2.25	0.671553199526736\\
2.5	0.734514417461693\\
2.75	0.785207956559374\\
3	0.824979631670626\\
3.25	0.849872601766144\\
3.5	0.870792028313042\\
3.75	0.884891901953419\\
4	0.896217081208552\\
4.25	0.905710582817592\\
4.5	0.914173759186952\\
4.75	0.919315718580053\\
5	0.925391958791599\\
5.25	0.931694815821536\\
5.5	0.935552507402394\\
5.75	0.937905968234928\\
6	0.939152298722932\\
6.25	0.940667120495143\\
6.5	0.941596344778257\\
6.75	0.942659438549848\\
7	0.943978500926255\\
7.25	0.944134881750002\\
7.5	0.944347695108482\\
7.75	0.944834304405323\\
8	0.944845320154716\\
};
\addlegendentry{$\text{ALIP, }\lambda\text{=0.65}$}

\end{axis}
\end{tikzpicture}%
}
                    \caption{\centering Max-lift leakage \label{fig:WatchdogU} }
                \end{subfigure}
                \caption{Privacy-utility tradeoff of the watchdog mechanism with complete merging randomization under $\eps$-LDP and $(\epsl,\epsu)$-ALIP, where $|\X|=17$, $|\Sen|=5$, $\eps_{\text{LDP}} \in\{0.25,0.5,0.75,\cdots,8\}$, $\lambda \in \{0.35,0.5,0.65\}$, $\epsl =\lambda\eps$, and $\epsu=(1-\lambda)\eps$. \label{Fig:Watchdog} }
            \end{figure}
            
            \begin{figure}
            \centering
                \begin{subfigure}{0.32\textwidth}
                \centering
                    \scalebox{0.42}{\definecolor{mycolor1}{rgb}{0.25, 0.4, 0.96}
\definecolor{mycolor2}{rgb}{0.63, 0.36, 0.94}
\definecolor{mycolor3}{rgb}{0.89, 0.26, 0.2}
\definecolor{mycolor4}{rgb}{0.24, 0.82, 0.44}
\begin{tikzpicture}

\begin{axis}[%
width=4.521in,
height=3.8in,
scale only axis,
xmin=0,
xmax=8,
xlabel style={font=\color{white!15!black}},
xlabel={\Huge $\eps$},
ymin=0.65,
ymax=1,
ylabel style={font=\color{white!15!black}},
ylabel={\huge NMI},
axis background/.style={fill=white},
xmajorgrids,
ymajorgrids,
yticklabel style = {font=\huge},
xticklabel style = {font=\huge},
legend style={at={(0.679,0.153)}, anchor=south west, legend cell align=left, align=left, draw=white!15!black, legend pos =south east, font=\huge}
]
\addplot [color=mycolor2, line width=2.0pt]
  table[row sep=crcr]{%
0.25	0.685720205397846\\
0.5	0.750619076746737\\
0.75	0.790967309187876\\
1	0.820376853273993\\
1.25	0.844032749565848\\
1.5	0.864305521637844\\
1.75	0.882126007241052\\
2	0.898003179444313\\
2.25	0.91202513209668\\
2.5	0.924307088101658\\
2.75	0.935074031882208\\
3	0.944290186237739\\
3.25	0.952334149587033\\
3.5	0.959218564628346\\
3.75	0.965214267509019\\
4	0.970316069983736\\
4.25	0.974662145854739\\
4.5	0.978344640327457\\
4.75	0.981473059970284\\
5	0.984161451125588\\
5.25	0.986471342783902\\
5.5	0.988415187237861\\
5.75	0.99005407476271\\
6	0.991483578365954\\
6.25	0.992712391411955\\
6.5	0.99373990678028\\
6.75	0.994627792389701\\
7	0.995388819730065\\
7.25	0.996059692089243\\
7.5	0.996657006258592\\
7.75	0.997190021548083\\
8	0.997632283407297\\
};
\addlegendentry{AORR, $\lambda\text{=0.35}$}

\addplot [color=mycolor3, line width=2.0pt]
  table[row sep=crcr]{%
0.25	0.694778480553461\\
0.5	0.766403428804463\\
0.75	0.81564258282322\\
1	0.852090797123226\\
1.25	0.880812555696133\\
1.5	0.904487162278662\\
1.75	0.923728978742452\\
2	0.938953455653323\\
2.25	0.95123421250991\\
2.5	0.96101124684881\\
2.75	0.968948599600426\\
3	0.975231100112655\\
3.25	0.980177287028981\\
3.5	0.984161451125588\\
3.75	0.987349909934432\\
4	0.989833491305215\\
4.25	0.9918584000922\\
4.5	0.993464012621917\\
4.75	0.994744239305282\\
5	0.995784723144322\\
5.25	0.99665700625861\\
5.5	0.997386764633659\\
5.75	0.997951411109052\\
6	0.998392072093433\\
6.25	0.998737037004516\\
6.5	0.998998204550125\\
6.75	0.999207075062062\\
7	0.999372038296647\\
7.25	0.999496597868686\\
7.5	0.999588334968171\\
7.75	0.999663489907824\\
8	0.999728287225898\\
};
\addlegendentry{AORR, $\lambda\text{=0.5}$}

\addplot [color=mycolor4, line width=2.0pt]
  table[row sep=crcr]{%
0.25	0.68512818481191\\
0.5	0.755983417660997\\
0.75	0.807479064768845\\
1	0.848256525387804\\
1.25	0.883495002884657\\
1.5	0.913199517014649\\
1.75	0.936894525777225\\
2	0.954506668832154\\
2.25	0.967335886765504\\
2.5	0.976901762035664\\
2.75	0.98362775574481\\
3	0.988252326494806\\
3.25	0.991438148459975\\
3.5	0.993721577435834\\
3.75	0.995286306450814\\
4	0.996490037498342\\
4.25	0.997449432458055\\
4.5	0.998142950479267\\
4.75	0.998643421885847\\
5	0.998998204550126\\
5.25	0.999261750302152\\
5.5	0.999452826905907\\
5.75	0.9995800890168\\
6	0.999677009473788\\
6.25	0.999756789331191\\
6.5	0.999819912333908\\
6.75	0.999875026254469\\
7	0.999911294710614\\
7.25	0.999937155348675\\
7.5	0.999954590436779\\
7.75	0.999966964265398\\
8	0.999975595156083\\
};
\addlegendentry{AORR, $\lambda\text{=0.65}$}

\end{axis}

\end{tikzpicture}%
}
                    \caption{ \centering Utility \label{fig:ORRutility}}
                \end{subfigure}%
                ~
                \begin{subfigure}{0.32\textwidth}
                \centering
                    \scalebox{0.42}{\definecolor{mycolor1}{rgb}{0.25, 0.4, 0.96}
\definecolor{mycolor2}{rgb}{0.63, 0.36, 0.94}
\definecolor{mycolor3}{rgb}{0.89, 0.26, 0.2}
\definecolor{mycolor4}{rgb}{0.24, 0.82, 0.44}

\begin{tikzpicture}

\begin{axis}[%
width=4.521in,
height=3.8in,
scale only axis,
xmin=0,
xmax=8,
xlabel style={font=\color{white!15!black}},
xlabel={\Huge $\eps$},
ymin=0,
ymax=4,
ylabel style={font=\color{white!15!black}},
ylabel={\huge $|\log (\min_{y}\Psi(y))|$},
axis background/.style={fill=white},
xmajorgrids,
ymajorgrids,
yticklabel style = {font=\huge},
xticklabel style = {font=\huge},
legend style={at={(0.677,0.137)}, anchor=south west, legend cell align=left, align=left, draw=white!15!black, legend pos =south east, font=\huge}
]
\addplot [color=mycolor2, line width=2.0pt]
  table[row sep=crcr]{%
0.25	0.0875000000000096\\
0.5	0.175000000000011\\
0.75	0.26250000000001\\
1	0.350000000000011\\
1.25	0.437500000000004\\
1.5	0.525000000000003\\
1.75	0.612500000000002\\
2	0.700000000000002\\
2.25	0.787500000000002\\
2.5	0.875000000000002\\
2.75	0.962500000000001\\
3	1.05\\
3.25	1.1375\\
3.5	1.225\\
3.75	1.3125\\
4	1.4\\
4.25	1.4875\\
4.5	1.575\\
4.75	1.66249999999999\\
5	1.74999999999999\\
5.25	1.83749999999999\\
5.5	1.92403060019714\\
5.75	2.00978060019714\\
6	2.09498850546186\\
6.25	2.17986350546185\\
6.5	2.26349609658308\\
6.75	2.34507725375285\\
7	2.42567449398571\\
7.25	2.50600795900665\\
7.5	2.58460534591637\\
7.75	2.66027040390909\\
8	2.73497226676339\\
};
\addlegendentry{AORR, $\lambda\text{=0.35}$}

\addplot [color=mycolor3, line width=2.0pt]
  table[row sep=crcr]{%
0.25	0.125000000000007\\
0.5	0.250000000000008\\
0.75	0.375000000000007\\
1	0.500000000000004\\
1.25	0.625000000000002\\
1.5	0.750000000000002\\
1.75	0.875000000000002\\
2	1\\
2.25	1.125\\
2.5	1.25\\
2.75	1.375\\
3	1.5\\
3.25	1.62499999999999\\
3.5	1.74999999999999\\
3.75	1.87485903246957\\
4	1.99753060019713\\
4.25	2.11923850546185\\
4.5	2.23992918030685\\
4.75	2.3566744939857\\
5	2.4716744939857\\
5.25	2.58460534591635\\
5.5	2.69252040390906\\
5.75	2.79717158585803\\
6	2.89696184708789\\
6.25	2.99205906320828\\
6.5	3.08121967448893\\
6.75	3.16464217605245\\
7	3.24242322251696\\
7.25	3.30942705561304\\
7.5	3.36603525081281\\
7.75	3.41957703934468\\
8	3.47197106006183\\
};
\addlegendentry{AORR, $\lambda\text{=0.5}$}

\addplot [color=mycolor4, line width=2.0pt]
  table[row sep=crcr]{%
0.25	0.162500000000013\\
0.5	0.325000000000008\\
0.75	0.487500000000007\\
1	0.649390256315193\\
1.25	0.8125\\
1.5	0.975\\
1.75	1.1375\\
2	1.3\\
2.25	1.4625\\
2.5	1.62499999999999\\
2.75	1.78749999999999\\
3	1.94853060019713\\
3.25	2.10711350546185\\
3.5	2.26349609658307\\
3.75	2.4141744939857\\
4	2.56250529491661\\
4.25	2.70327040390907\\
4.5	2.83758870192115\\
4.75	2.96430906320829\\
5	3.08121967448893\\
5.25	3.18888834965375\\
5.5	3.28457919801133\\
5.75	3.36053525081282\\
6	3.43007703934468\\
6.25	3.49651067927509\\
6.5	3.55594578471996\\
6.75	3.61218163522663\\
7	3.65763155073907\\
7.25	3.69629805334337\\
7.5	3.72684425609169\\
7.75	3.75268406536494\\
8	3.77237548521583\\
};
\addlegendentry{AORR, $\lambda\text{=0.65}$}

\end{axis}

% \begin{axis}[%
% width=5.833in,
% height=4.375in,
% at={(0in,0in)},
% scale only axis,
% xmin=0,
% xmax=1,
% ymin=0,
% ymax=1,
% axis line style={draw=none},
% ticks=none,
% axis x line*=bottom,
% axis y line*=left
% ]
% \end{axis}
\end{tikzpicture}%
}
                    \caption{ \centering \label{fig:ORRleakL} Min-lift leakage}
                \end{subfigure}
                ~
                \begin{subfigure}{0.32\textwidth}
                \centering
                    \scalebox{0.42}{\definecolor{mycolor1}{rgb}{0.25, 0.4, 0.96}
\definecolor{mycolor2}{rgb}{0.63, 0.36, 0.94}
\definecolor{mycolor3}{rgb}{0.89, 0.26, 0.2}
\definecolor{mycolor4}{rgb}{0.24, 0.82, 0.44}
\begin{tikzpicture}

\begin{axis}[%
width=4.521in,
height=3.8in,
scale only axis,
xmin=1,
xmax=8,
xlabel style={font=\color{white!15!black}},
xlabel={\Huge $\eps$},
ymin=0.3,
ymax=1,
ylabel style={font=\color{white!15!black}},
ylabel={\huge $\log (\max_{y}\Lambda(y))|$},
axis background/.style={fill=white},
xmajorgrids,
ymajorgrids,
yticklabel style = {font=\huge},
xticklabel style = {font=\huge},
legend style={at={(0.676,0.163)}, anchor=south west, legend cell align=left, align=left, draw=white!15!black, legend pos =south east, font=\huge}
]
\addplot [color=mycolor2, line width=2.0pt]
  table[row sep=crcr]{%
0.25	0.162500000000001\\
0.5	0.323210446452125\\
0.75	0.470947958248966\\
1	0.581525969241216\\
1.25	0.666530109865236\\
1.5	0.719019625527048\\
1.75	0.764003766215179\\
2	0.804629352373904\\
2.25	0.823032339156159\\
2.5	0.840170767701391\\
2.75	0.855577199788193\\
3	0.870547301666824\\
3.25	0.884799443068428\\
3.5	0.894580460130676\\
3.75	0.904064168648248\\
4	0.910957854581468\\
4.25	0.916904053595634\\
4.5	0.921280635554788\\
4.75	0.924782829635191\\
5	0.92771840614981\\
5.25	0.930576669620234\\
5.5	0.932451040048731\\
5.75	0.934036913439309\\
6	0.93524969721735\\
6.25	0.936379600015997\\
6.5	0.937404987589956\\
6.75	0.938230038351117\\
7	0.938914128142929\\
7.25	0.939528628224378\\
7.5	0.940100897965756\\
7.75	0.940583777710556\\
8	0.940963298971468\\
};
\addlegendentry{AORR, $\lambda\text{=0.35}$}

\addplot [color=mycolor3, line width=2.0pt]
  table[row sep=crcr]{%
0.25	0.125000000000004\\
0.5	0.250000000000002\\
0.75	0.375000000000003\\
1	0.498590268483015\\
1.25	0.615904744553702\\
1.5	0.721572530093644\\
1.75	0.804196254061166\\
2	0.851940950468079\\
2.25	0.882273337650788\\
2.5	0.897907506294232\\
2.75	0.908782545651874\\
3	0.917545312189116\\
3.25	0.923302048070007\\
3.5	0.92771840614981\\
3.75	0.931421607742773\\
4	0.933832101432475\\
4.25	0.935581808094598\\
4.5	0.937126238268964\\
4.75	0.93833160345154\\
5	0.939267122246808\\
5.25	0.940100897965758\\
5.5	0.940754173490681\\
5.75	0.941243794892261\\
6	0.941648726373184\\
6.25	0.941976482688812\\
6.5	0.942232943355927\\
6.75	0.942451273824895\\
7	0.94260507136842\\
7.25	0.942727558301171\\
7.5	0.942826539380225\\
7.75	0.942908597569851\\
8	0.942984960871575\\
};
\addlegendentry{AORR, $\lambda\text{=0.5}$}

\addplot [color=mycolor4, line width=2.0pt]
  table[row sep=crcr]{%
0.25	0.0875000000000057\\
0.5	0.175000000000002\\
0.75	0.262500000000005\\
1	0.350000000000001\\
1.25	0.4375\\
1.5	0.524999999999999\\
1.75	0.612345889257678\\
2	0.698576339197841\\
2.25	0.775582504920269\\
2.5	0.840109864989798\\
2.75	0.887392254180233\\
3	0.915028022743425\\
3.25	0.929540745327351\\
3.5	0.936948425181463\\
3.75	0.938820154279955\\
4	0.939942618984163\\
4.25	0.940808137477662\\
4.5	0.941415086467447\\
4.75	0.941890263351745\\
5	0.942232943355927\\
5.25	0.942500479957974\\
5.5	0.942682690204912\\
5.75	0.94281775011981\\
6	0.942923804379706\\
6.25	0.943018839129262\\
6.5	0.94308801322027\\
6.75	0.943137943014827\\
7	0.943172356682406\\
7.25	0.943200438127495\\
7.5	0.943219548954129\\
7.75	0.943235785990171\\
8	0.943249582336637\\
};
\addlegendentry{AORR, $\lambda\text{=0.65}$}

\end{axis}

% \begin{axis}[%
% width=5.833in,
% height=4.375in,
% at={(0in,0in)},
% scale only axis,
% xmin=0,
% xmax=1,
% ymin=0,
% ymax=1,
% axis line style={draw=none},
% ticks=none,
% axis x line*=bottom,
% axis y line*=left
% ]
% \end{axis}
\end{tikzpicture}%
}
                    \caption{ \centering \label{fig:ORRleakU}Max-lift leakage}
                \end{subfigure}
                \caption{Privacy-utility tradeoff of AORR where $|\X|=17$, $|\Sen|=5$, $\eps_{\text{LDP}}\in\{0.25,0.5,0.75,\cdots,8\}$, $\lambda \in \{0.35,0.5,0.65\}$, $\epsl =\lambda\eps$, and $\epsu=(1-\lambda)\eps$. \label{Fig:ORR} }
            \end{figure}
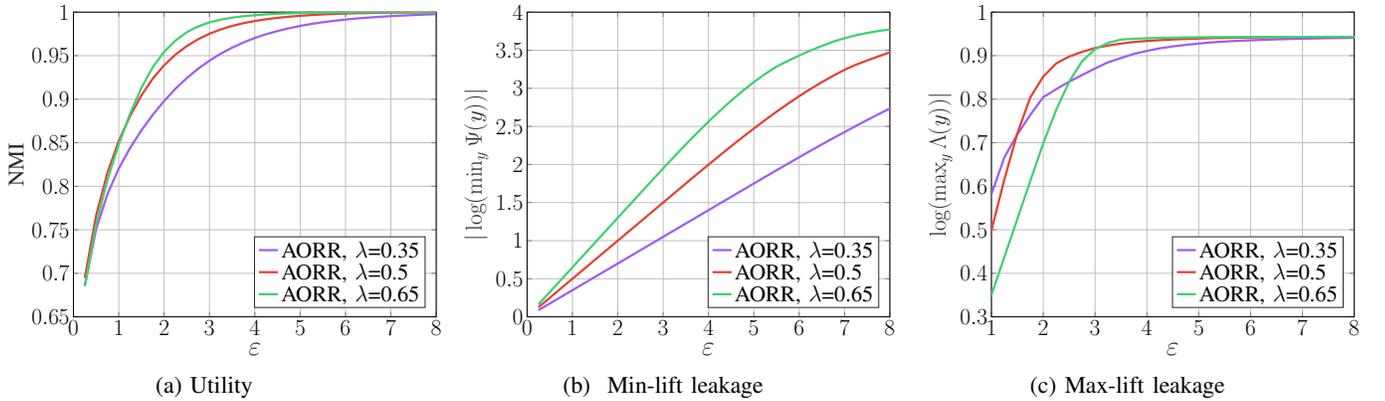

% %=================================================================
% %=================================================================

 \section{Subset Merging in Watchdog Mechanism} \label{sec:subset rand}

    The watchdog mechanism with $X$-invariant randomization is a low-complexity method that can be easily applied when the privacy measures are symbol-wise.
    $X$-invariant randomization is the optimal privacy protection for the high-risk symbols that minimizes privacy leakage in $\XH$ and necessarily results in the worst data resolution.
    Thus, in this section, we propose the \textit{subset merging} algorithm to improve data resolution by randomizing disjoint subsets of high-risk symbols and enhancing utility in the watchdog mechanism.
    In the following, we show that applying $X$-invariant randomization to disjoint subsets of $\XH$ increases the utility.

    Let $\mathcal{G}_{\XH}=\{\X_1,\X_2,\cdots \X_g\}$ be a partition of $\XH$ where for every $i \in [g]$, $\X_i \subseteq\XH$: $\X_i \cap \X_j = \emptyset, i\neq j$, and $\displaystyle \XH=\cup_{i=1}^{g}\X_i$.
    We randomize each subset $\X_i \in \mathcal{G}_{\XH}$ by $X$-invariant randomization $\R_{\Y_i}(y)$ for $x \in \X_i$ and $y \in \Y_i$, where $\sum_{y \in \Y_i}\R_{\Y_i}(y)=1$.
    The resulting  MI between $X$ and $Y$ is
    
    \begin{equation}\label{eq:subset MI}
    I(X;Y)=H(X)-\sum_{i=1}^{g}\sum_{x \in \X_i}\PX(x) \log \frac{P(\X_i)}{\PX(x)}.
    \end{equation}
    
    \begin{Definition}
         Assume two partitions $\mathcal{G}_{\XH}=\{\X_{1},\cdots,\X_{g}\}$ and $\mathcal{G}_{\XH}^{'} = \{\X_{1}^{'}, \cdots, \X_{g'}^{'} \}$. We say\footnote{This definition is inspired from \cite[Definition 10]{2020Liuindexcoding}.}  $\mathcal{G}_{\XH}^{'}$  is a refinement of $\mathcal{G}_{\XH}$, or $\mathcal{G}_{\XH}$ is an aggregation of $\mathcal{G}_{\XH}^{'}$, if for every $i \in [g]$, $\X_{i}=\cup_{j \in J_i}\X_{j}^{'}$ where $J_i\subseteq [g']$, and $P(\X_{i})=\sum_{j\in J_i}{P(\X_{j}^{'})}$.
    \end{Definition}
    
    \noindent If $\mathcal{G}_{\XH}^{'}$ is a refinement of $\mathcal{G}_{\XH}$ then $I_{\mathcal{G}_{\XH}}(X;Y) \leq I_{\mathcal{G}_{\XH}^{'}}(X;Y)$.

    Obtaining the optimal $\mathcal{G}_{\XH}$ that maximizes utility and satisfies privacy constraints is a combinatorial optimization problem over all possible partitions of $\XH$, which is cumbersome to solve.
    Therefore, we propose a heuristic method in the following.
%=================================================================
    \subsection{Greedy Algorithm to Make Refined Subsets of High-Risk Symbols}
    
        In Algorithm \ref{alg:Wsubset algo}, we propose a bottom-up algorithm that constitutes a partition of $\XH$ by merging high-risk symbols in disjoint subsets.
        It works based on a leakage risk metric for each $x \in \XH$: $\omega(x)=\Lambda(x)+\Psi(x)$ for ALIP and $\omega(x)=\Gamma(x)$ for LDP.
        For LIP, $\omega(x)=\max\{ \log\Lambda(x), |\log\Psi(x)| \}$.
        This metric is used to order the subsets by the privacy risk level. 
        Accordingly, to constitute a subset $\X_i \subseteq \XH$, Algorithm \ref{alg:Wsubset algo} bootstraps from the highest risk symbol $\X_i=\{\argmax_{x\in \XH}\omega(x)\}$ (line 5). 
        Then, it merges a symbol $x^{*}$ with $\X_i$ that minimizes $\omega(\X_i\cup x^{*})$ (line 7), as long as the privacy constraints are satisfied in $\X_i$.    
        In Algorithm \ref{alg:Wsubset algo}, we have used ALIP privacy constraints for the while loops condition in lines 4, 6, and 12. 
        For LDP, the privacy constraint is changed to $\Gamma(\X_{Q}) > \eps$, and $\omega(x)$ for LDP is applied. 
        After the constitution of the partition $\mathcal{G}_{\XH}$, the last subset $\X_{g}$ may not meet privacy constraints. 
        Therefore, the leakage of  $\X_{g}$ is checked (line 12), and if there is a privacy breach, an agglomerate $\X_g$ is made by merging other subsets to it that minimizes $\omega(\X_g)$ (lines 13-14), until privacy constraints are satisfied.

%=================================================================
 \begin{algorithm}[t]
            \textbf{Input}: $\X,\epsl,\epsu, \PSX.$\\
            \textbf{Output}: $\mathcal{G}_{\XH}=\{\X_1,\X_2,\cdots\,\X_g\}.$ \\
            \textbf{Initialize}: Obtain $\{\XL,\XH\}$, $\X_{Q} \leftarrow \XH$, and $g=1.$
            
            \While{$\left(\Psi(\X_Q) < \e^{-\epsl} \quad \text{or} \quad  \Lambda(\X_Q) > \e^{\epsu} \right) \quad \text{and} \quad  |\X_{Q}|>0$}
            {
            $\X_g=\displaystyle\argmax_{x\in \X_{Q}}{\omega(x)}$, and
            $\X_{Q} \leftarrow\X_{Q}\setminus \X_g$;
            
            \While{$\left(\Psi(\X_g) < \e^{-\epsl} \quad  \text{or} \quad  \Lambda(\X_g) > \e^{\epsu} \right) \quad  \text{and} \quad |\X_{Q}|>0$}
            {
                $x^{*}=\displaystyle\argmin_{x \in \X_{Q}}\omega(\X_g \cup \{x\});$\\
                $\X_g \leftarrow \X_g \cup \{x^{*}\},$ and $\X_{Q}\leftarrow \X_{Q}\setminus\{x^{*}\}$;
            }
            $\mathcal{G}_{\X_Q}=\{\X_1,\X_2,\cdots,\X_g\}$, and  $g \leftarrow g+1$;
            }
            
            \While {$\left(\Psi(\X_g) < \e^{-\epsl} \quad  \text{or} \quad  \Lambda(\X_g) > \e^{\epsu} \right) \quad and \quad |\mathcal{G}_{\X_Q}|>1$ }
            {
            $i^{*}=\displaystyle\argmin_{1\leq i<g}\omega(\X_g\cup\X_i), \quad$ and $\X_g \leftarrow \X_g\cup\X_{i^{*}}$;\\
            
            For $i^{*}+1\leq j \leq g$ update the indices of $\X_j$'s to $\X_{j-1}$ and $g \leftarrow{g-1};$
            
            $\mathcal{G}_{\X_Q}=\{\X_1,\X_2,\cdots,\X_g\};$
            }
            \caption{Subset merging in the watchdog mechanism}\label{alg:Wsubset algo}
        \end{algorithm}

    \subsection{Numerical Results}\label{subsec:Num Sub}
        We show PUT for ALIP and LDP under subset merging randomization in Figure~\ref{Fig:Subset} with the same setup for the watchdog mechanism in Section \ref{subsubsec:Num watchdog}.
        Compared with the complete merging (Figure \ref{Fig:Watchdog}), the utility has been enhanced significantly for both LDP and ALIP in all scenarios under the same privacy constraint.
        For instance, consider the symmetric case $\lambda=0.5$ when $\eps=2$ and compare PUT between the subset and complete merging.
        Figure~\ref{fig:SubsetUt} demonstrates a utility value of around $0.83$ for the subset merging compared to the utility of 0.52 for the complete merging in Figure~\ref{fig:WatchdogUt}, which is almost $60\%$ utility enhancement.         
        Moreover, as Figures \ref{fig:SubsetL} and \ref{fig:SubsetU} illustrate, privacy constraints have been satisfied in all cases.

         \begin{figure}
            \centering
            \begin{subfigure}{0.32\textwidth}
                \centering
                \scalebox{0.42}{\definecolor{mycolor1}{rgb}{0.25, 0.4, 0.96}
\definecolor{mycolor2}{rgb}{0.63, 0.36, 0.94}
\definecolor{mycolor3}{rgb}{0.89, 0.26, 0.2}
\definecolor{mycolor4}{rgb}{0.24, 0.82, 0.44}
\begin{tikzpicture}

\begin{axis}[%
width=4.521in,
height=3.8in,
scale only axis,
xmin=0,
xmax=8,
xlabel style={font=\color{white!15!black}},
xlabel={\Huge $\eps$},
ymin=0.3,
ymax=1,
ylabel style={font=\color{white!15!black}},
ylabel={\huge NMI},
axis background/.style={fill=white},
yticklabel style = {font=\huge},
xticklabel style = {font=\huge},
xmajorgrids,
ymajorgrids,
legend style={at={(0.678,0.144)}, anchor=south west, legend cell align=left, align=left, draw=white!15!black, legend pos =south east, font=\huge}
]
\addplot [color=mycolor1, line width=2.0pt]
  table[row sep=crcr]{%
0.25	0.459110073650291\\
0.5	0.627873759362456\\
0.75	0.720662430743174\\
1	0.768442403895286\\
1.25	0.805461021579951\\
1.5	0.836115726177183\\
1.75	0.863599430479574\\
2	0.886723839023723\\
2.25	0.906769684638462\\
2.5	0.924119955436181\\
2.75	0.937448414879212\\
3	0.94928350076793\\
3.25	0.959445645343815\\
3.5	0.967720706018441\\
3.75	0.974383357014169\\
4	0.978397116136727\\
4.25	0.982141449022878\\
4.5	0.986004874396508\\
4.75	0.988214331390727\\
5	0.990461714742568\\
5.25	0.992295983386668\\
5.5	0.994018595477692\\
5.75	0.995329711133203\\
6	0.996461228240369\\
6.25	0.997136439778287\\
6.5	0.997829605412205\\
6.75	0.998445424698712\\
7	0.998845808951355\\
7.25	0.998975612524407\\
7.5	0.999367360931565\\
7.75	0.999521143953775\\
8	0.99973739335366\\
};
\addlegendentry{LDP}

\addplot [color=mycolor2, line width=2.0pt]
  table[row sep=crcr]{%
0.25	0.358478062631037\\
0.5	0.529607659768487\\
0.75	0.622473533000207\\
1	0.68188230740167\\
1.25	0.720543243194707\\
1.5	0.747405666476861\\
1.75	0.766031057161358\\
2	0.782791013421573\\
2.25	0.799912471116115\\
2.5	0.813464428801842\\
2.75	0.82808906871485\\
3	0.840008415530739\\
3.25	0.850449348333397\\
3.5	0.861615159835804\\
3.75	0.870957580270112\\
4	0.879974256472731\\
4.25	0.888556992523242\\
4.5	0.896265331516895\\
4.75	0.90437716430036\\
5	0.910892573897623\\
5.25	0.9171194522307\\
5.5	0.923470257666593\\
5.75	0.929326201013155\\
6	0.93604374370543\\
6.25	0.940031997452866\\
6.5	0.944384029815996\\
6.75	0.947787072294927\\
7	0.952051254954379\\
7.25	0.955545940456403\\
7.5	0.959409920774493\\
7.75	0.962646736218458\\
8	0.965355552597384\\
};
\addlegendentry{$\text{ALIP, }\lambda\text{=0.35}$}

\addplot [color=mycolor3, line width=2.0pt]
  table[row sep=crcr]{%
0.25	0.423804330641361\\
0.5	0.590405363778221\\
0.75	0.684677310943303\\
1	0.736420576644848\\
1.25	0.767740351108311\\
1.5	0.792155610277785\\
1.75	0.813373048682645\\
2	0.833288456878917\\
2.25	0.849283505285081\\
2.5	0.863822135742257\\
2.75	0.87774225400724\\
3	0.88969951472988\\
3.25	0.900335017819169\\
3.5	0.910892573897623\\
3.75	0.919803977222241\\
4	0.928454419832643\\
4.25	0.936903782914696\\
4.5	0.943127581765371\\
4.75	0.948336949879445\\
5	0.953889614612995\\
5.25	0.959409920774493\\
5.5	0.963891608629518\\
5.75	0.967593922818465\\
6	0.971208749127081\\
6.25	0.973971488902286\\
6.5	0.976518582725321\\
6.75	0.979008418956583\\
7	0.980717345300058\\
7.25	0.982889925424401\\
7.5	0.984654344815168\\
7.75	0.986026357271945\\
8	0.987430932029907\\
};
\addlegendentry{$\text{ALIP, }\lambda\text{=0.5}$}

\addplot [color=mycolor4, line width=2.0pt]
  table[row sep=crcr]{%
0.25	0.376996639840564\\
0.5	0.558259060278165\\
0.75	0.662930538357567\\
1	0.729824797282164\\
1.25	0.773194581708466\\
1.5	0.809682088262084\\
1.75	0.838271561184812\\
2	0.863310854047473\\
2.25	0.882948263431518\\
2.5	0.898962589026597\\
2.75	0.913126584776974\\
3	0.924669587548002\\
3.25	0.936477698512387\\
3.5	0.944509766317719\\
3.75	0.95179176025874\\
4	0.958030830797931\\
4.25	0.964309756303203\\
4.5	0.969183574107671\\
4.75	0.973119744048833\\
5	0.976518582725321\\
5.25	0.979538010895329\\
5.5	0.981995776563415\\
5.75	0.984531274733933\\
6	0.986349630965539\\
6.25	0.988061918339243\\
6.5	0.989492565042365\\
6.75	0.990930410625954\\
7	0.992033342740412\\
7.25	0.993302396228624\\
7.5	0.99432821094937\\
7.75	0.995250455670272\\
8	0.995866352108962\\
};
\addlegendentry{$\text{ALIP, }\lambda\text{=0.65}$}

\end{axis}
\end{tikzpicture}%
}
                \caption{\centering \label{fig:SubsetUt} Utility}
            \end{subfigure}%
            ~
            \begin{subfigure}{0.32\textwidth}
                \centering
                \scalebox{0.42}{\definecolor{mycolor1}{rgb}{0.25, 0.4, 0.96}
\definecolor{mycolor2}{rgb}{0.63, 0.36, 0.94}
\definecolor{mycolor3}{rgb}{0.89, 0.26, 0.2}
\definecolor{mycolor4}{rgb}{0.24, 0.82, 0.44}

\begin{tikzpicture}

\begin{axis}[%
width=4.521in,
height=3.8in,
scale only axis,
xmin=0,
xmax=8,
xlabel style={font=\color{white!15!black}},
xlabel={\Huge $\eps$},
ymin=0,
ymax=4.5,
ylabel style={font=\color{white!15!black}},
ylabel={\huge $|\log (\min_{y}\Psi(y))|$},
yticklabel style = {font=\huge},
xticklabel style = {font=\huge},
axis background/.style={fill=white},
xmajorgrids,
ymajorgrids,
legend style={at={(0.678,0.144)}, anchor=south west, legend cell align=left, align=left, draw=white!15!black, legend pos =north west, font=\huge}
]
\addplot [color=mycolor1, line width=2pt]
  table[row sep=crcr]{%
0.25	0.121494325339405\\
0.5	0.265799863515263\\
0.75	0.415967347232995\\
1	0.576019607478822\\
1.25	0.762620602840971\\
1.5	0.956405609868805\\
1.75	1.16236248659756\\
2	1.36767002727827\\
2.25	1.57600992351055\\
2.5	1.77747708769132\\
2.75	1.98013078030946\\
3	2.18793809722251\\
3.25	2.37331147913692\\
3.5	2.56936800568908\\
3.75	2.73138168217286\\
4	2.86747279536285\\
4.25	3.0119113123323\\
4.5	3.15356912950409\\
4.75	3.23957121077668\\
5	3.36299835817931\\
5.25	3.46547749284933\\
5.5	3.55372374866484\\
5.75	3.62878977849134\\
6	3.70517761169415\\
6.25	3.74604056723889\\
6.5	3.80431983857265\\
6.75	3.85775355655917\\
7	3.89403249321602\\
7.25	3.9100013046765\\
7.5	3.95226176810032\\
7.75	3.97217479065729\\
8	4.00368559935366\\
};
\addlegendentry{LDP}

\addplot [color=mycolor2, line width=2.0pt]
  table[row sep=crcr]{%
0.25	0.0716476930818299\\
0.5	0.152045936372225\\
0.75	0.23465163576976\\
1	0.314178328392778\\
1.25	0.39046717984945\\
1.5	0.463795719608919\\
1.75	0.538387043610151\\
2	0.620505343312084\\
2.25	0.704993077614427\\
2.5	0.786502421773196\\
2.75	0.86785906525081\\
3	0.948171423225569\\
3.25	1.02953633267219\\
3.5	1.1132490078135\\
3.75	1.19960643449485\\
4	1.27566982468603\\
4.25	1.36702411393656\\
4.5	1.4454651566074\\
4.75	1.52646957452391\\
5	1.60103157333297\\
5.25	1.68382152465132\\
5.5	1.77155105321742\\
5.75	1.84942621026719\\
6	1.93756028124291\\
6.25	2.0047940039563\\
6.5	2.08022685344095\\
6.75	2.15025845247078\\
7	2.23072031478202\\
7.25	2.29841512612995\\
7.5	2.3705173967277\\
7.75	2.43459936651921\\
8	2.50166825285938\\
};
\addlegendentry{$\text{ALIP, }\lambda\text{=0.35}$}

\addplot [color=mycolor3, line width=2.0pt]
  table[row sep=crcr]{%
0.25	0.100470681758472\\
0.5	0.214010953623714\\
0.75	0.329254488866525\\
1	0.433690333021954\\
1.25	0.546518348097214\\
1.5	0.666785285868302\\
1.75	0.785893867908312\\
2	0.902402204746048\\
2.25	1.02061284906232\\
2.5	1.13709387166787\\
2.75	1.25877407230278\\
3	1.37761883027481\\
3.25	1.48753120113704\\
3.5	1.60103157333297\\
3.75	1.72119428014853\\
4	1.83861429469511\\
4.25	1.95859032684717\\
4.5	2.0616441697903\\
4.75	2.15954623522229\\
5	2.2719876922487\\
5.25	2.3705173967277\\
5.5	2.46119432208567\\
5.75	2.55206115049235\\
6	2.63269002867425\\
6.25	2.72512153773408\\
6.5	2.80191227049274\\
6.75	2.88287324180595\\
7	2.94876706276903\\
7.25	3.02821761741982\\
7.5	3.10028913598181\\
7.75	3.15036845200669\\
8	3.21194730825408\\
};
\addlegendentry{$\text{ALIP, }\lambda\text{=0.5}$}

\addplot [color=mycolor4, line width=2.0pt]
  table[row sep=crcr]{%
0.25	0.100001339029975\\
0.5	0.230333848703204\\
0.75	0.368303425086537\\
1	0.51010339586177\\
1.25	0.663713813017899\\
1.5	0.838339419460715\\
1.75	0.999531204588441\\
2	1.16316721218608\\
2.25	1.32418181125073\\
2.5	1.4785120229254\\
2.75	1.63105096446341\\
3	1.78687777342404\\
3.25	1.94750261875678\\
3.5	2.08020257735981\\
3.75	2.22016920882192\\
4	2.34290212358649\\
4.25	2.47463617348708\\
4.5	2.58725588269975\\
4.75	2.69832005663473\\
5	2.80191227049274\\
5.25	2.90116464637675\\
5.5	2.98871122346467\\
5.75	3.09605865649054\\
6	3.16002138762021\\
6.25	3.23972568227927\\
6.5	3.32502838802308\\
6.75	3.39053681144082\\
7	3.45118786768393\\
7.25	3.5191266348141\\
7.5	3.57427377513839\\
7.75	3.62717062745526\\
8	3.66672417457115\\
};
\addlegendentry{$\text{ALIP, }\lambda\text{=0.65}$}

\end{axis}
\end{tikzpicture}%
}
                \caption{\centering \label{fig:SubsetL}Min-lift leakage}
            \end{subfigure}
            ~
            \begin{subfigure}{0.32\textwidth}
                \centering
                \scalebox{0.42}{\definecolor{mycolor1}{rgb}{0.25, 0.4, 0.96}
\definecolor{mycolor2}{rgb}{0.63, 0.36, 0.94}
\definecolor{mycolor3}{rgb}{0.89, 0.26, 0.2}
\definecolor{mycolor4}{rgb}{0.24, 0.82, 0.44}
\begin{tikzpicture}

\begin{axis}[%
width=4.521in,
height=3.8in,
scale only axis,
xmin=0,
xmax=8,
xlabel style={font=\color{white!15!black}},
xlabel={\Huge $\eps$},
ymin=0,
ymax=1,
ylabel style={font=\color{white!15!black}},
ylabel={\huge $\log (\max_{y}\Lambda(y))|$},
yticklabel style = {font=\huge},
xticklabel style = {font=\huge},
axis background/.style={fill=white},
xmajorgrids,
ymajorgrids,
legend style={at={(0.678,0.15)}, anchor=south west, legend cell align=left, align=left, draw=white!15!black, legend pos =south east, font=\huge}
]
\addplot [color=mycolor1, line width=2pt]
  table[row sep=crcr]{%
0.25	0.112974602955844\\
0.5	0.233624184181784\\
0.75	0.341590538753755\\
1	0.430807135395785\\
1.25	0.513736292411222\\
1.5	0.588985619375826\\
1.75	0.654820157940961\\
2	0.707248667859875\\
2.25	0.75593237953718\\
2.5	0.796607490325322\\
2.75	0.823144100124574\\
3	0.846492961794939\\
3.25	0.867165473637785\\
3.5	0.883894651979881\\
3.75	0.895491276788933\\
4	0.900782254318123\\
4.25	0.910359314277489\\
4.5	0.917350001212985\\
4.75	0.922889221951078\\
5	0.925720668440168\\
5.25	0.93025025674278\\
5.5	0.933116535461875\\
5.75	0.934733550240487\\
6	0.936160123239669\\
6.25	0.936592364274554\\
6.5	0.937149699266919\\
6.75	0.9383782295747\\
7	0.939511561911719\\
7.25	0.939738893815545\\
7.5	0.940395338924634\\
7.75	0.940395338924634\\
8	0.940557619109968\\
};
\addlegendentry{LDP}

\addplot [color=mycolor2, line width=2.0pt]
  table[row sep=crcr]{%
0.25	0.0884371877339209\\
0.5	0.186452588680739\\
0.75	0.283940908302974\\
1	0.347575640575687\\
1.25	0.40707795791183\\
1.5	0.456251975800938\\
1.75	0.498382904953184\\
2	0.543592989242814\\
2.25	0.584480741546614\\
2.5	0.613236059930186\\
2.75	0.644375558807474\\
3	0.67075205580961\\
3.25	0.697246710085012\\
3.5	0.723368974500588\\
3.75	0.743100920154711\\
4	0.758834070122615\\
4.25	0.778620643643115\\
4.5	0.794026061963416\\
4.75	0.806214359741382\\
5	0.817201025176129\\
5.25	0.830024642737838\\
5.5	0.839257002471818\\
5.75	0.849081387251402\\
6	0.856621501110547\\
6.25	0.86171809344412\\
6.5	0.867280776134947\\
6.75	0.872770187246138\\
7	0.881272583448266\\
7.25	0.886699293265365\\
7.5	0.890948906975204\\
7.75	0.894570905870688\\
8	0.897709710521392\\
};
\addlegendentry{$\text{ALIP, }\lambda\text{=0.35}$}

\addplot [color=mycolor3, line width=2.0pt]
  table[row sep=crcr]{%
0.25	0.0978048292281106\\
0.5	0.203887930699411\\
0.75	0.30706115018029\\
1	0.393189384628559\\
1.25	0.470128458619609\\
1.5	0.545326266512001\\
1.75	0.600643361262103\\
2	0.647695106557557\\
2.25	0.690911082523663\\
2.5	0.727679909350835\\
2.75	0.756580013007945\\
3	0.781037793761278\\
3.25	0.80156340840847\\
3.5	0.817201025176129\\
3.75	0.832881655949906\\
4	0.84855142661314\\
4.25	0.858550948516187\\
4.5	0.866148639205962\\
4.75	0.874491831377261\\
5	0.885111612354062\\
5.25	0.890948906975204\\
5.5	0.896438238520096\\
5.75	0.901583492973817\\
6	0.904722807757693\\
6.25	0.910212217133057\\
6.5	0.912138448607207\\
6.75	0.9145416507032\\
7	0.91849281448742\\
7.25	0.921208711806055\\
7.5	0.923287464572238\\
7.75	0.924986763551075\\
8	0.926475333337892\\
};
\addlegendentry{$\text{ALIP, }\lambda\text{=0.5}$}

\addplot [color=mycolor4, line width=2.0pt]
  table[row sep=crcr]{%
0.25	0.0720500181715683\\
0.5	0.154623624218597\\
0.75	0.235037031747464\\
1	0.315033346393779\\
1.25	0.39163225552979\\
1.5	0.47105026658902\\
1.75	0.545143911463551\\
2	0.615188619534916\\
2.25	0.678381611100602\\
2.5	0.732581387503089\\
2.75	0.776490369114171\\
3	0.813679785314367\\
3.25	0.841539330187849\\
3.5	0.858888512417817\\
3.75	0.878680193681852\\
4	0.888623580382636\\
4.25	0.897601691647827\\
4.5	0.902957425648113\\
4.75	0.908502951348704\\
5	0.912138448607207\\
5.25	0.915786735992914\\
5.5	0.920988950038128\\
5.75	0.92323573906196\\
6	0.925160329866889\\
6.25	0.92694391924252\\
6.5	0.929766393032593\\
6.75	0.931367973734501\\
7	0.932417348774588\\
7.25	0.934017881892497\\
7.5	0.935024052852873\\
7.75	0.93578171836961\\
8	0.9364812357856\\
};
\addlegendentry{$\text{ALIP, }\lambda\text{=0.65}$}

\end{axis}
\end{tikzpicture}%
}
                \caption{\centering \label{fig:SubsetU}Max-lift leakage \vspace{5pt}}
            \end{subfigure}
            \caption{\label{Fig:Subset} Privacy-utility tradeoff of subset merging randomization under $\eps$-LDP and $(\epsl,\epsu)$-ALIP, where $|\X|=17$, $|\Sen|=5$, $\eps_{\text{LDP}}\in\{0.25,0.5,0.75,\cdots,8\}$, $\lambda \in \{0.35,0.5,0.65\}$, $\epsl=\lambda\eps$, and $\epsu=(1-\lambda)\eps$.}
        \end{figure}
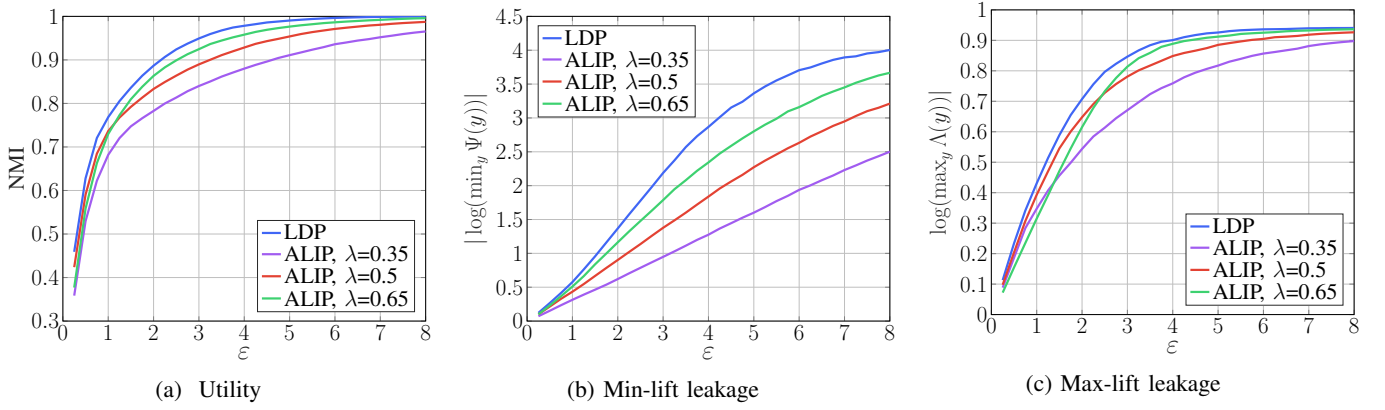
        
%=================================================================
%=================================================================
\section{Subset Random Response}\label{sec:ORRsub}

    In the previous section, we have shown that subset merging enhances utility in the watchdog mechanism significantly.
    In this section, we propose a method to decrease the complexity of AORR for large datasets.
    We adopt AORR for subsets of $\XH$ to decrease the complexity of AORR for large sets such that random response becomes applicable for typically an order of magnitude larger $\X$.

    The AORR optimization problem in \eqref{eq:AORR} is equivalent to the following problem

    \begin{align} \label{eq:AORR Hxgy}
        &H(x) - \min_{\PXgY, \PY} H(X|Y)\\
        \nonumber &\textrm{s.t.}  \quad S-X-Y\\
        \nonumber&\hspace{25pt}\e^{-\epsl} \leq \Psi(y) \quad \text{and} \quad \Lambda(y)\leq \nonumber \e^{\epsu}, \hspace{5pt} \forall y \in 	\Y. 
    \end{align}
    To reduce the complexity of \eqref{eq:AORR Hxgy}, we divide $\X$ into $\XL$ and $\XH$, similar to the watchdog mechanism, and make a partition $\mathcal{G}_{\XH}=\{\X_1,\X_2,\cdots \X_g\}$ from $\XH$. 
    We randomize each subset $\X_{i} \in \mathcal{G}_{\XH}$, $i \in [g]$, separately by a randomization pair $\left(\Q_i,\q_i\right)$, where $\Q_i$ is a matrix in $\Real^{|\X_i|\times|\Y_{i}|}$ and $\q_i$ is a vector in $\Real^{|\Y_i|}$. 
    The elements of $\Q_i$ and $\q_i$ are given by $\Q_i(x|y)=\Pr[X=x|Y=y], x \in \X_i, y \in \Y_i$ and $\q_i(y)=\Pr[Y=y], y \in \Y_i$, respectively. For each $y \in \Y_i$, we have $\sum_{x\in\X_i}\Q_i(x|y)=1$.
    Consequently, $H(X|Y)=\sum_{i\in[g]}H_{i}(X|Y)$, where $H_{i}(X|Y)=-\sum_{y\in\Y_i}\q_{i}(y) \sum_{x\in\X_i} \Q_{i}(x|y)\log\Q_i(x|y)$. 
    This setting turns \eqref{eq:AORR Hxgy} into $g$ optimization problems for each subset $\X_i \in \mathcal{G}_{\XH},$ $i \in [g]$ as follows:
    
    \begin{align}\label{eq:AORR Hxgy part}
        &\min_{\Q_i, \q_i} H_i(X|Y)\\
        %  \quad & S-X-Y,\\
        \textrm{s.t.} \quad & 0 \leq \q_i(y),                           && \forall y \in \Y_i, \\
        & 0 \leq \Q_i(x|y),                         && \forall x \in \X_i, \forall y \in \Y_i, \\
        & \sum_{x \in \X_i}\Q_i(x|y)=1,             && \forall y\in\Y_{i}, \\
        & \sum_{y \in \Y_i}\Q_i(x|y)\q_i(y)=\PX(x), && x\in\X_i, \\
        & \e^{-\epsl}{\PS(s)} \leq \sum_{x \in \X_{i}}\PSgX(s|x)\Q_{i}(x|y) \leq \e^{\epsu} \PS(s), && \forall s \in \Sen, y \in \Y_i.    
    \end{align}
    
    The columns of the randomization matrix $\Q_i$, $i \in [g]$ can be expressed as the members of a convex and bounded polytope $\Pi_i$, which is given by the following constraints
    
    \begin{equation}\label{eq:polytope}
        {\Pi_i=\left\{
            \begin{array}{ll}
                &\mathbf{v} \in \mathbb{R}^{|\X_{i}|}: \\
                &0 \leq v_{k}, \forall k \in [|\X_i|],\\
                & \sum_{k=1}^{|\X_i|} v_{k}=1, \\
                & \mathrm{e}^{-\epsl}\PS(s)\leq\sum_{x\in\X_i}\PSgX(s|x)v_{k} \leq \mathrm{e}^{\epsu}\PS(s), \forall s \in \Sen, k \in [|\X_{i}|]
            \end{array} \right\}}.
    \end{equation}
    For each $\X_i \in \mathcal{G}_{\XH},$  $i \in [g]$, let $\mathcal{V}_i=\{\mathbf{v}^{i}_1\cdots,\mathbf{v}^{i}_M\}$ be the vertices of $\Pi_i$ in \eqref{eq:polytope},   $\mathrm{H}\left(\mathbf{v}^{i}_{k}\right)$ be the entropy of each $\mathbf{v}^{i}_{k}$ for $k \in [M]$, $\boldsymbol{P^{i}_{X}}$ be the probability vector of  $x \in \X_{i}$, and $\boldsymbol{{\beta}}^i$ be the solution to the following  optimization
    
    \begin{equation}\label{eq:optim1}
        \begin{aligned}
            &\min_{\boldsymbol{{\beta}}^i \in \mathbb{R}^{M}} \sum_{k=1}^{M} \mathrm{H}\left(\mathbf{v}^{i}_{k}\right) \beta^{i}_{k}, \\
             \text {s.t.}\quad & \beta^{i}_{k} \geq 0, \quad \forall k \in [M],  \\
            &\sum_{k=1}^{M}\mathbf{v}^{i}_{k} \beta^{i}_{k} =\boldsymbol{P^{i}_{X}}.
        \end{aligned}
    \end{equation}
    Then, $\Y_i$ and $\left(\Q_i,\q_i \right)$ are given by
    \begin{gather}
         \Y_i=\{y: {\beta}^{i}_{y}\neq 0\};\\
         \q_i(y)={\beta}^{i}_{y} \quad \text{and} \quad \Q_i(.|y)= {\mathbf{v}^{i}_{y}}, \quad y \in \Y_i.
    \end{gather}
    The $(\epsl,\epsu)$-ALIP  protocol, $\mathcal{M}:\mathcal{X} \rightarrow \mathcal{Y}$, is given by the pair of $(\PXgY,\PY)$ as follows:
    \begin{equation} \label{eq:Subset ORRpxgy}
        \PXgY = \begin{cases}
            \mathbf{1}_{\{x=y\}}, & x,y\in \XL=\YL,\\
            \Q_{i}(x|y),		 			 & x \in \X_i, y \in \Y_i, i\in[g],\\
            0,		 			 & \textup{otherwise};
        \end{cases}\\
    \end{equation}
    
    \begin{equation} \label{eq:Subset ORR py}
        \PY =\begin{cases}
            \PX(y),    & y \in \YL,\\
            \q_{i}(y), & y \in \Y_i,  i \in [g].
        \end{cases}
    \end{equation}  
%================================================================
        \begin{algorithm}
            \textbf{Input}: $\X,\epsl,\epsu, \PSX$.\\
            \textbf{Output}: $\mathcal{O}_{\XH}=\{\X'_1,\X'_2,\cdots\,\X'_{g'}\}$, $\PXgY$, and $\PY$. \\
            \textbf{Initialize}: Obtain $\mathcal{G}_{\XH}=\{\X_1,\X_2,\cdots\,\X_g\}$ by Algorithm \ref{alg:Wsubset algo} and $g'=1$.\\
            \While{$|\mathcal{G}_{\XH}|>0$}
            {   $\X'_{g'}=\X_1,$ find $\Pi_{g'}$ in \eqref{eq:polytope}, and $\mathcal{G}_{\XH}\leftarrow\mathcal{G}_{\XH}\setminus\X_{1}$;\\
                Update subset indices in $\mathcal{G}_{\XH}$, such that $\X_{i-1} \leftarrow \X_i$ for $2\leq i \leq g$ and $g \leftarrow g-1$\label{alg:polytop};\\
                \While{$\Pi_{g'}=\emptyset$ and  $ \left( |\mathcal{G}_{\XH}|>0 \textrm{ or } |\mathcal{O}_{\XH}|>0 \right)$}
                {   \uIf{$|\mathcal{G}_{\XH}|>0$}
                    {   $\X'_{g'}=\X'_{g'}\cup\X_{1},$ find $\Pi_{g'}$ in \eqref{eq:polytope}, and $\mathcal{G}_{\XH} \leftarrow \mathcal{G}_{\XH} \setminus\X_{1}$; \\
                    Update subset indices in $\mathcal{G}_{\XH}$, such that $\X_{i-1} \leftarrow \X_i$ for $2\leq i \leq g$ and $g \leftarrow g-1$;\\
                    }
                    \Else
                    {   $\X'_{g'}=\X'_{g'}\cup\X'_{g'-1}$, $g' \leftarrow g'-1$, and find $\Pi_{g'}$ in \eqref{eq:polytope}; 
            
                    } 
                }
                Set $\mathcal{O}_{\XH}=\{\X'_1,\X'_2,\cdots\,\X'_{g'}\}$ and $g'\leftarrow g'+1;$
            }
            
            \uIf{$|\mathcal{O}_{\XH}|=1$ and $\Pi_1=\emptyset$}
                {  Apply subset merging mechanism for $\mathcal{G}_{\XH}$.}
                \Else
                {   \For{$i \leftarrow 1$ to $g'$}
                    {Solve optimization in \eqref{eq:optim1} for $\X^{'}_i$ and obtain $\boldsymbol{{\beta_{i}}}$, $\Q_{i}$, and $\q_{i}$;}   
                 Obtain $\PXgY$ in \eqref{eq:Subset ORRpxgy} and $\PY$ in \eqref{eq:Subset ORR py}.
                }
            \caption{Subset random response}\label{alg:SubORR_alg}
        \end{algorithm}

    \subsection{Algorithm for Subset Random Response}
    
        We propose Algorithm \ref{alg:SubORR_alg} to implement AORR for subsets of $\XH$, which we call \textit{subset random response} (SRR).
        In this algorithm, first, we obtain a partition $\mathcal{G}_{\XH}=\{\X_1,\X_2,\cdots\,\X_{g}\}$ of $\XH$ via Algorithm \ref{alg:Wsubset algo}.
        Then, we find the optimal random response for each subset $\X^{'}_i \in \mathcal{G}_{\XH}$ (line 5).
        By obtaining the optimal random responses for all subsets, we get a pair $(\Q_{i}, \q_{i})$ for each subset $\X^{'}_i$ and consequently $\PXgY$ and $\PY$ by \eqref{eq:Subset ORRpxgy} and \eqref{eq:Subset ORR py} (lines 20-23).
        The while loop in lines 7-14 is for the particular cases when the polytope in \eqref{eq:polytope} is empty for a subset $\X'_{i}$. 
        It may occur for strict privacy conditions where the privacy budget is very small.  Since we reduce the dimension of the original polytope in \eqref{eq:AORR Hxgy}, it increases the possibility that no feasible random response exists in some cases. 
        Therefore, in such cases, we make a union with other subsets in $\mathcal{G}_{\XH}$ (line 9) or $\mathcal{O}_{\XH}$ (line 12) until we have a nonempty polytope.
        The if condition in lines 17-18 is for the cases where there is no feasible polytope after making a union of all subsets.
        Whenever this occurs, we apply the subset merging mechanism.

        The number of  polytope vertices in AORR is $n \sim O\left(\exp{(|\X|-1)}\right)$, and the time complexity is $O(n)$.
        As $|\X|$ increases, the complexity of AORR increases exponentially in the $|\X|-1$.
        In SRR, for each $\X_i \in \mathcal{G}_{\XH}$, $i \in [g]$, the number of vertices is $n_i \sim O(\exp{(|\X_i|-1)})$, and the complexity of SRR is $O(\max_{i} n_i)$. 
        Thus, the complexity of SRR increases in the maximum subset size, which can be much lower than that in AORR.
          
%=================================================================
    \subsection{Numerical Results}
        Here, we compare the PUT of AORR with SRR in Algorithm \ref{alg:SubORR_alg} and subset merging in Algorithm \ref{alg:Wsubset algo}. 
        Figure \ref{Fig:SubsetORR} depicts mean values of utility, leakage, and time complexity for 100 randomly generated distributions where $\lambda=0.65$ and simulation setup is the same as that in Section \ref{subsubsec:Num watchdog}.

        Figures~\ref{fig:SRRutility} to \ref{fig:SRRleakU} demonstrate that SRR results in better utility and higher leakage than subset merging, and its PUT is very close to AORR.
        Figure~\ref{fig:ORRtime} illustrates the processing time for each mechanism from which we observe that the complexity of AORR and SRR is much higher than the subset merging. 
        Running SRR is less complex than AORR for strict privacy constraints ($\eps<1$) and for $\eps>2.5$. While SRR shows higher complexity for some privacy budgets, $1 \leq \eps\leq 2.5$, it has the advantage in high dimension systems. 
        Figure~\ref{Fig:SRRlarge} shows a PUT comparison between SRR and subset merging where $|\X|=200$, $|\Sen|=15$, $\eps\in\{1,1.25,\cdots,8\}$, and $\lambda=0.5$. 
        This experiment shows that both SRR and subset merging are applicable to large datasets. Obviously, SRR provides better utility (Figure \ref{fig:SRRlargeUt}) and higher leakage (Figures \ref{fig:SRRlargeL} and \ref{fig:SRRlargeU}), which is still below the given budgets $\epsl$ and $\epsu$.  

         \begin{figure}
            \centering
            \begin{subfigure}{0.49\textwidth}
            \centering
                \scalebox{0.45}{\definecolor{mycolor1}{rgb}{0.25, 0.4, 0.96}
\definecolor{mycolor2}{rgb}{0.63, 0.36, 0.94}
\definecolor{mycolor3}{rgb}{0.89, 0.26, 0.2}
\definecolor{mycolor4}{rgb}{0.24, 0.82, 0.44}

\begin{tikzpicture}

\begin{axis}[%
width=4.521in,
height=3.8in,
at={(0.7in,0.519in)},
scale only axis,
xmin=0,
xmax=8,
xlabel style={font=\color{white!15!black}},
xlabel={\Huge $\eps$},
ymin=0.3,
ymax=1,
ylabel style={font=\color{white!15!black}},
ylabel={\huge NMI},
axis background/.style={fill=white},
xmajorgrids,
ymajorgrids,
yticklabel style = {font=\huge},
xticklabel style = {font=\huge},
legend style={legend cell align=left, align=left, draw=white!15!black, legend pos= south east, font=\huge}
]
\addplot [color=mycolor1, line width=3.0pt]
  table[row sep=crcr]{%
0.25	0.682815395988472\\
0.5	0.753690966085351\\
0.75	0.806650490732917\\
1	0.848858488367128\\
1.25	0.884654264914269\\
1.5	0.913158213168902\\
1.75	0.934992089535257\\
2	0.952288826741206\\
2.25	0.96520370233802\\
2.5	0.975001675609365\\
2.75	0.981735180617857\\
3	0.986597498747196\\
3.25	0.990194850375445\\
3.5	0.992780418196244\\
3.75	0.994706624801367\\
4	0.996102431704974\\
4.5	0.997807783479946\\
5	0.998760724000254\\
5.75	0.999468438092729\\
6.75	0.999840562287352\\
8	0.999956544189233\\
};
\addlegendentry{AORR}

\addplot [color=mycolor2, line width=3.0pt]
  table[row sep=crcr]{%
0.25	0.646640162616452\\
0.5	0.726185658542931\\
0.75	0.790161525533449\\
1	0.837505274458278\\
1.25	0.876629488957375\\
1.5	0.905753464929939\\
1.75	0.931036360935499\\
2	0.950553181070148\\
2.25	0.962525704331128\\
2.5	0.975424697695482\\
2.75	0.984082262538287\\
3.25	0.992194435997114\\
3.5	0.99496037607663\\
3.75	0.996922925279538\\
4.5	0.999661846129699\\
5	1\\
8	1\\
};
\addlegendentry{SRR, Alg. \ref{alg:SubORR_alg}}

\addplot [color=mycolor3, line width=3.0pt]
  table[row sep=crcr]{%
0.25	0.375069726607423\\
0.5	0.55062448215695\\
0.75	0.66222819961054\\
1	0.724429462407237\\
1.25	0.779462052423645\\
1.5	0.816894611357032\\
1.75	0.844028427486224\\
2	0.862856400252017\\
2.25	0.888336422804167\\
2.5	0.899944903793466\\
2.75	0.915114085715809\\
3	0.927480832578217\\
3.25	0.937667176332129\\
3.5	0.945554834337198\\
3.75	0.951861099921931\\
4	0.95911445961378\\
4.25	0.9670118219186\\
4.5	0.972418864870395\\
4.75	0.973818684844646\\
5	0.975561162803821\\
5.25	0.97953620285276\\
5.5	0.981050704048178\\
5.75	0.98340419868234\\
6	0.985220785442788\\
6.25	0.987365765336435\\
6.5	0.988561140688473\\
6.75	0.989443291556894\\
7	0.991951996991677\\
7.75	0.994596128804478\\
8	0.99485073780698\\
};
\addlegendentry{Subset merging, Alg. \ref{alg:Wsubset algo}}

%\addplot [color=mycolor4, dashdotted, line width=2.0pt]
%  table[row sep=crcr]{%
%1	0.051280104812145\\
%1.25	0.0731183761754597\\
%1.5	0.0984346709047017\\
%1.75	0.126863669363612\\
%2	0.158025814302377\\
%2.25	0.191521333595285\\
%2.5	0.22693195546487\\
%2.75	0.263827601166703\\
%3	0.301773996364254\\
%4	0.455860164632433\\
%4.25	0.493144351404563\\
%4.5	0.529342151638795\\
%4.75	0.564234884169577\\
%5	0.59764899938866\\
%5.25	0.629454250703523\\
%5.5	0.65956065354402\\
%5.75	0.68791459337892\\
%6	0.714494421841835\\
%6.25	0.739305836430548\\
%6.5	0.762377284372878\\
%6.75	0.783755573312639\\
%7	0.803501893996367\\
%7.25	0.821687866798191\\
%7.5	0.838392783636928\\
%7.75	0.853700703096651\\
%8	0.867698076136557\\
%};
%\addlegendentry{GRR}

\end{axis}
% \begin{axis}[%
% width=5.833in,
% height=4.375in,
% at={(0in,0in)},
% scale only axis,
% xmin=0,
% xmax=1,
% ymin=0,
% ymax=1,
% axis line style={draw=none},
% ticks=none,
% axis x line*=bottom,
% axis y line*=left
% ]
% \end{axis}
\end{tikzpicture}%
}
                \caption{ \centering \label{fig:SRRutility} Utility \vspace{10pt}}
            \end{subfigure}%
            ~
            \centering
            \begin{subfigure}{0.49\textwidth}
            \centering
                \scalebox{0.45}{\definecolor{mycolor1}{rgb}{0.25, 0.4, 0.96}
\definecolor{mycolor2}{rgb}{0.63, 0.36, 0.94}
\definecolor{mycolor3}{rgb}{0.89, 0.26, 0.2}
\definecolor{mycolor4}{rgb}{0.24, 0.82, 0.44}

\begin{tikzpicture}

\begin{axis}[%
width=4.521in,
height=3.8in,
at={(0.7in,0.519in)},
scale only axis,
xmin=0,
xmax=8,
xlabel style={font=\color{white!15!black}},
xlabel={\Huge $\eps$},
ymin=0,
ymax=4,
ylabel style={font=\color{white!15!black}},
ylabel={\huge $|\log \min_{y}\Psi(y)|$},
axis background/.style={fill=white},
xmajorgrids,
ymajorgrids,
yticklabel style = {font=\huge},
xticklabel style = {font=\huge},
legend style={legend cell align=left, align=left, draw=white!15!black, legend pos= south east, font=\huge}
]
\addplot [color=mycolor1, line width=3.0pt]
  table[row sep=crcr]{%
0.25	0.162500000000037\\
0.5	0.325000000000013\\
0.75	0.487500000000008\\
1	0.650000000000001\\
1.25	0.8125\\
1.5	0.975\\
1.75	1.1375\\
2	1.3\\
2.25	1.4625\\
2.5	1.62499999999999\\
2.75	1.78749999999999\\
3	1.94999999999999\\
3.25	2.11249999999998\\
3.5	2.27499999999998\\
3.75	2.43471721148449\\
4	2.58954513641665\\
4.25	2.73879125828602\\
4.5	2.87491769056624\\
4.75	3.00329269056623\\
5	3.12952054313678\\
5.25	3.24543700583397\\
5.5	3.34824061749863\\
5.75	3.4438840273829\\
6	3.5341698512827\\
6.25	3.61357594985504\\
6.5	3.68518257808418\\
6.75	3.7527437825692\\
7	3.80867732915153\\
7.25	3.85586462820088\\
7.5	3.90030273650351\\
7.75	3.93806669377351\\
8	3.97312037189684\\
};
\addlegendentry{AORR}

\addplot [color=mycolor2, line width=3.0pt]
  table[row sep=crcr]{%
0.25	0.162500000000011\\
0.5	0.325000000000005\\
0.75	0.487500000000006\\
1	0.650000000000004\\
1.25	0.812500000000006\\
1.5	0.975000000000007\\
1.75	1.13750000000001\\
2	1.29939649318202\\
2.25	1.46101457297716\\
2.5	1.61870027439815\\
2.75	1.78382151925962\\
3	1.91283526010047\\
3.25	2.06927623973817\\
3.5	2.22151736993071\\
3.75	2.33147407546978\\
4	2.44712231780923\\
4.25	2.59064600170036\\
4.5	2.7207263885312\\
4.75	2.75840079124335\\
5	2.78905167337261\\
5.25	2.89640487578948\\
5.5	2.95942685287015\\
5.75	3.05592155368579\\
6	3.11802228904224\\
6.25	3.25035838521653\\
6.5	3.30953059884078\\
6.75	3.35765017922505\\
7	3.5298125655573\\
7.25	3.56997592550104\\
7.5	3.61336740105415\\
7.75	3.63707971316172\\
8	3.66244853020805\\
};
\addlegendentry{SRR, Alg. \ref{alg:SubORR_alg}}

\addplot [color=mycolor3, line width=3.0pt]
  table[row sep=crcr]{%
0.25	0.102966993912121\\
0.5	0.227830533556745\\
0.75	0.36703737243072\\
1	0.501680187977861\\
1.25	0.691218315367897\\
1.5	0.843337881434739\\
1.75	1.00455957225558\\
2	1.16828326744382\\
2.25	1.32116186415741\\
2.5	1.51366194158149\\
2.75	1.64944350564708\\
3	1.7593574741322\\
3.25	1.93348395948459\\
3.5	2.10684081064188\\
3.75	2.26018197388352\\
4	2.3847957130463\\
4.25	2.54041715633442\\
4.5	2.67588459732123\\
4.75	2.72661509893642\\
5	2.78905167337261\\
5.25	2.89640487578948\\
5.5	2.95942685287015\\
5.75	3.05592155368579\\
6	3.11802228904224\\
6.25	3.25035838521653\\
6.5	3.30953059884078\\
6.75	3.35765017922505\\
7	3.5298125655573\\
7.25	3.56997592550104\\
7.5	3.61336740105415\\
7.75	3.63707971316172\\
8	3.66244853020805\\
};
\addlegendentry{Subset merging, Alg. \ref{alg:Wsubset algo}}

%\addplot [color=mycolor4, dashdotted, line width=2.0pt]
%  table[row sep=crcr]{%
%1	0.18403265901282\\
%1.25	0.238201004010882\\
%1.5	0.295622315753043\\
%1.75	0.356278903781149\\
%2	0.420158674115312\\
%2.25	0.487082676507047\\
%2.5	0.556931563281518\\
%2.75	0.6295477854075\\
%3	0.704857914547294\\
%3.25	0.783104370723265\\
%3.5	0.863908068758507\\
%3.75	0.947058514426082\\
%4	1.03208856959205\\
%4.5	1.20700005224493\\
%5	1.38730209172856\\
%5.5	1.57156693819325\\
%6	1.75967265754063\\
%7.5	2.32826561262184\\
%8	2.51495054385077\\
%};
%\addlegendentry{GRR}

\end{axis}

% \begin{axis}[%
% width=5.833in,
% height=4.375in,
% at={(0in,0in)},
% scale only axis,
% xmin=0,
% xmax=1,
% ymin=0,
% ymax=1,
% axis line style={draw=none},
% ticks=none,
% axis x line*=bottom,
% axis y line*=left
% ]
% \end{axis}
\end{tikzpicture}%
}
                \caption{ \centering \label{fig:SRRleakL} Min-lift leakage \hspace{5pt} \vspace{10pt}}
            \end{subfigure}
            ~\\
            \centering
            \begin{subfigure}{0.49\textwidth}
            \centering
                \scalebox{0.45}{\definecolor{mycolor1}{rgb}{0.25, 0.4, 0.96}
\definecolor{mycolor2}{rgb}{0.63, 0.36, 0.94}
\definecolor{mycolor3}{rgb}{0.89, 0.26, 0.2}
\definecolor{mycolor4}{rgb}{0.24, 0.82, 0.44}

\begin{tikzpicture}

\begin{axis}[%
width=4.521in,
height=3.8in,
at={(0.7in,0.519in)},
scale only axis,
xmin=0,
xmax=8,
xlabel style={font=\color{white!15!black}},
xlabel={\Huge $\eps$},
ymin=0,
ymax=1,
ylabel style={font=\color{white!15!black}},
ylabel={\huge $\log \max_{y}\Lambda(y)|$},
axis background/.style={fill=white},
xmajorgrids,
ymajorgrids,
yticklabel style = {font=\huge},
xticklabel style = {font=\huge},
legend style={legend cell align=left, align=left, draw=white!15!black, legend pos= south east, font=\huge}
]
\addplot [color=mycolor1, line width=3.0pt]
  table[row sep=crcr]{%
0.25	0.087500000000049\\
0.5	0.175000000000047\\
0.75	0.262500000000003\\
1	0.350000000000001\\
1.25	0.4375\\
1.5	0.524891550056635\\
1.75	0.611955819908446\\
2	0.697044763692553\\
2.25	0.776708106479431\\
2.5	0.842622881445377\\
2.75	0.888107628319311\\
3	0.917636003069557\\
3.25	0.934743213602704\\
3.5	0.943632289621734\\
3.75	0.9493259779194\\
4	0.952050595158794\\
4.25	0.953580230922403\\
4.5	0.954539225411719\\
4.75	0.955226285669905\\
5	0.955747743651603\\
5.25	0.95615211159145\\
5.5	0.956431152243903\\
5.75	0.956624462151239\\
6	0.9567591320998\\
6.25	0.956852393896241\\
6.5	0.956921397740421\\
6.75	0.956986024731988\\
7	0.957027586761456\\
7.25	0.957055984828602\\
7.5	0.957077972958827\\
7.75	0.957096656032657\\
8	0.957112531754589\\
};
\addlegendentry{AORR}

\addplot [color=mycolor2, line width=3.0pt]
  table[row sep=crcr]{%
0.25	0.0875000000000049\\
0.5	0.175000000000005\\
0.75	0.262500000000005\\
1	0.350000000000002\\
1.25	0.437500000000002\\
1.5	0.524905941238803\\
1.75	0.611397641897021\\
2	0.695807482561611\\
2.25	0.773894040927237\\
2.5	0.829138750536192\\
2.75	0.868493113761874\\
3	0.893824939795973\\
3.25	0.907126776532967\\
3.5	0.913850535876207\\
3.75	0.915933787948166\\
4	0.917909971590651\\
4.25	0.912406894137282\\
4.5	0.916012124747183\\
4.75	0.916329734898339\\
5	0.9192140585709\\
5.25	0.92646086780049\\
5.5	0.932370847482684\\
5.75	0.94141045098295\\
6	0.938620878330571\\
6.25	0.943723422675869\\
6.5	0.944662879659916\\
6.75	0.943485597426638\\
7	0.944957118454611\\
7.25	0.946077326459732\\
7.5	0.946077326459732\\
7.75	0.946077326459732\\
8	0.946077326459732\\
};
\addlegendentry{SRR, Alg. \ref{alg:SubORR_alg}}

\addplot [color=mycolor3, line width=3.0pt]
  table[row sep=crcr]{%
0.25	0.073170380848692\\
0.5	0.154892444266512\\
0.75	0.238472328333857\\
1	0.314414729128172\\
1.25	0.398507425078993\\
1.5	0.474905548768955\\
1.75	0.541559778474053\\
2	0.618380124590266\\
2.25	0.677596110445476\\
2.5	0.745134067208641\\
2.75	0.787715132559113\\
3	0.819201143188403\\
3.25	0.851159968523688\\
3.5	0.869820852155432\\
3.75	0.880653651417382\\
4	0.893258273698844\\
4.25	0.908817554145816\\
4.5	0.913787454441923\\
4.75	0.914459534621091\\
5	0.9192140585709\\
5.25	0.92646086780049\\
5.5	0.932370847482684\\
5.75	0.94141045098295\\
6	0.938620878330571\\
6.25	0.943723422675869\\
6.5	0.944662879659916\\
6.75	0.943485597426638\\
7	0.944957118454611\\
7.25	0.946077326459732\\
7.5	0.946077326459732\\
7.75	0.946077326459732\\
8	0.946077326459732\\
};
\addlegendentry{Subset merging, Alg. \ref{alg:Wsubset algo}}

%\addplot [color=mycolor4, dashdotted, line width=2.0pt]
%  table[row sep=crcr]{%
%1	0.183808896959627\\
%1.25	0.229394808449822\\
%1.5	0.274360317558516\\
%1.75	0.318552686178583\\
%2	0.361692306519679\\
%2.25	0.403516194912667\\
%2.5	0.443991612767556\\
%2.75	0.48304018299039\\
%3	0.520375769358644\\
%3.25	0.555895610229188\\
%3.5	0.589632234509148\\
%3.75	0.621524637397\\
%4	0.651450213912792\\
%4.25	0.679489129496133\\
%4.5	0.705672166056743\\
%4.75	0.730130071423977\\
%5	0.752764861877381\\
%5.25	0.773581299934007\\
%5.5	0.792651394504817\\
%5.75	0.810034310693638\\
%6	0.825829764978113\\
%6.25	0.840165818297162\\
%6.5	0.853133411288141\\
%6.75	0.86492419009744\\
%7	0.875590411486344\\
%7.25	0.885175042508353\\
%7.5	0.893768071513843\\
%7.75	0.901459137602259\\
%8	0.908348983810864\\
%};
%\addlegendentry{GRR}

\end{axis}

% \begin{axis}[%
% width=5.833in,
% height=4.375in,
% at={(0in,0in)},
% scale only axis,
% xmin=0,
% xmax=1,
% ymin=0,
% ymax=1,
% axis line style={draw=none},
% ticks=none,
% axis x line*=bottom,
% axis y line*=left
% ]
% \end{axis}
\end{tikzpicture}%
}
                \caption{ \centering \label{fig:SRRleakU} Max-lift leakage}
            \end{subfigure}
            ~
            \centering
            \begin{subfigure}{0.49\textwidth}
            \centering
                \scalebox{0.45}{\definecolor{mycolor1}{rgb}{0.25, 0.4, 0.96}
\definecolor{mycolor2}{rgb}{0.63, 0.36, 0.94}
\definecolor{mycolor3}{rgb}{0.89, 0.26, 0.2}
\definecolor{mycolor4}{rgb}{0.24, 0.82, 0.44}

\begin{tikzpicture}

\begin{axis}[%
width=4.521in,
height=3.8in,
at={(0.7in,0.519in)},
scale only axis,
xmin=0,
xmax=8,
xlabel style={font=\color{white!15!black}},
xlabel={\Huge $\eps$},
ymin=0,
ymax=0.9,
ylabel style={font=\color{white!15!black}},
ylabel={\huge time},
axis background/.style={fill=white},
xmajorgrids,
ymajorgrids,
yticklabel style = {font=\huge},
xticklabel style = {font=\huge},
legend style={legend cell align=left, align=left, draw=white!15!black, legend pos= north east, font=\huge}
]
\addplot [color=mycolor1, line width=3.0pt]
  table[row sep=crcr]{%
0.25	0.896996666\\
0.5	0.629235106\\
0.75	0.193455181\\
1	0.089710446\\
1.25	0.067817675\\
1.5	0.060566753\\
1.75	0.056938433\\
2	0.052876566\\
2.25	0.051747014\\
2.5	0.04778002\\
2.75	0.049007593\\
3	0.047387334\\
3.25	0.046956072\\
3.5	0.047986258\\
3.75	0.045602632\\
4	0.045426196\\
4.25	0.04623482\\
4.5	0.044423005\\
4.75	0.045616451\\
5	0.0455143\\
5.25	0.043587506\\
5.5	0.043921725\\
5.75	0.043749329\\
6	0.041951884\\
6.25	0.043907239\\
6.5	0.043883084\\
6.75	0.044013504\\
7	0.044516646\\
7.25	0.04324758\\
7.5	0.043337581\\
7.75	0.042872931\\
8	0.042293096\\
};
\addlegendentry{AORR}

\addplot [color=mycolor2, line width=3.0pt]
  table[row sep=crcr]{%
0.25	0.104376537\\
0.5	0.117439956\\
0.75	0.131011831\\
1	0.129131191\\
1.25	0.117804496\\
1.5	0.105747474\\
1.75	0.092775913\\
2	0.073774884\\
2.25	0.065202882\\
2.5	0.048362995\\
2.75	0.038470048\\
3	0.032525698\\
3.25	0.0231211\\
3.5	0.017981445\\
3.75	0.011619503\\
4	0.009135415\\
4.25	0.006249278\\
4.5	0.00378412\\
4.75	0.001978283\\
5	0.000856727\\
5.25	0.000767094\\
5.5	0.000682084\\
5.75	0.000671262\\
6	0.000630855\\
6.25	0.000564739\\
6.5	0.000548451\\
6.75	0.000554496\\
7	0.000491832\\
7.25	0.000527766\\
7.5	0.000451279\\
7.75	0.000420443\\
8	0.000390076\\
};
\addlegendentry{SRR, Alg. \ref{alg:SubORR_alg}}

\addplot [color=mycolor3, line width=3.0pt]
table[row sep=crcr]{%
	0.25	0.004935978\\
0.5	0.004751829\\
0.75	0.005273733\\
1	0.004714022\\
1.25	0.004053355\\
1.5	0.003669776\\
1.75	0.003289998\\
2	0.002904846\\
2.25	0.002507314\\
2.5	0.002264415\\
2.75	0.002176971\\
3	0.001874606\\
3.25	0.001579613\\
3.5	0.001650929\\
3.75	0.00133566\\
4	0.001161271\\
4.25	0.000902516\\
4.5	0.000841592\\
4.75	0.000848988\\
5	0.000763846\\
5.25	0.000783728\\
5.5	0.000626307\\
5.75	0.000639369\\
6	0.00054107\\
6.25	0.000478734\\
6.5	0.000490566\\
6.75	0.00049892\\
7	0.000439102\\
7.25	0.000452643\\
7.5	0.00038184\\
7.75	0.000380123\\
8	0.000349455\\
};
\addlegendentry{Subset merging, Alg. \ref{alg:Wsubset algo}}

\end{axis}
\end{tikzpicture}%
}
                \caption{ \centering \label{fig:ORRtime} Time Complexity}
            \end{subfigure}
            \caption{ \label{Fig:SubsetORR} Comparison of privacy-utility tradeoff and time complexity between AORR, SRR, and subset merging, where $|\X|=17$, $|\Sen|=5$, $\eps_{\text{LDP}}\in\{0.25,0.5,0.75,\cdots,8\}$, $\lambda=0.65$, $\epsl=\lambda\eps_{\text{LDP}}$, and $\epsu=(1-\lambda)\eps_{\text{LDP}}$.}
        \end{figure}
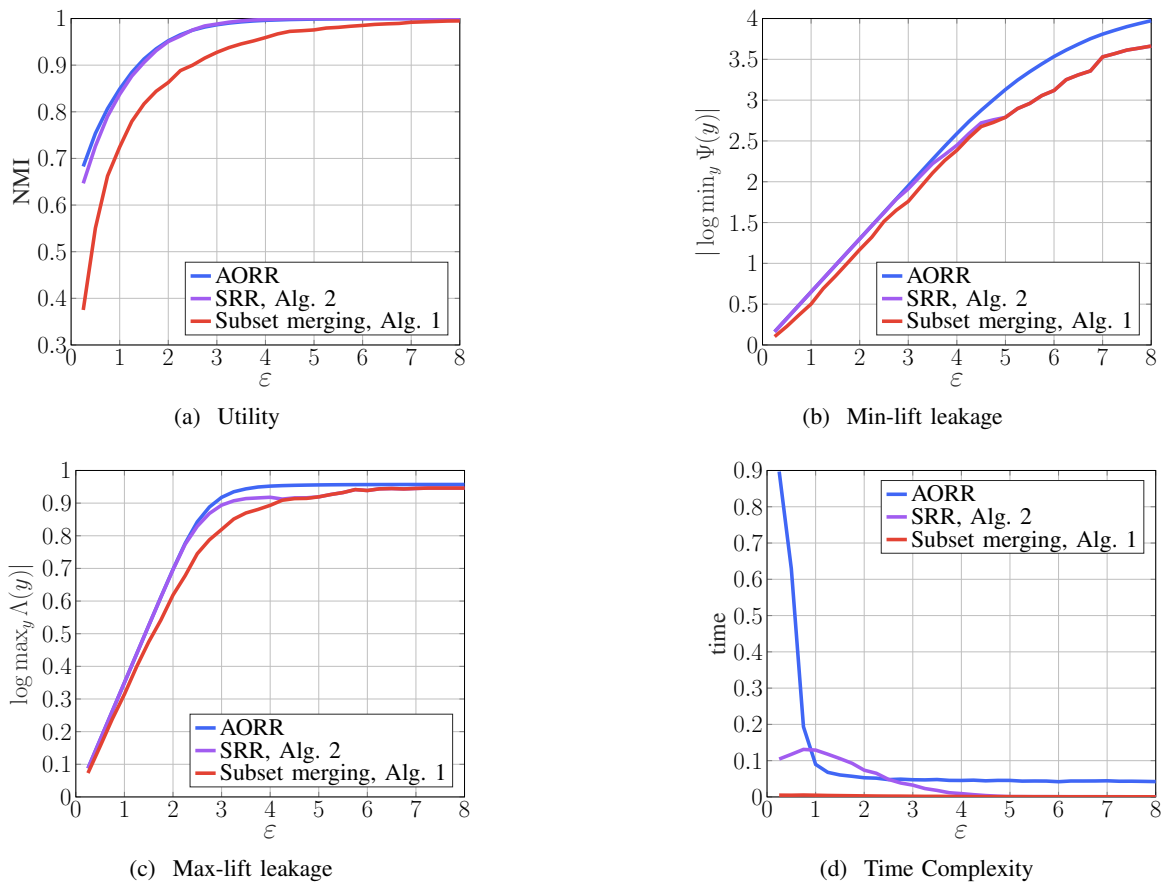
        
         \begin{figure}
            \begin{subfigure}{0.32\textwidth}
                \centering
                \scalebox{0.4}{\definecolor{mycolor1}{rgb}{0.25, 0.4, 0.96}
\definecolor{mycolor2}{rgb}{0.63, 0.36, 0.94}
\definecolor{mycolor3}{rgb}{0.89, 0.26, 0.2}
\definecolor{mycolor4}{rgb}{0.24, 0.82, 0.44}

\begin{tikzpicture}

\begin{axis}[%
width=4.521in,
height=3.8in,
scale only axis,
xmin=1,
xmax=8,
xlabel style={font=\color{white!15!black}},
xlabel={\Huge $\eps$},
ymin=0.8,
ymax=1,
ylabel style={font=\color{white!15!black}},
ylabel={\huge NMI},
axis background/.style={fill=white},
xmajorgrids,
ymajorgrids,
yticklabel style = {font=\huge},
xticklabel style = {font=\huge},
legend style={legend cell align=left, align=left, draw=white!15!black, legend pos= south east, font=\huge}
]
\addplot [color=mycolor2,  line width=3.0pt]
  table[row sep=crcr]{%
1	0.850455841691415\\
1.25	0.875518003628232\\
1.5	0.889766060196548\\
1.75	0.901267169317267\\
2	0.912568294422881\\
2.25	0.923598293389988\\
2.5	0.934047714749097\\
2.75	0.943843951719623\\
3	0.952569093255015\\
3.25	0.960301202897293\\
3.5	0.966828043538021\\
3.75	0.972509238045292\\
4	0.97740506899463\\
4.25	0.981520767796702\\
4.5	0.985020727992284\\
4.75	0.987846343173446\\
5	0.990256076235523\\
5.25	0.992226274751514\\
5.5	0.993785928609828\\
5.75	0.995069160174845\\
6	0.996080531018379\\
6.25	0.99695971222131\\
6.5	0.997648820730683\\
6.75	0.998164400820407\\
7	0.998589050782426\\
7.25	0.998919970498864\\
7.75	0.999391421747349\\
8	0.999544124601625\\
};
\addlegendentry{SRR, Alg. \ref{alg:SubORR_alg}}

\addplot [color=mycolor3, line width=3.0pt]
  table[row sep=crcr]{%
1	0.806421822259901\\
1.25	0.84067423044495\\
1.5	0.858330096358376\\
1.75	0.865949830264077\\
2.25	0.875527158369925\\
2.5	0.880233934761566\\
2.75	0.885145822030296\\
3	0.890983196331323\\
3.25	0.89736324998494\\
3.5	0.903491163249701\\
3.75	0.910586440501577\\
4	0.917255556225866\\
4.25	0.923558115562271\\
4.5	0.929749672820401\\
5	0.941605398677771\\
5.25	0.946889096614575\\
5.5	0.95186554931011\\
5.75	0.95665147560924\\
6	0.960969367605129\\
6.25	0.964730254904506\\
6.5	0.968595873973367\\
7	0.974737914369134\\
7.5	0.979818507470029\\
7.75	0.981956043856014\\
8	0.983888711748062\\
};
\addlegendentry{Subset merging, Alg. \ref{alg:Wsubset algo}}

\end{axis}

% \begin{axis}[%
% width=5.833in,
% height=4.375in,
% at={(0in,0in)},
% scale only axis,
% xmin=0,
% xmax=1,
% ymin=0,
% ymax=1,
% axis line style={draw=none},
% ticks=none,
% axis x line*=bottom,
% axis y line*=left
% ]
% \end{axis}
\end{tikzpicture}%
}
                \caption{ \centering \label{fig:SRRlargeUt} Utility}
            \end{subfigure}%
            ~
            \begin{subfigure}{0.32\textwidth}
                \scalebox{0.4}{ \definecolor{mycolor1}{rgb}{0.25, 0.4, 0.96}
\definecolor{mycolor2}{rgb}{0.63, 0.36, 0.94}
\definecolor{mycolor3}{rgb}{0.89, 0.26, 0.2}
\definecolor{mycolor4}{rgb}{0.24, 0.82, 0.44}

\begin{tikzpicture}

\begin{axis}[%
width=4.521in,
height=3.8in,
scale only axis,
xmin=1,
xmax=8,
xlabel style={font=\color{white!15!black}},
xlabel={\Huge $\eps$},
ymin=0.5,
ymax=4,
ylabel style={font=\color{white!15!black}},
ylabel={\huge $|\log \min_{y}\Psi(y)|$},
axis background/.style={fill=white},
xmajorgrids,
ymajorgrids,
yticklabel style = {font=\huge},
xticklabel style = {font=\huge},
legend style={legend cell align=left, align=left, draw=white!15!black, legend pos= south east, font=\huge}
]
\addplot [color=mycolor2, line width=3.0pt]
  table[row sep=crcr]{%
1	0.499999999549445\\
1.25	0.625000000000032\\
1.5	0.749999999962744\\
1.75	0.875000000000035\\
2	1.00000000000003\\
2.25	1.12500000000007\\
2.5	1.25000000000004\\
2.75	1.37500000000004\\
3	1.50000000000005\\
3.25	1.62500000000006\\
3.5	1.75000000651003\\
3.75	1.87500000000008\\
4	2.00000000000008\\
4.25	2.1250000000001\\
4.5	2.25000000000011\\
4.75	2.37500000000013\\
5	2.50000000000015\\
5.25	2.62500000000015\\
5.5	2.75000000000018\\
5.75	2.87500000000019\\
6	3.00000000000022\\
6.25	3.12500000000027\\
6.5	3.25000000000028\\
6.75	3.37500000000031\\
7	3.50000000000033\\
7.25	3.62500000000039\\
7.5	3.75000000000044\\
7.75	3.87500000000052\\
8	4.00000000000055\\
};
\addlegendentry{SRR, Alg. \ref{alg:SubORR_alg}}

\addplot [color=mycolor3, line width=3.0pt]
  table[row sep=crcr]{%
1	0.49495384545389\\
1.25	0.619363917882579\\
1.5	0.738432465727112\\
1.75	0.856527595090877\\
2	0.976032978452683\\
2.25	1.10429638608885\\
2.5	1.23338178874199\\
2.75	1.35980646012204\\
3	1.4858538528397\\
3.25	1.6108335590179\\
3.5	1.73707285009533\\
3.75	1.86452559094269\\
4	1.98893373990014\\
4.25	2.11264203050312\\
4.5	2.23507509717078\\
4.75	2.36065605034735\\
5	2.4837159646212\\
5.25	2.60986413975428\\
5.5	2.73187826819187\\
5.75	2.85948091070676\\
6	2.9796278188317\\
6.25	3.10524540266104\\
6.5	3.22935402215338\\
6.75	3.34818339962099\\
7	3.4685550419122\\
7.25	3.60063510609819\\
7.5	3.72027125122822\\
7.75	3.84178957238658\\
8	3.95310057183913\\
};
\addlegendentry{Subset merging, Alg. \ref{alg:Wsubset algo}}

\end{axis}
\end{tikzpicture}%
}
                \caption{\centering \label{fig:SRRlargeL}Min-lift leakage}
            \end{subfigure}
            ~
            \begin{subfigure}{0.32\textwidth}
                \scalebox{0.4}{\definecolor{mycolor1}{rgb}{0.25, 0.4, 0.96}
\definecolor{mycolor2}{rgb}{0.63, 0.36, 0.94}
\definecolor{mycolor3}{rgb}{0.89, 0.26, 0.2}
\definecolor{mycolor4}{rgb}{0.24, 0.82, 0.44}
\definecolor{mycolor5}{rgb}{0.0, 1.0, 1.0}
\begin{tikzpicture}

\begin{axis}[%
width=4.521in,
height=3.8in,
scale only axis,
xmin=1,
xmax=8,
xlabel style={font=\color{white!15!black}},
xlabel={\Huge $\eps$},
ymin=0.5,
ymax=1.1,
ylabel style={font=\color{white!15!black}},
ylabel={\huge $\log (\max_{y}\Lambda(y))$},
axis background/.style={fill=white},
xmajorgrids,
ymajorgrids,
yticklabel style = {font=\huge},
xticklabel style = {font=\huge},
legend style={legend cell align=left, align=left, draw=white!15!black, legend pos= south east,font=\huge}
]
\addplot [color=mycolor2, line width=3.0pt]
  table[row sep=crcr]{%
1	0.5000000106006\\
1.25	0.625000000000004\\
1.5	0.743608687302251\\
1.75	0.822229707544961\\
2	0.866438456098766\\
2.25	0.907153245148754\\
2.75	0.96543307936656\\
3	0.984400267982556\\
3.25	1.00440894790473\\
3.5	1.01952893368105\\
3.75	1.0326224663329\\
4	1.0435423396583\\
4.25	1.05402418976793\\
4.5	1.06170164867442\\
4.75	1.06795431669209\\
5	1.07312664256011\\
6	1.08612375233491\\
6.5	1.09057412555199\\
6.75	1.09173126260842\\
7.25	1.09349323743867\\
8	1.09527412214498\\
};
\addlegendentry{SRR, Alg. \ref{alg:SubORR_alg}}

\addplot [color=mycolor3, line width=3.0pt]
  table[row sep=crcr]{%
1	0.487514932069866\\
1.25	0.589846747229503\\
1.5	0.641091759526111\\
1.75	0.666508774589841\\
2	0.701612337459414\\
2.25	0.750127566553301\\
2.75	0.827284955261362\\
3	0.864937210540372\\
3.25	0.88411395800162\\
3.5	0.908644945339987\\
3.75	0.932486011029567\\
4	0.941822073478072\\
4.25	0.964303623915528\\
4.5	0.983656639015912\\
4.75	0.994863034004466\\
5	1.00368388621285\\
5.25	1.0162322888367\\
5.5	1.01813145609581\\
5.75	1.02576681316461\\
6	1.02805366320394\\
6.25	1.0405309577853\\
6.75	1.04745700444504\\
7	1.05302149882745\\
7.25	1.05706511852123\\
7.5	1.06574159726584\\
7.75	1.06890251168784\\
8	1.06892784646823\\
};
\addlegendentry{Subset merging, Alg. \ref{alg:Wsubset algo}}

\end{axis}
\end{tikzpicture}%
}
                \caption{ \centering \label{fig:SRRlargeU}Max-lift leakage \vspace{5pt}}
            \end{subfigure}
            \caption{ \label{Fig:SRRlarge} Comparison of privacy-utility tradeoff and time complexity between SRR and subset merging, where $|\X|=200$, $|\Sen|=15$, $\eps_{\text{LDP}}\in\{1,1.25,1.5,1.75,\cdots,8\}$, and $\epsl=\epsu=\frac{\eps_{\text{LDP}}}{2}$.}
        \end{figure}

%=================================================================
%=================================================================

\section{Lift-Based and Lift-Inverse Measures}\label{sec:lift-basedmeasures}

    In this section, we consider some recently proposed privacy measures that quantify the divergence between the posterior and prior belief on sensitive features, including $\ell_{1}$-norm \cite{2022DataDsclsurell1Priv}, strong $\chi^2$-privacy criterion \cite{2021StrongChi2}, and $\alpha$-lift \cite{2021Alpha-LiftWatchdog}. 
    They have been proposed as a stronger version of their corresponding average measures, which are the total variation distance \cite{2019PUTTotalDistnce}, $\chi^2$-divergence \cite{2019PrivEstimGuarant}, and Sibson MI \cite{2020MaxL}, respectively.
    We call them \textit{lift-based} measures and define them in the following.
    
    \begin{Definition}\label{def:lift-based}
        For each $y \in \Y$, lift-based privacy measures are defined as follows:
        \begin{itemize}
            \item The $\ell_{1}$-lift is given by
            \begin{equation}\label{eq:total lift}
                \Lambda_{\ell_{1}}(y)\triangleq\sum_{s \in \Sen}\PS(s) |\lifty-1|,
            \end{equation}
            then the total variation distance will be $$T(S;Y)=\frac{1}{2}\E_{Y}[\Lambda_{\ell_{1}}(Y)].$$
            
            \item The $\chi^2$-lift is given by
            \begin{equation}\label{eq:chi2 lift}
                \Lambda_{\chi^2}(y) \triangleq \sum_{s \in \Sen}\PS(s)\big(\lifty-1\big)^2,
            \end{equation}
            then $\chi^2$-divergence will be
            \begin{equation*}
                \chi^2(S;Y)=\E_{Y}[\Lambda_{\chi^2}(Y)].
            \end{equation*}
            
            \item The $\alpha$-lift is given by
            \begin{equation}\label{eq:sibson lift}
                \Lambda_{\alpha}^{S}(y) \triangleq\left(\sum_{s \in\Sen}\PS(s)\lifty^{\alpha}\right)^{1/\alpha},
            \end{equation}
            then Sibson MI will be
            \begin{equation*}
                I_{\alpha}^{S}(S;Y)=\frac{\alpha}{\alpha-1} \log \E_{Y}[\Lambda_{\alpha}^{S}(Y)].
            \end{equation*}
        \end{itemize}
    \end{Definition} 
    Here, we reveal their relationship with ALIP.
    
    \begin{Proposition} \label{prop:ALIP-lift-base} If $(\epsl,\epsu)$-ALIP is satisfied, then,
        \begin{enumerate}
            \item \label{subprop:ALIP total-lift}$\displaystyle \max_{y \in \Y}\Lambda_{\ell_{1}}(y) \leq \e^{\epsu}-1$,
    
            \item \label{subprop:ALIP chi2-lift}$\displaystyle\max_{y \in \Y}\Lambda_{\chi^2}(y) \leq (\e^{\epsu}-1)^2$,
    
            \item\label{subprop:ALIP sibson-lift} $\displaystyle \max_{y \in \Y}\Lambda_{\alpha}^{S}(y) \leq \e^{\epsu}$.
        \end{enumerate}
    \end{Proposition}
    
    The proof is given in Appendix \ref{apx:proof of ALIP-lift-base}.
    \begin{Proposition}\label{prop:lif-base-ALIP}
        Let $\displaystyle {s_y}=\argmax_{s\in \Sen}[\lifty]$ and $\bar{y}=\argmax_{y \in \Y}\Lambda(y)$, then,
        \begin{enumerate}
            \item \label{prop:total ALIP} If $\displaystyle \max_{y \in \Y}\Lambda_{\ell_{1}}(y) \leq {\eps} \Rightarrow \max_{y \in \Y}\Lambda(y) \leq {{\eps}}/{\PS(s_{\bar{y}})}+1$,
    
            \item \label{prop:chi2 ALIP} If $\displaystyle \max_{y \in \Y}\Lambda_{\chi^2}(y) \leq {\eps} \Rightarrow \max_{y \in \Y}\Lambda(y) \leq {\sqrt{{\eps}/ \PS(s_{\bar{y}})}}+1$,
    
            \item \label{prop:sibson ALIP} If $\displaystyle \max_{y \in \Y}\Lambda_{\alpha}^{S}(y) \leq {\eps} \Rightarrow \max_{y \in \Y}\Lambda(y) \leq {\eps}/{ \PS(s_{\bar{y}})^{\frac{1}{\alpha}} }$.
        \end{enumerate}
    \end{Proposition}
    \begin{proof}
        The proof is given in Appendix \ref{apx:proof of lif-base-ALIP}.
    \end{proof}
    
    Proposition \ref{prop:ALIP-lift-base} shows that lift-based measures, similar to their corresponding average leakages, are upper bounded by the max-lift bound.
    Proposition \ref{prop:lif-base-ALIP} indicates that if we bound lift-based measures, they can only restrict the max-lift leakage.
    Accordingly, if one only applies a lift-based measure to protect privacy, like previous works \cite{2022DataDsclsurell1Priv,2021StrongChi2,2021Alpha-LiftWatchdog}, it  may cause significant leakage on the min-lift.
    Therefore, in the following, we propose lift-inverse measures to bound the min-lift leakage.    
    
%*************************************************************************************************************************
    
    \subsection{Lift-Inverse Measures}\label{subsec:min-lift relaxation}
    
        In Propositions \ref{prop:ALIP-lift-base} and \ref{prop:lif-base-ALIP}, we have shown that lift-based measures only bound the max-lift.
        In this subsection, we present lift-inverse measures to restrict the min-lift leakage, $\min_{y \in \Y}\Psi(y)=\min_{y \in \Y}[\min_{s \in \Sen}\lifty$].

        \begin{Definition}\label{def:lift-inverse}
            For lift-based measures in \eqref{eq:total lift} to \eqref{eq:sibson lift}, we replace $\lifty$ with $\frac{1}{\lifty}$ and call the resulting quantities lift-inverse measures.
            \begin{itemize}
                \item The $\ell_{1}$-lift-inverse is given by
                \begin{equation}\label{eq:total inv lift}
                    \Psi_{\ell_{1}}(y)=\sum_{s \in 	\Sen}\PS(s)\left|\frac{1}{\ell(s,y)}-1\right|.
                \end{equation}
                \item
                The $\chi^2$-lift-inverse is given by
                \begin{equation}\label{eq:chi2 inv lift}
                    \Psi_{\chi^2}(y) \triangleq \sum_{s \in \Sen}\PS(s)\left(\frac{1}{\ell(s,y)}-1\right)^2.
                \end{equation}
                \item
                The $\alpha$-lift-inverse is given by
                \begin{equation}\label{eq:sibson inv lift}
                    \Psi_{\alpha}^{S}(y) \triangleq\left(\sum_{s\in\Sen}\PS(s)\left(\frac{1}{\ell(s,y)}\right)^{\alpha}\right)^{1/\alpha}.
                \end{equation}   
            \end{itemize}
        \end{Definition}

        In the following propositions, we show the relationship between $(\epsl,\epsu)$-ALIP.

        \begin{Proposition}\label{prop:ALIP min properties}
            If $(\epsl,\epsu)$-ALIP is achieved, we have
            \begin{enumerate}
                \item $\displaystyle \max_{y \in \Y}\Psi_{\ell_{1}}(y) \leq \e^{\epsl}-1,$\vspace{2pt}
    
                \item $\displaystyle \max_{y \in \Y}\Psi_{\chi^2}(y) \leq \left(\e^{\epsl}-1\right)^2,$\vspace{2pt}
    
                \item $\displaystyle \max_{y \in \Y}\Psi_{\alpha}^{S}(y) \leq \e^{\epsl}.$
            \end{enumerate}
        \end{Proposition}
        \begin{proof}
                The proof is provided in Appendix \ref{apx:proof of ALIP-lift-inverse}.
        \end{proof}
  
        \begin{Proposition}\label{prop:Inv ALIP min properties}
            Let $\displaystyle s_{y}=\argmin_{s}\lifty$ and  $\displaystyle \underline{y}=\argmin_{y}[\Psi(y)]$, then,
            \begin{enumerate}
                \item If $\displaystyle \max_{y \in \Y}\Psi_{\ell_{1}}(y) \leq {\eps} \Rightarrow \min_{y \in \Y}\Psi(y) \geq \frac{ \PS(s_{\underline{y}}) }{\eps+{\PS(s_{\underline{y}})}},$
            
                \item If $\displaystyle \max_{y \in \Y}\Psi_{\chi^2}(y) \leq {\eps} \Rightarrow \min_{y \in \Y}\Psi(y) \geq \frac{  \sqrt{ \PS(s_{\underline{y}}) }  }{  \sqrt{\eps} + \sqrt{\PS(s_{\underline{y}})}},$
            
                \item If $\displaystyle \max_{y \in \Y}\Psi_{\alpha}^{S}(y) \leq {\eps} \Rightarrow \min_{y \in \Y}\Psi(y) \geq
                {\eps ^{-1}}\PS(s_{\underline{y}})^{\frac{1}{\alpha}}.$
            \end{enumerate}
        \end{Proposition}
        \begin{proof}
            The proof is provided in Appendix \ref{apx:proof of lift-inverse-ALIP}.
        \end{proof}
        Propositions \ref{prop:ALIP min properties} and \ref{prop:Inv ALIP min properties} demonstrate that the   aforementioned lift-inverse measures are associated with the min-lift, and bounding them can restrict the min-lift leakage. 
        Since lift-based and lift-inverse measures quantify privacy leakage by a function of lift averaged over sensitive features, they can be regarded as more relaxed measures than the min and max lifts.       
%=================================================================
    \subsection{PUT and Numerical Results}
    
        Optimal randomization for $\ell_{1}$-lift and $\chi^2$-lift privacy which maximizes MI as the utility measure have been proposed in \cite{2022DataDsclsurell1Priv} and \cite{2021StrongChi2}, respectively.
        Note that ORR is not applicable to these measures since their privacy constraints are not convex.
        However, here, we apply the watchdog mechanism with subset merging randomization to investigate the PUT for lift-based and lift-inverse measures.
        This application shows that the watchdog mechanism with $X$-invariant randomization is a low-complexity method that can be applied to all the aforementioned measures.
        Moreover, our subset merging algorithm significantly enhances the utility, which is comparable to the optimal solutions.

        To apply lift-based and lift-inverse measures to the subset merging algorithm, we replace $\Lambda(y)$ and $\Psi(y)$ in \eqref{eq:ALIP XH} and Algorithm \ref{alg:Wsubset algo}, with the corresponding lift-based and lift-inverse measures in Definitions \ref{def:lift-based} and \ref{def:lift-inverse}, respectively.
        For example, $\XL$ and $\XH$ for the $\alpha$-lift are obtained as
        \begin{gather}
            \XL \triangleq \{x\in \X:  \Psi_{\alpha}^{S}(x) \leq \e^{\epsl} \quad and \quad \Lambda_{\alpha}^{S}(x) \leq \e^{\epsu}\} \quad and \quad \X_H=\X \setminus \XL \label{eq:ell1 XH}.
        \end{gather}
        The privacy risk measure in Algorithm \ref{alg:Wsubset algo} is also given by  $\omega(x)=\Psi_{\alpha}^{S}(x)+\Lambda_{\alpha}^{S}(x)$. 
        We compare the PUT of lift-based and lift-inverse privacy with ($\epsl,\epsu$)-ALIP, where lift-based and lift-inverse measures are bounded as follows:
          
        \begin{itemize}
        
            \item $\ell_{1}$-privacy: \hspace{26pt} $\Lambda_{\ell_{1}}(y) \hspace{1pt} \leq \e^{\epsu}-1$ \hspace{20pt} and $\quad \Psi_{\ell_{1}}(y) \leq \e^{\epsl}-1,  \hspace{24pt}  \forall y \in \Y.$
            
            \item $\chi^2$-privacy: \hspace{24pt} $\Lambda_{\chi^2}(y) \leq (\e^{\epsu} -1)^2$  \quad  and $\quad \Psi_{\chi^{2}}(y) \hspace{-1pt} \leq (\e^{\epsl}-1)^2, \quad \forall y \in \Y.$
            
            \item	
                $\alpha$-lift-privacy: \hspace{13pt} $\Lambda_{\alpha}^{S}(y) \hspace{4pt} \leq \e^{\epsu}$ \hspace{38pt} and $\quad \Psi_{\alpha}^{S}(y) \hspace{2pt} \leq \e^{\epsl}, \hspace{41pt}  \forall y \in \Y.$   
                
        \end{itemize}
        We apply the subset merging mechanism with the simulation setup in Section \ref{subsubsec:Num watchdog}.
        Figure~\ref{Fig:var-chi2-PUT} demonstrates the PUT of ALIP for $\lambda=0.5$ and $\ell_{1}$ and $\chi^2$ privacy for $\lambda=\{0.5,0.65\}$.
        When $\lambda=0.5$, $\ell_{1}$ and $\chi^2$ privacy result in higher utility compared to ALIP for all values of $\eps$ since lift-based and lift-inverse measures are relaxations of max and min lift.
        To observe the effect of the asymmetric scenario, we have depicted $\ell_{1}$ and $\chi^2$ privacy for $\lambda=0.65$. 
        From Figure \ref{fig:var-chi2-utility} we observe that lift-inverse relaxation ($\lambda=0.65$) enhances utility significantly for $\eps>1$, but worsens utility for $\eps<1$. 
        The reason is that when $\eps<1$, lift-inverse privacy constraint in the asymmetric scenario is strict, which requires more symbols to be merged in each subset and causes larger subsets and utility degradation.
        
        A comparison between $\alpha$-lift privacy and ALIP is shown in Figure \ref{Fig:alphaPUT} for $\alpha \in \{2,10,100\}$. 
        $\alpha$-lift privacy is tunable  such that when  $\alpha=\infty$, it is equivalent to ALIP, and when $\alpha < \infty$ it results in a relaxation scenario. 
        We observe tunable property in Figure \ref{Fig:alphaPUT} where $\alpha=2$ has the highest utility. Moreover, when $\alpha$ increases the PUT of $\alpha$-lift privacy becomes closer to ALIP.

        \begin{figure}
            \centering
            \begin{subfigure}{0.32\textwidth}
                \centering
                \scalebox{0.4}{\definecolor{mycolor1}{rgb}{0.25, 0.4, 0.96}
\definecolor{mycolor2}{rgb}{0.63, 0.36, 0.94}
\definecolor{mycolor3}{rgb}{0.89, 0.26, 0.2}
\definecolor{mycolor4}{rgb}{0.24, 0.82, 0.44}
\definecolor{mycolor5}{rgb}{0.0, 1.0, 1.0}

\begin{tikzpicture}

\begin{axis}[%
width=4.521in,
height=3.8in,
scale only axis,
xmin=0,
xmax=8,
xlabel style={font=\color{white!15!black}},
xlabel={\Huge $\eps$},
ymin=0.4,
ymax=1,
ylabel style={font=\color{white!15!black}},
ylabel={\huge NMI},
axis background/.style={fill=white},
xmajorgrids,
ymajorgrids,
yticklabel style = {font=\huge},
xticklabel style = {font=\huge},
legend style={legend cell align=left, align=left, draw=white!15!black, legend pos= south east, font=\huge}
]
\addplot [color=mycolor1, line width=2.0pt]
  table[row sep=crcr]{%
0.25	0.417654745366093\\
0.5	0.592257050805583\\
0.75	0.681241902233134\\
1	0.740016309300705\\
1.25	0.768862583869991\\
1.5	0.791272458796886\\
1.75	0.811035838783193\\
2	0.829382547789577\\
2.25	0.848437562052871\\
2.5	0.866885581518787\\
2.75	0.878728504666757\\
3	0.890551855483462\\
3.25	0.898571173191746\\
3.5	0.910151561564669\\
3.75	0.923986693676144\\
4	0.92956508482811\\
4.25	0.937102369381299\\
4.5	0.94604682133377\\
4.75	0.951178254612174\\
5	0.957393009191293\\
5.25	0.961127043238325\\
5.5	0.965132296046684\\
5.75	0.968615104358202\\
6	0.972721759597698\\
6.25	0.974973479353764\\
6.5	0.978184051641318\\
6.75	0.980029851606514\\
7	0.982817627703303\\
7.25	0.984568264063305\\
7.5	0.98611036303351\\
7.75	0.986427763345832\\
8	0.987582220371495\\
};
\addlegendentry{ALIP, $\lambda=0.5$}

\addplot [color=mycolor2, line width=2.0pt]
  table[row sep=crcr]{%
0.25	0.565556760925211\\
0.5	0.735612364454551\\
0.75	0.793437443737498\\
1	0.832707102553836\\
1.25	0.866983992956893\\
1.5	0.891265721051304\\
1.75	0.914333383153274\\
2	0.927712089101827\\
2.25	0.941726114773916\\
2.5	0.950197501427612\\
2.75	0.957473485450223\\
3	0.96336176203939\\
3.25	0.970590067859295\\
3.5	0.974773025985863\\
3.75	0.977628031408628\\
4	0.981642684347677\\
4.25	0.983650437788386\\
4.5	0.98673284911912\\
4.75	0.986872634982286\\
5	0.987435893260048\\
5.25	0.989965648479354\\
5.5	0.990973273232379\\
5.75	0.991606030027545\\
6	0.992821489504858\\
6.25	0.994000094817536\\
6.5	0.993999038020127\\
6.75	0.994472292024962\\
7	0.995359974718026\\
7.25	0.995895065416708\\
7.5	0.996650555425383\\
7.75	0.997577979403753\\
8	0.997577979403753\\
};
\addlegendentry{$\ell_{1}$-privacy, $\lambda=0.5$}

\addplot [color=mycolor3, line width=2.0pt]
  table[row sep=crcr]{%
0.25	0.72903701841118\\
0.5	0.762306961007692\\
0.75	0.791607398928067\\
1	0.813504189489105\\
1.25	0.829509262460675\\
1.5	0.853879073090185\\
1.75	0.870944338366318\\
2	0.884692204664483\\
2.25	0.902029883125375\\
2.5	0.913869847182242\\
2.75	0.92461224024455\\
3	0.934447647253577\\
3.25	0.942908070055578\\
3.5	0.95135296954762\\
3.75	0.957307605753574\\
4	0.963555489497173\\
4.25	0.968148368738287\\
4.5	0.971867404010328\\
4.75	0.975046187838625\\
5	0.978531705164625\\
5.25	0.980362459232936\\
5.5	0.98392995379014\\
5.75	0.985059587447133\\
6	0.986208010909648\\
6.25	0.987016340733066\\
6.5	0.988327698967681\\
6.75	0.989992013562148\\
7	0.991257287595391\\
7.25	0.992017522958639\\
7.5	0.992211975654739\\
7.75	0.993727881508334\\
8	0.994000094817536\\
};
\addlegendentry{$\chi^2$-privacy, $\lambda=0.5$}

\addplot [color=mycolor4, line width=2.0pt]
  table[row sep=crcr]{%
0.25	0.496860252032716\\
0.5	0.675667048699893\\
0.75	0.773734425826733\\
1	0.838375984485332\\
1.25	0.890072401717506\\
1.5	0.922695670431542\\
1.75	0.942218718026525\\
2	0.952270069914083\\
2.25	0.96213941972389\\
2.5	0.970590067859295\\
2.75	0.974975041279828\\
3	0.979404044913423\\
3.25	0.983099274291665\\
3.5	0.98673284911912\\
3.75	0.987211898165679\\
4	0.9896244209902\\
4.25	0.991257287595391\\
4.5	0.991995418474759\\
4.75	0.993625099217134\\
5	0.993999038020127\\
5.25	0.994472292024962\\
5.5	0.995668044426936\\
5.75	0.996650555425383\\
6	0.997577979403753\\
6.25	0.998126600636404\\
6.5	0.998126600636404\\
6.75	0.998619281683418\\
7	0.999158644741106\\
7.25	0.999461970760379\\
7.5	0.999461970760379\\
7.75	0.999461970760379\\
8	1\\
};
\addlegendentry{$\ell_{1}$-privacy, $\lambda=0.65$}

\addplot [color=mycolor5, line width=2.0pt]
  table[row sep=crcr]{%
0.25	0.727257949549047\\
0.5	0.756669826356137\\
0.75	0.790724884027143\\
1	0.822982001854083\\
1.25	0.858844772097305\\
1.5	0.875568779204209\\
1.75	0.8991303931312\\
2	0.918589621212485\\
2.25	0.931691383320187\\
2.5	0.942908070055578\\
2.75	0.953133707975874\\
3	0.962485578463769\\
3.25	0.967565951958611\\
3.5	0.971787463253408\\
3.75	0.97676103000395\\
4	0.979823568028879\\
4.25	0.98392995379014\\
4.5	0.985766364220702\\
4.75	0.986716889917067\\
5	0.988327698967681\\
5.25	0.99025956355974\\
5.5	0.991257287595391\\
5.75	0.992211975654739\\
6	0.994000094817536\\
6.25	0.994000094817536\\
6.5	0.994699313014734\\
6.75	0.995895065416708\\
7	0.996650555425383\\
7.25	0.997577979403753\\
7.5	0.998126600636404\\
7.75	0.998126600636404\\
8	0.998619281683418\\
};
\addlegendentry{$\chi^2$-privacy, $\lambda=0.65$}

\end{axis}
\end{tikzpicture}%
}
                \caption{ \centering \label{fig:var-chi2-utility} Utility}
            \end{subfigure}
            ~
            \begin{subfigure}{0.32\textwidth}
                 \centering
                \scalebox{0.4}{\definecolor{mycolor1}{rgb}{0.25, 0.4, 0.96}
\definecolor{mycolor2}{rgb}{0.63, 0.36, 0.94}
\definecolor{mycolor3}{rgb}{0.89, 0.26, 0.2}
\definecolor{mycolor4}{rgb}{0.24, 0.82, 0.44}
\definecolor{mycolor5}{rgb}{0.0, 1.0, 1.0}
\begin{tikzpicture}

\begin{axis}[%
width=4.521in,
height=3.8in,
scale only axis,
xmin=0,
xmax=8,
xlabel style={font=\color{white!15!black}},
xlabel={\Huge$\eps$},
ymin=0,
ymax=4,
yticklabel style={font=\large},
ytick={0.5,1,1.5,2,2.5,3,3.5,4,4.5,5,5.5},
ylabel style={font=\color{white!15!black}},
ylabel={\huge $|\log(\min_{y}\Psi(y))|$},
axis background/.style={fill=white},
xmajorgrids,
ymajorgrids,
yticklabel style = {font=\huge},
xticklabel style = {font=\huge},
legend style={legend cell align=left, align=left, draw=white!15!black, legend pos= south east,font=\huge}
]
\addplot [color=mycolor1, line width=2.0pt]
  table[row sep=crcr]{%
0.25	0.0979869750690848\\
0.5	0.208683729630452\\
0.75	0.331291808169896\\
1	0.439846345010464\\
1.25	0.541581630522294\\
1.5	0.667848315225924\\
1.75	0.793594978863712\\
2	0.9029441279805\\
2.25	1.03026577892371\\
2.5	1.14479172975389\\
2.75	1.2553982883439\\
3	1.3503677814063\\
3.25	1.4826434769649\\
3.5	1.6134881755522\\
3.75	1.72490085193563\\
4	1.84599048251564\\
4.25	1.97337879704601\\
4.5	2.06357329672677\\
4.75	2.17211683087607\\
5	2.2483954219815\\
5.25	2.33902003976229\\
5.5	2.46617036459891\\
5.75	2.53981964820346\\
6	2.64878514590437\\
6.25	2.75291406874789\\
6.5	2.83801780091972\\
6.75	2.89517736819572\\
7	3.00973136091614\\
7.25	3.08083273016365\\
7.5	3.13697716216168\\
7.75	3.14400186777479\\
8	3.18293049988807\\
};
\addlegendentry{ALIP, $\lambda=0.5$}

\addplot [color=mycolor2, line width=2.0pt]
  table[row sep=crcr]{%
0.25	0.233653090717906\\
0.5	0.492638593212515\\
0.75	0.796859908566874\\
1	1.03070118830066\\
1.25	1.27970277289152\\
1.5	1.51780002988659\\
1.75	1.74350425321451\\
2	1.90635884116478\\
2.25	2.08830754431793\\
2.5	2.23357136848089\\
2.75	2.38267316511453\\
3	2.49811399320736\\
3.25	2.61539328239662\\
3.5	2.72826564485974\\
3.75	2.8489937892729\\
4	2.97330343053007\\
4.25	3.05529297342008\\
4.5	3.1314522263583\\
4.75	3.13807385162671\\
5	3.18752079285608\\
5.25	3.25684460412378\\
5.5	3.2742324887676\\
5.75	3.29624016225027\\
6	3.35406627122796\\
6.25	3.41044396502187\\
6.5	3.40881839669977\\
6.75	3.44396688016345\\
7	3.49308590363038\\
7.25	3.52877608979285\\
7.5	3.60321905365582\\
7.75	3.65108612148159\\
8	3.65108612148159\\
};
\addlegendentry{$\ell_{1}$-privacy, $\lambda=0.5$}

\addplot [color=mycolor3, line width=2.0pt]
  table[row sep=crcr]{%
0.25	0.439416323176564\\
0.5	0.527534291459341\\
0.75	0.627945133895404\\
1	0.783103745087412\\
1.25	0.920204465529736\\
1.5	1.11082375432119\\
1.75	1.25349256607508\\
2	1.39203511882933\\
2.25	1.57829972465114\\
2.5	1.71295921255057\\
2.75	1.84157901811931\\
3	1.99472251730462\\
3.25	2.10160423019091\\
3.5	2.23446192210421\\
3.75	2.30970876208416\\
4	2.44815467685073\\
4.25	2.53998445747896\\
4.5	2.65584212213018\\
4.75	2.74425145908023\\
5	2.85840427924959\\
5.25	2.91927739303832\\
5.5	3.03775823136525\\
5.75	3.08519504054527\\
6	3.14400186777479\\
6.25	3.15725122227116\\
6.5	3.19704279342892\\
6.75	3.25032997800831\\
7	3.27998960983634\\
7.25	3.30802334632473\\
7.5	3.33447876533235\\
7.75	3.40498547755862\\
8	3.41044396502187\\
};
\addlegendentry{$\chi^2$-privacy, $\lambda=0.5$}

\addplot [color=mycolor4, line width=2.0pt]
  table[row sep=crcr]{%
0.25	0.16612546273072\\
0.5	0.405514810123675\\
0.75	0.827232765426957\\
1	1.24701268780781\\
1.25	1.61422583358307\\
1.5	1.87252136534802\\
1.75	2.09329638785189\\
2	2.28634998305705\\
2.25	2.4703951434639\\
2.5	2.61539328239662\\
2.75	2.7328456642579\\
3	2.89956433265284\\
3.25	3.03916011990915\\
3.5	3.1314522263583\\
3.75	3.16101279631628\\
4	3.25349128981235\\
4.25	3.27998960983634\\
4.5	3.32209723810234\\
4.75	3.38932760910339\\
5	3.40881839669977\\
5.25	3.44396688016345\\
5.5	3.5120672956781\\
5.75	3.60321905365582\\
6	3.65108612148159\\
6.25	3.69827882836524\\
6.5	3.69827882836524\\
6.75	3.75862595049989\\
7	3.82664319667473\\
7.25	3.85364620560398\\
7.5	3.85364620560398\\
7.75	3.85364620560398\\
8	3.9066364606673\\
};
\addlegendentry{$\ell_{1}$-privacy, $\lambda=0.65$}

\addplot [color=mycolor5, line width=2.0pt]
  table[row sep=crcr]{%
0.25	0.457541687105856\\
0.5	0.557655409428269\\
0.75	0.671511167034332\\
1	0.87873607389173\\
1.25	1.13847561871118\\
1.5	1.33112291747468\\
1.75	1.58445391867883\\
2	1.7579978456901\\
2.25	1.95639318083081\\
2.5	2.10160423019091\\
2.75	2.25462286394972\\
3	2.39826912768133\\
3.25	2.53403485895001\\
3.5	2.66491494517444\\
3.75	2.79399540980199\\
4	2.89661073375654\\
4.25	3.03775823136525\\
4.5	3.10230082219775\\
4.75	3.15526880313762\\
5	3.19704279342892\\
5.25	3.26170953147184\\
5.5	3.27998960983634\\
5.75	3.33447876533235\\
6	3.41044396502187\\
6.25	3.41044396502187\\
6.5	3.4606756742782\\
6.75	3.52877608979285\\
7	3.60321905365582\\
7.25	3.65108612148159\\
7.5	3.69827882836524\\
7.75	3.69827882836524\\
8	3.75862595049989\\
};
\addlegendentry{$\chi^2$-privacy, $\lambda=0.65$}

\end{axis}
\end{tikzpicture}%
}
                \caption{ \centering \label{fig:var-chi2-leakl} Min-lift leakage}
            \end{subfigure}
            ~
            \begin{subfigure}{0.32\textwidth}
                \centering
                \scalebox{0.4}{\definecolor{mycolor1}{rgb}{0.25, 0.4, 0.96}
\definecolor{mycolor2}{rgb}{0.63, 0.36, 0.94}
\definecolor{mycolor3}{rgb}{0.89, 0.26, 0.2}
\definecolor{mycolor4}{rgb}{0.24, 0.82, 0.44}
\definecolor{mycolor5}{rgb}{0.0, 1.0, 1.0}
\begin{tikzpicture}

\begin{axis}[%
width=4.521in,
height=3.8in,
scale only axis,
xmin=0,
xmax=8,
xlabel style={font=\color{white!15!black}},
xlabel={\Huge $\eps$},
ymin=0.1,
ymax=1,
ylabel style={font=\color{white!15!black}},
ylabel={\huge $\log (\max_{y}\Lambda(y))$},
axis background/.style={fill=white},
xmajorgrids,
ymajorgrids,
yticklabel style = {font=\huge},
xticklabel style = {font=\huge},
legend style={legend cell align=left, align=left, draw=white!15!black, legend pos= south east,font=\huge}
]
\addplot [color=mycolor1, line width=2.0pt]
  table[row sep=crcr]{%
0.25	0.0983896484372749\\
0.5	0.204378059741777\\
0.75	0.310118149234734\\
1	0.384114560021818\\
1.25	0.458665900352039\\
1.5	0.526552411862018\\
1.75	0.591767976604949\\
2	0.635359044729219\\
2.25	0.683646436575283\\
2.5	0.711884870132384\\
2.75	0.750850649242243\\
3	0.779214975451227\\
3.25	0.805859852037564\\
3.5	0.822313201673911\\
3.75	0.844271398211833\\
4	0.866336688073687\\
4.25	0.869613621612038\\
4.5	0.871116687102037\\
4.75	0.888174322847907\\
5	0.888935143380475\\
5.25	0.893250080111714\\
5.5	0.899291748887247\\
5.75	0.902792343814705\\
6	0.906609283914636\\
6.25	0.921809749486578\\
6.5	0.923835510417111\\
6.75	0.920625601740898\\
7	0.923466962112776\\
7.25	0.926428566180574\\
7.5	0.927202637974536\\
7.75	0.927202637974536\\
8	0.927414039355058\\
};
\addlegendentry{ALIP, $\lambda=0.5$}

\addplot [color=mycolor2, line width=2.0pt]
  table[row sep=crcr]{%
0.25	0.203706717669662\\
0.5	0.386694739237844\\
0.75	0.515360213211245\\
1	0.602887161837536\\
1.25	0.685127138811628\\
1.5	0.747599208088078\\
1.75	0.802876778710326\\
2	0.834501102554842\\
2.25	0.857284353541689\\
2.5	0.867245101025404\\
2.75	0.885723601486253\\
3	0.892017292455465\\
3.25	0.901281483368124\\
3.5	0.910035692475293\\
3.75	0.912932499199225\\
4	0.918437997039976\\
4.25	0.920783451582297\\
4.5	0.922761331626973\\
4.75	0.925444674069902\\
5	0.925656075450424\\
5.25	0.927894200906834\\
5.5	0.928215613570569\\
5.75	0.928215613570569\\
6	0.931072246187072\\
6.25	0.931134530929226\\
6.5	0.931134530929226\\
6.75	0.931134530929226\\
7	0.931596960303779\\
7.25	0.931596960303779\\
7.5	0.932266461481995\\
7.75	0.932266461481995\\
8	0.932266461481995\\
};
\addlegendentry{$\ell_{1}$-privacy, $\lambda=0.5$}

\addplot [color=mycolor3, line width=2.0pt]
  table[row sep=crcr]{%
0.25	0.372080125410246\\
0.5	0.424227099896387\\
0.75	0.482753440512781\\
1	0.569313907335672\\
1.25	0.627572240089266\\
1.5	0.690956400951857\\
1.75	0.725437326703194\\
2	0.771668872760247\\
2.25	0.816188014690778\\
2.5	0.831897366005508\\
2.75	0.855101594567034\\
3	0.867019138985562\\
3.25	0.874315923903862\\
3.5	0.888935143380475\\
3.75	0.892079577197619\\
4	0.895870510052912\\
4.25	0.901120713711585\\
4.5	0.909256357694583\\
4.75	0.916996521930455\\
5	0.91821980641625\\
5.25	0.920783451582297\\
5.5	0.920783451582297\\
5.75	0.924190440724164\\
6	0.927134673320802\\
6.25	0.927682799526312\\
6.5	0.927962165560568\\
6.75	0.928215613570569\\
7	0.928215613570569\\
7.25	0.930203593105777\\
7.5	0.931072246187072\\
7.75	0.931134530929226\\
8	0.931134530929226\\
};
\addlegendentry{$\chi^2$-privacy, $\lambda=0.5$}

\addplot [color=mycolor4, line width=2.0pt]
  table[row sep=crcr]{%
0.25	0.140602276822343\\
0.5	0.296439830071749\\
0.75	0.428514784548672\\
1	0.565348618126309\\
1.25	0.681201439529086\\
1.5	0.784297062675002\\
1.75	0.84948350785981\\
2	0.868029029018349\\
2.25	0.892079577197619\\
2.5	0.901281483368124\\
2.75	0.910035692475293\\
3	0.918280147198578\\
3.25	0.92043353879415\\
3.5	0.922761331626973\\
3.75	0.925444674069902\\
4	0.927894200906834\\
4.25	0.928215613570569\\
4.5	0.930203593105777\\
4.75	0.931134530929226\\
5	0.931134530929226\\
5.25	0.931134530929226\\
5.5	0.931596960303779\\
5.75	0.932266461481995\\
6	0.932266461481995\\
6.25	0.932266461481995\\
6.5	0.932266461481995\\
6.75	0.937882165482855\\
7	0.937882165482855\\
7.25	0.937882165482855\\
7.5	0.937882165482855\\
7.75	0.937882165482855\\
8	0.942896932900177\\
};
\addlegendentry{$\ell_{1}$-privacy, $\lambda=0.65$}

\addplot [color=mycolor5, line width=2.0pt]
  table[row sep=crcr]{%
0.25	0.362900751798696\\
0.5	0.404381286033877\\
0.75	0.459882614982973\\
1	0.528030832746314\\
1.25	0.611745882421323\\
1.5	0.677174547212703\\
1.75	0.764842326328604\\
2	0.831523251764993\\
2.25	0.863819902183737\\
2.5	0.874315923903862\\
2.75	0.888935143380475\\
3	0.892871558774876\\
3.25	0.898125832426004\\
3.5	0.90666152262275\\
3.75	0.920409249694121\\
4	0.920783451582297\\
4.25	0.920783451582297\\
4.5	0.924190440724164\\
4.75	0.927682799526312\\
5	0.927962165560568\\
5.25	0.928215613570569\\
5.5	0.928215613570569\\
5.75	0.931072246187072\\
6	0.931134530929226\\
6.25	0.931134530929226\\
6.5	0.931134530929226\\
6.75	0.931596960303779\\
7	0.932266461481995\\
7.25	0.932266461481995\\
7.5	0.932266461481995\\
7.75	0.932266461481995\\
8	0.937882165482855\\
};
\addlegendentry{$\chi^2$-privacy, $\lambda=0.65$}

\end{axis}
\end{tikzpicture}%
}
                \caption{ \centering \label{fig:var-chi2leaku} Max-lift leakage}
            \end{subfigure}
            ~
            \caption{\label{Fig:var-chi2-PUT} Comparison of privacy-utility tradeoff between ALIP, $\ell_{1}$-privacy, and  $\chi^2$-privacy where $|\X|=17$, $|\Sen|=5$, $\eps_{\text{LDP}}\in\{0.25,0.5,0.75,\cdots,8\}$, $\lambda \in \{0.5,0.65\}$, $\epsl=\lambda\eps$ and $\epsu=(1-\lambda)\eps$.}
        \end{figure}
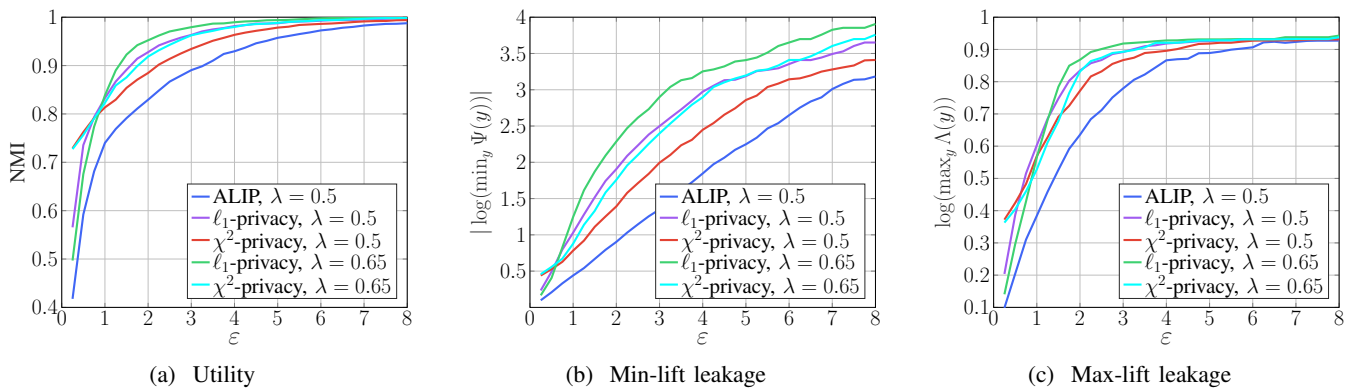
        
        \begin{figure}
            \centering
            \begin{subfigure}{0.32\textwidth}
                \scalebox{0.4}{\definecolor{mycolor1}{rgb}{0.25, 0.4, 0.96}
\definecolor{mycolor2}{rgb}{0.63, 0.36, 0.94}
\definecolor{mycolor3}{rgb}{0.89, 0.26, 0.2}
\definecolor{mycolor4}{rgb}{0.24, 0.82, 0.44}
\definecolor{mycolor5}{rgb}{0.0, 1.0, 1.0}

\begin{tikzpicture}

\begin{axis}[%
width=4.521in,
height=3.8in,
scale only axis,
xmin=0,
xmax=8,
xlabel style={font=\color{white!15!black}},
xlabel={\Huge $\eps$},
ymin=0.3,
ymax=1,
ylabel style={font=\color{white!15!black}},
ylabel={\huge NMI},
axis background/.style={fill=white},
xmajorgrids,
ymajorgrids,
yticklabel style = {font=\huge},
xticklabel style = {font=\huge},
legend style={legend cell align=left, align=left, draw=white!15!black, legend pos= south east, font=\huge}
]
\addplot [color=mycolor1,  line width=2.0pt]
  table[row sep=crcr]{%
0.25	0.417402438326139\\
0.5	0.58931222696749\\
0.75	0.681131944142307\\
1	0.735068971311797\\
1.25	0.767535635799408\\
1.5	0.793893220805491\\
1.75	0.81427007326988\\
2	0.832478604821794\\
2.25	0.849103616075165\\
2.5	0.864303204083255\\
2.75	0.87761823639748\\
3	0.890461224992069\\
3.25	0.901228199616678\\
3.5	0.911111178719206\\
3.75	0.920687756842197\\
4	0.928937539897423\\
4.25	0.937375225480885\\
4.5	0.943495658792141\\
4.75	0.949807683519909\\
5	0.95476426084505\\
5.25	0.959836889081937\\
5.5	0.964405451188499\\
5.75	0.967463824565835\\
6	0.970385886157748\\
6.25	0.973712219194275\\
6.5	0.976339561137803\\
6.75	0.978893683425077\\
7	0.980621305600287\\
7.25	0.982359448090332\\
7.5	0.984617581678389\\
7.75	0.986316642626134\\
8	0.987818251813234\\
};
\addlegendentry{$\text{ALIP, }\lambda\text{=0.5}$}

\addplot [color=mycolor2, line width=2.0pt]
  table[row sep=crcr]{%
0.25	0.781229913302124\\
0.5	0.829148915073405\\
0.75	0.858808603999732\\
1	0.881897704381183\\
1.25	0.900972251698511\\
1.5	0.915836728018808\\
1.75	0.928881570298047\\
2	0.939534636949714\\
2.25	0.947026633372185\\
2.5	0.95382110138678\\
2.75	0.959413682159091\\
3	0.963882015725606\\
3.25	0.96734138907487\\
3.5	0.970323139654668\\
3.75	0.972952527098418\\
4	0.974451843101911\\
4.25	0.975808931087381\\
4.5	0.978010330612693\\
4.75	0.979454483960098\\
5	0.981101876017141\\
5.25	0.982368188142429\\
5.5	0.983830196412785\\
5.75	0.985071499397027\\
6	0.986516915165122\\
6.25	0.987937607680315\\
6.5	0.988914119261199\\
6.75	0.989895890743185\\
7	0.991111282626461\\
7.25	0.992158885254096\\
7.5	0.993004808621135\\
7.75	0.993649072602388\\
8	0.994434436139164\\
};
\addlegendentry{$\alpha\text{-lift, }\alpha\text{=2}$}

\addplot [color=mycolor3, line width=2.0pt]
  table[row sep=crcr]{%
0.25	0.62632852975374\\
0.5	0.734986404035792\\
0.75	0.78916316238623\\
1	0.819185450281975\\
1.25	0.839961860098924\\
1.5	0.859204484524999\\
1.75	0.876201302697164\\
2	0.891361790935967\\
2.25	0.905685818146559\\
2.5	0.917441344659296\\
2.75	0.926650245048331\\
3	0.936839739319946\\
3.25	0.944665297974692\\
3.5	0.950562177538237\\
3.75	0.956041084509737\\
4	0.960776902649462\\
4.25	0.964749650798412\\
4.5	0.967551142305667\\
4.75	0.970782359861745\\
5	0.972825408018363\\
5.25	0.97455864337701\\
5.5	0.975922005336694\\
5.75	0.977755335452469\\
6	0.979492028493953\\
6.25	0.980774179203838\\
6.5	0.982380353127406\\
6.75	0.983463210435697\\
7	0.984733916852839\\
7.25	0.986142102955886\\
7.5	0.987439717864595\\
7.75	0.988878699333043\\
8	0.989784357799903\\
};
\addlegendentry{$\alpha\text{-lift, }\alpha\text{=10}$}

\addplot [color=mycolor4, line width=2.0pt]
  table[row sep=crcr]{%
0.25	0.485086942992053\\
0.5	0.640361005463037\\
0.75	0.731715559429305\\
1	0.785888135372923\\
1.25	0.816115701064282\\
1.5	0.83688945413413\\
1.75	0.857206076516598\\
2	0.874450408966886\\
2.25	0.889260073205214\\
2.5	0.903262383251579\\
2.75	0.9154089966484\\
3	0.924840225135138\\
3.25	0.935230591078591\\
3.5	0.943153755736747\\
3.75	0.949593685687396\\
4	0.95537366215586\\
4.25	0.960075428314966\\
4.5	0.963880532700656\\
4.75	0.967092390520532\\
5	0.970226002637778\\
5.25	0.972302856149697\\
5.5	0.974095702244106\\
5.75	0.975601292149442\\
6	0.977448845255868\\
6.25	0.979214239476736\\
6.5	0.98044992457865\\
6.75	0.982024953923024\\
7	0.983095766548977\\
7.25	0.984479852462062\\
7.5	0.985931393193487\\
7.75	0.987202082840636\\
8	0.988616135736731\\
};
\addlegendentry{$\alpha\text{-lift, }\alpha\text{=100}$}

\end{axis}

\begin{axis}[%
width=5.833in,
height=4.375in,
at={(0in,0in)},
scale only axis,
xmin=0,
xmax=1,
ymin=0,
ymax=1,
axis line style={draw=none},
ticks=none,
axis x line*=bottom,
axis y line*=left
]
\end{axis}
\end{tikzpicture}%
}
                \caption{ \centering \label{fig:alphautility} Utility}
            \end{subfigure}
            ~
            \begin{subfigure}{0.32\textwidth}
             \centering
                \scalebox{0.4}{\definecolor{mycolor1}{rgb}{0.25, 0.4, 0.96}
\definecolor{mycolor2}{rgb}{0.63, 0.36, 0.94}
\definecolor{mycolor3}{rgb}{0.89, 0.26, 0.2}
\definecolor{mycolor4}{rgb}{0.24, 0.82, 0.44}
\definecolor{mycolor5}{rgb}{0.0, 1.0, 1.0}
\begin{tikzpicture}

\begin{axis}[%
width=4.521in,
height=3.8in,
scale only axis,
xmin=0,
xmax=8,
xlabel style={font=\color{white!15!black}},
xlabel={\Huge$\eps$},
ymin=0,
ymax=4,
yticklabel style={font=\large},
ytick={0.5,1,1.5,2,2.5,3,3.5,4,4.5,5,5.5},
ylabel style={font=\color{white!15!black}},
ylabel={\huge $|\log(\min_{y}\Psi(y))|$},
axis background/.style={fill=white},
xmajorgrids,
ymajorgrids,
yticklabel style = {font=\huge},
xticklabel style = {font=\huge},
legend style={legend cell align=left, align=left, draw=white!15!black, legend pos= south east,font=\huge}
]
\addplot [color=mycolor1,  line width=2.0pt]
  table[row sep=crcr]{%
0.25	0.100582036927389\\
0.5	0.213001355783079\\
0.75	0.329541957595612\\
1	0.439157557409202\\
1.25	0.548252407032391\\
1.5	0.668740112247519\\
1.75	0.785397980158013\\
2	0.904339391978971\\
2.25	1.02660924244317\\
2.5	1.14258175823106\\
2.75	1.25808755181953\\
3	1.36970244060656\\
3.25	1.49550624004222\\
3.5	1.60598571675944\\
3.75	1.71764401797463\\
4	1.83457245197722\\
4.25	1.958387289018\\
4.5	2.06137050750757\\
4.75	2.16708206263675\\
5	2.27748251724914\\
5.25	2.37917032443033\\
5.5	2.47060957460946\\
5.75	2.55721386892084\\
6	2.65416994214416\\
6.25	2.74570180233496\\
6.5	2.83008775319773\\
6.75	2.90269729644482\\
7	2.95028669603723\\
7.25	3.01036420308929\\
7.5	3.09891978494779\\
7.75	3.16715868448053\\
8	3.23258876996307\\
};
\addlegendentry{$\text{ALIP, }\lambda\text{=0.5}$}

\addplot [color=mycolor2, line width=2.0pt]
  table[row sep=crcr]{%
0.25	0.842470023833996\\
0.5	1.02195112617699\\
0.75	1.19430645839546\\
1	1.37574193498948\\
1.25	1.57389527418892\\
1.5	1.72734276138546\\
1.75	1.90008183751967\\
2	2.04970021623035\\
2.25	2.15754101980143\\
2.5	2.26909112153058\\
2.75	2.38249916460772\\
3	2.47362048837401\\
3.25	2.54657976489619\\
3.5	2.63742527466781\\
3.75	2.69953841561094\\
4	2.74166638425427\\
4.25	2.79045640806882\\
4.5	2.85360430022337\\
4.75	2.91247121389761\\
5	2.97176522530912\\
5.25	3.00751661262388\\
5.5	3.05554086745853\\
5.75	3.10337743209854\\
6	3.16457183002785\\
6.25	3.23337912426709\\
6.5	3.27856172868163\\
6.75	3.32761376117986\\
7	3.38475596926551\\
7.25	3.43257173434631\\
7.5	3.48682034387307\\
7.75	3.525300889276\\
8	3.5689401113985\\
};
\addlegendentry{$\alpha\text{-lift, }\alpha\text{=2}$}

\addplot [color=mycolor3, line width=2.0pt]
  table[row sep=crcr]{%
0.25	0.417480596805605\\
0.5	0.652911547170587\\
0.75	0.831281884470966\\
1	0.969041298935639\\
1.25	1.09080196303868\\
1.5	1.20891670866922\\
1.75	1.33469908137587\\
2	1.47762727954929\\
2.25	1.64626073530153\\
2.5	1.76866746666058\\
2.75	1.87193311418779\\
3	2.02274722130871\\
3.25	2.1384352754217\\
3.5	2.22769571485\\
3.75	2.32994669040504\\
4	2.41754973613196\\
4.25	2.49740784558178\\
4.5	2.55910787047663\\
4.75	2.64286553034791\\
5	2.69819536384551\\
5.25	2.73828507137673\\
5.5	2.78654166033346\\
5.75	2.8522361003249\\
6	2.91292016850855\\
6.25	2.96271392743684\\
6.5	3.00385262594925\\
6.75	3.04317766150322\\
7	3.09305048251064\\
7.25	3.15285304384065\\
7.5	3.21020790830005\\
7.75	3.26921462370087\\
8	3.31914826420905\\
};
\addlegendentry{$\alpha\text{-lift, }\alpha\text{=10}$}

\addplot [color=mycolor4, line width=2.0pt]
  table[row sep=crcr]{%
0.25	0.369387397047182\\
0.5	0.480037268915212\\
0.75	0.662338464632749\\
1	0.835778751995378\\
1.25	0.968494215864646\\
1.5	1.07012300957604\\
1.75	1.2041797040609\\
2	1.33609231871209\\
2.25	1.45939239742685\\
2.5	1.60857194778734\\
2.75	1.73887233818476\\
3	1.84439655302841\\
3.25	1.99940731925203\\
3.5	2.12248510586227\\
3.75	2.21016554037399\\
4	2.31993261005632\\
4.25	2.40434415525764\\
4.5	2.47226280278133\\
4.75	2.54222471281024\\
5	2.63066971802018\\
5.25	2.6775629254248\\
5.5	2.72538754421149\\
5.75	2.77068704520374\\
6	2.84261747168182\\
6.25	2.90092184822739\\
6.5	2.9473172314897\\
6.75	2.99856777987045\\
7	3.0327023860062\\
7.25	3.08141038596751\\
7.5	3.14407211352351\\
7.75	3.20080905909809\\
8	3.25971576034639\\
};
\addlegendentry{$\alpha\text{-lift, }\alpha\text{=100}$}

\end{axis}
\end{tikzpicture}%
}
                \caption{ \centering \label{fig:alphaleakL} Min-lift leakage}
            \end{subfigure}
            ~
            \begin{subfigure}{0.32\textwidth}
             \centering
                \scalebox{0.4}{\definecolor{mycolor1}{rgb}{0.25, 0.4, 0.96}
\definecolor{mycolor2}{rgb}{0.63, 0.36, 0.94}
\definecolor{mycolor3}{rgb}{0.89, 0.26, 0.2}
\definecolor{mycolor4}{rgb}{0.24, 0.82, 0.44}
\definecolor{mycolor5}{rgb}{0.0, 1.0, 1.0}
\begin{tikzpicture}

\begin{axis}[%
width=4.521in,
height=3.8in,
scale only axis,
xmin=0,
xmax=8,
xlabel style={font=\color{white!15!black}},
xlabel={\Huge $\eps$},
ymin=0,
ymax=1,
ylabel style={font=\color{white!15!black}},
ylabel={\huge $\log (\max_{y}\Lambda(y))$},
axis background/.style={fill=white},
xmajorgrids,
ymajorgrids,
yticklabel style = {font=\huge},
xticklabel style = {font=\huge},
legend style={legend cell align=left, align=left, draw=white!15!black, legend pos= south east,font=\huge}
]
\addplot [color=mycolor1, line width=2.0pt]
  table[row sep=crcr]{%
0.25	0.0979692764783066\\
0.5	0.203130517313862\\
0.75	0.307464415242383\\
1	0.390228819533383\\
1.25	0.46513973721931\\
1.5	0.537629482040508\\
1.75	0.596626019912188\\
2	0.647166448667881\\
2.25	0.69489611074371\\
2.5	0.728755975507446\\
2.75	0.756680081703568\\
3	0.781117572023144\\
3.25	0.803670542073382\\
3.5	0.821866362591287\\
3.75	0.839142266161572\\
4	0.851830941666598\\
4.25	0.864524364285055\\
4.5	0.875078995237106\\
4.75	0.883267909043203\\
5	0.889320272980626\\
5.25	0.899534843039318\\
5.5	0.904761681238586\\
5.75	0.907298842261833\\
6	0.911558298611934\\
6.25	0.915823343273453\\
6.5	0.918813261909352\\
6.75	0.922606206805521\\
7	0.92596969139713\\
7.25	0.929497235490041\\
7.5	0.933745218736901\\
7.75	0.934961966513487\\
8	0.936504264670947\\
};
\addlegendentry{$\text{ALIP, }\lambda\text{=0.5}$}

\addplot [color=mycolor2, line width=2.0pt]
  table[row sep=crcr]{%
0.25	0.500968739322374\\
0.5	0.576694619865726\\
0.75	0.647858098263864\\
1	0.700110086862888\\
1.25	0.7489044014482\\
1.5	0.779779174107889\\
1.75	0.809859048613958\\
2	0.837043044096659\\
2.25	0.856226158035356\\
2.5	0.864250429965327\\
2.75	0.876060599170194\\
3	0.887069198770031\\
3.25	0.894599474708054\\
3.5	0.900423768213634\\
3.75	0.906660253232139\\
4	0.909331519677325\\
4.25	0.913543879739163\\
4.5	0.915859073906214\\
4.75	0.920084171017176\\
5	0.922363755526713\\
5.25	0.925275606842995\\
5.5	0.928827412878695\\
5.75	0.932218549982075\\
6	0.935158389387001\\
6.25	0.93675692624116\\
6.5	0.938224438973097\\
6.75	0.939665579739202\\
7	0.940555267845906\\
7.25	0.942249981973743\\
7.5	0.943812975641967\\
7.75	0.944134656386902\\
8	0.944203383488969\\
};
\addlegendentry{$\alpha\text{-lift, }\alpha\text{=2}$}

\addplot [color=mycolor3, line width=2.0pt]
  table[row sep=crcr]{%
0.25	0.28531535610109\\
0.5	0.415612118056687\\
0.75	0.494206261560791\\
1	0.559066031797981\\
1.25	0.607192897286293\\
1.5	0.662872621735395\\
1.75	0.705540214832501\\
2	0.742706208686989\\
2.25	0.773552075977205\\
2.5	0.798605032373118\\
2.75	0.815855730290971\\
3	0.83476117675504\\
3.25	0.852696360871929\\
3.5	0.864325685831111\\
3.75	0.874262234547306\\
4	0.882791706851473\\
4.25	0.892290926900527\\
4.5	0.896283553824932\\
4.75	0.901147631207624\\
5	0.907285557693917\\
5.25	0.910795437429487\\
5.5	0.912996398001618\\
5.75	0.915805101527944\\
6	0.921392454028472\\
6.25	0.922040661096743\\
6.5	0.925861247349447\\
6.75	0.928723691468485\\
7	0.931288301790626\\
7.25	0.935604986913685\\
7.5	0.93658725410418\\
7.75	0.938475381510285\\
8	0.939705894570025\\
};
\addlegendentry{$\alpha\text{-lift, }\alpha\text{=10}$}

\addplot [color=mycolor4, line width=2.0pt]
  table[row sep=crcr]{%
0.25	0.2087828377091\\
0.5	0.306155916130048\\
0.75	0.409677149032344\\
1	0.487227229257588\\
1.25	0.553041281547124\\
1.5	0.594713817107814\\
1.75	0.65653191761271\\
2	0.699709137197092\\
2.25	0.738587523736802\\
2.5	0.76960471237851\\
2.75	0.793995009066618\\
3	0.812212375703599\\
3.25	0.831690004696771\\
3.5	0.850595812305237\\
3.75	0.862165820218261\\
4	0.872666217405475\\
4.25	0.881755245532987\\
4.5	0.889369848116652\\
4.75	0.895220650190388\\
5	0.899776404122768\\
5.25	0.906772514747586\\
5.5	0.910507944648035\\
5.75	0.911778458513989\\
6	0.915212632638278\\
6.25	0.919989816897318\\
6.5	0.921774156579969\\
6.75	0.925834845226239\\
7	0.927803124224881\\
7.25	0.931029476144601\\
7.5	0.935501307803121\\
7.75	0.936258488071086\\
8	0.938113727517666\\
};
\addlegendentry{$\alpha\text{-lift, }\alpha\text{=100}$}

\end{axis}

\begin{axis}[%
width=5.833in,
height=4.375in,
at={(0in,0in)},
scale only axis,
xmin=0,
xmax=1,
ymin=0,
ymax=1,
axis line style={draw=none},
ticks=none,
axis x line*=bottom,
axis y line*=left
]
\end{axis}
\end{tikzpicture}%
}
                \caption{ \centering \label{fig:alphaleakU} Max-lift leakage}
            \end{subfigure}
            ~
            \caption{\label{Fig:alphaPUT} Comparison of privacy-utility tradeoff between ALIP and $\alpha$-lift-privacy, where $|\X|=17$, $|\Sen|=5$, $\eps_{\text{LDP}}\in\{0.25,0.5,0.75,\cdots,8\}$, $\epsl=\epsu=\frac{\eps_{\text{LDP}}}{2}$, and $\alpha \in \{2,10,100\}$.}
        \end{figure}
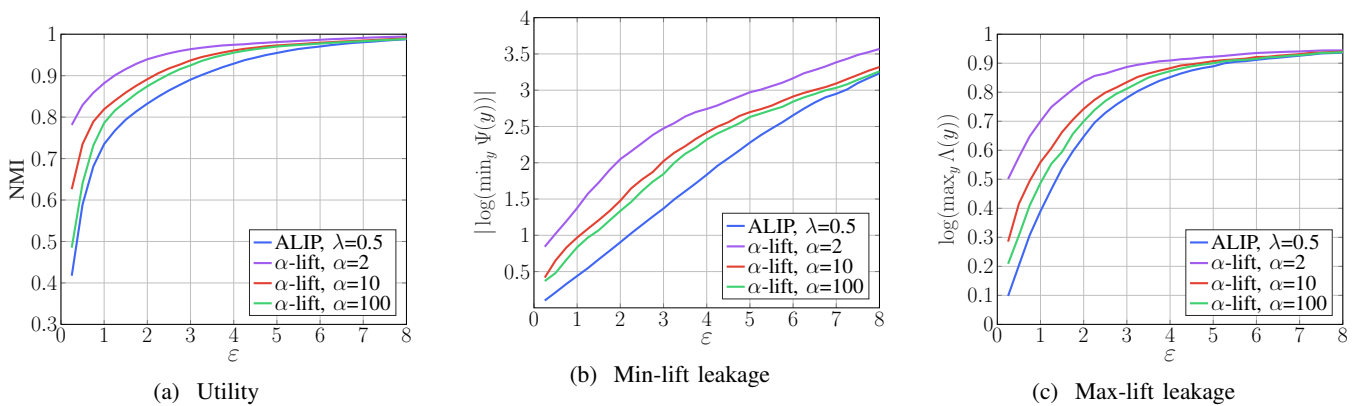
% %=================================================================
% %=================================================================
\section{Conclusions}

    In this paper, we have studied lift, the likelihood ratio between posterior and prior belief about sensitive features in a dataset.
    We demonstrated the distinction between the min and max lifts in terms of data privacy concerns.
    We proposed ALIP as a generalized version of LIP to have a more compatible notion of privacy with lift asymmetry. 
    ALIP can enhance utility in the watchdog and ORR mechanisms, two main approaches to achieve lift-based privacy.
    We proposed two subset randomization methods to enhance the utility of the watchdog mechanism and reduce ORR complexity for large datasets.  
    We also investigated the existing lift-based measures, showing that they could incur significant leakage on the min lift. Thus, we proposed lift-inverse measures to restrict the min-lift leakage. 
    Finally, we applied the watchdog mechanism to study the PUT of lift-based and lift-inverse measures. 
    For future work, one can consider the applicable operational meaning of the min-lift and max-lift.
    Subset randomization can be applied to decrease the complexity and enhance the utility of other privacy mechanisms.
    Moreover, optimal randomization for $\alpha$-lift is also unknown and could be considered.

\appendices
\section{}\label{apx:proof of ALIP-LDP-Average}
    \begin{enumerate}
        \item
            In LDP, for all $y \in \Y$, we have
            \begin{align*}
                \Gamma(y) =\sup_{s,s^{\prime} \in \Sen}\frac{\PYgS(y|s)}{\PYgS(y|s^{\prime})} =\frac{\max_{s \in \Sen}\PYgS(y|s)}{\min_{s \in \Sen}\PYgS(y|s)}
                =\frac{\max_{s \in \Sen}\PYgS(y|s)/\PY(y)}{\min_{s \in \Sen}\PYgS(y|s)/\PY(y)}
                =\frac{\Lambda(y)}{\Psi(y)} \leq \frac{\e^{\epsu}}{\e^{-\epsl}}=\e^{\epsl+\epsu}.
            \end{align*}
        %*********************************************************************************
        \item
            For MI, we have $$\displaystyle I(S;Y)=\E_{\PSY}[i(S,Y)] \leq \E_{\PSY}[\epsu]=\epsu.$$
        %*********************************************************************************
        \item
            \begin{itemize}
            \item For the total variation distance, we have
                \begin{align*}
                    \displaystyle T(S;Y) &=    \frac{1}{2}\sum_{y \in \Y}\PY(y)\sum_{s \in \Sen}\PS(s)|\lifty-1| 
                                          \leq \frac{1}{2}\sum_{y \in \Y}\PY(y)\sum_{s \in \Sen}\PS(s)|\Lambda(y)-1|\\
                                         &=\frac{1}{2}\sum_{y \in \Y}\PY(y)|\Lambda(y)-1|\leq \frac{1}{2}\sum_{y \in \Y}\PY(y)|\e^{\epsu}-1|=\frac{1}{2}(\e^{\epsu}-1).
                \end{align*}
            %***************************************************************
            \item For $\chi^2$-divergence, we have
                \begin{align*}
                    \chi^2(S;Y) &= \sum_{y \in \Y}\PY(y)\sum_{s \in \Sen}\PS(s)\big(\lifty-1\big)^2 
                                \leq \sum_{y \in \Y}\PY(y)\sum_{s \in \Sen}\PS(s)\big(\Lambda(y)-1\big)^2\\
                    &= \sum_{y \in \Y}\PY(y)\big(\Lambda(y)-1\big)^2 \leq \sum_{y \in \Y}\PY(y)\big(\e^{\epsu}-1\big)^2 = (\e^{\epsu}-1)^2.
                \end{align*}
        \end{itemize}
        %*********************************************************************************
        \item
        \begin{itemize}
            \item
            For Sibson MI, we have
            \begin{align*}
                I_{\alpha}^{S}(S;Y) 
                &= \frac{\alpha}{\alpha-1} \log \sum_{y \in \Y} \PY(y) \left( \sum_{s \in \Sen} \PS(s) \lifty^{\alpha} \right)^{1/\alpha}  \\
                &\leq \frac{\alpha}{\alpha-1} \log \sum_{y \in \Y} \PY(y) \left( \sum_{s \in \Sen} \PS(s) \Lambda(y)^{\alpha} \right)^{1/\alpha}\\
                 &\leq \frac{\alpha}{\alpha-1} \log \sum_{y \in \Y} \PY(y) \left( \sum_{s \in \Sen} \PS(s)\e^{\epsu\alpha} \right)^{1/\alpha}=\frac{\epsu\alpha}{\alpha-1}.
            \end{align*}
  
            \item For Arimoto MI, we have
            \begin{align*}
                I_{\alpha}^{A}(S;Y) 
                &=\frac{\alpha}{\alpha-1}\log \sum_{y \in \Y}\PY(y) \left(\sum_{s \in \Sen} P_{S_{\alpha}}(s)\lifty^{\alpha}\right)^{1/\alpha}\\
                &\leq \frac{\alpha}{\alpha-1}\log \sum_{y \in \Y}\PY(y) \left( \sum_{s \in \Sen}P_{S_{\alpha}}(s) \Lambda(y)^{\alpha} \right)^{1/\alpha}\\
                &\leq \frac{\alpha}{\alpha-1}\log \sum_{y \in \Y}\PY(y) \left( \sum_{s \in \Sen}P_{S_{\alpha}}(s) \e^{\epsu\alpha} \right)^{1/\alpha}= \frac{\epsu\alpha}{\alpha-1},
            \end{align*}
            where $P_{S_{\alpha}}(s)=\frac{\PS(s)^{\alpha}}{\sum_{s \in \Sen} \PS(s)^{\alpha} }.$
        \end{itemize}
    \end{enumerate}  
%=================================================================
%=================================================================
\section{}\label{apx:proof of minimum leakage LDP}

    Here, we prove that $X$-invariant randomization minimizes privacy leakage in $\XH$ for LDP.
    \begin{Proposition}\label{prop:LDP watchdog}
        A randomization $r(y|x)$, $x,y, \in \XH$ can attain $(\eps,\XH)$-LDP if and only if:
        \begin{equation}\label{LDPXH}
            \Gamma_{LDP}(\XH)=\frac{\max_{s \in \Sen}P(\XH|s)}{\min_{s \in \Sen}P(\XH|s)} \leq \e^\eps.
        \end{equation}
    \end{Proposition}
    
    \begin{proof}
        Sufficient condition: Consider an $X$-invariant randomization where $r(y|x)=\R(y)$, $\forall x \in \XH$ and  $y \in \YH$. If \eqref{LDPXH} holds, then for all  $s,s' \in \Sen$, we have
        \begin{equation*}
            \begin{aligned}
                \frac{P(\XH|s)}{P(\XH|s')} & \leq \frac{\max_{s \in \Sen}P(\XH|s)}{\min_{s \in \Sen}P(\XH|s)} \leq \e^\eps  \Rightarrow \\
                P(\XH|s)   & \leq P(\XH|s') \e^\eps \Rightarrow\\
                \R(y)P(\XH|s)  & \leq \R(y)P(\XH|s')\e^\eps \Rightarrow\\
                \sum_{x\in\XH} r(y|x) \PXgS(x|s)   & \leq  \sum_{x\in\XH} r(y|x) \PXgS(x|s') \e^\eps
                \Rightarrow \\
                \PYgS(y|s) & \leq  \PYgS(y|s') \e^{\eps}.
            \end{aligned}
        \end{equation*}
        For the necessary condition, note that for all $s,s'\in\Sen$ and $y \in \YH$, we have
        \begin{equation*}
            \begin{aligned}
                \PYgS(y|s)\leq \e^{\eps}\PYgS(y|s')  \Rightarrow 
                \sum_{x \in \XH} r(y|x)\PXgS(x|s) \leq \e^{\eps} \sum_{x \in \XH}r(y|x)\PXgS(x|s'),
            \end{aligned}
        \end{equation*}
        then by a summation over all $ y \in \YH$ on both sides we get
        \begin{align}
            \nonumber &\sum_{x\in\XH}\PXgS(x|s)\underbrace{\sum_{y\in\YH} r(y|x)}_{=1}
            \leq \e^{\eps}\sum_{x \in \XH} \PXgS(x|s') \underbrace{\sum_{y \in \YH} r(y|x)}_{=1}\\
            &\Rightarrow P(\XH|s)  \leq P(\XH|s')\e^\eps \Rightarrow \frac{P(\XH|s)}{P(\XH|s')}\leq\e^\eps. \label{eq:proof 1}
        \end{align}
        Because \eqref{eq:proof 1}  holds for all $s,s' \in \Sen$, we have
        \begin{equation*}
            \max_{s,s' \in \Sen}\frac{P(\XH|s)}{P(\XH|s')}=\frac{\max_{s \in \Sen}P(\XH|s)}{\min_{s \in \Sen}P(\XH|s)} \leq \e^\eps.
        \end{equation*}
        
    \end{proof}

%=================================================================
%=================================================================

\section{}\label{apx:proof of ALIP-lift-base}

    \begin{enumerate}
    
        \item For $\ell_{1}$-lift, we have
            \begin{align*}
                \Lambda_{\ell_{1}}(y)=\sum_{s \in \Sen}\PS(s)|\lifty-1|\leq\sum_{s \in \Sen}\PS(s)|\Lambda(y)-1| =\Lambda(y)-1 \leq \e^{\epsu}-1.
            \end{align*}
        
        \item For $\chi^2$-lift, we have
            \begin{align*}
                \Lambda_{\chi^2}(y)= \sum_{s \in \Sen}\PS(s) \big(\lifty-1\big)^2 \leq \sum_{s \in \Sen}\PS(s)\left(\Lambda(y)-1\right)^2
                =\left(\Lambda(y)-1\right)^2 \leq \left(\e^{\epsu}-1\right)^2.
            \end{align*}
        
        \item For the $\alpha$-lift, we have
            \begin{align*}
                \Lambda_{\alpha}^{S}(y)=\left(\sum_{s \in \Sen}\PS(s)\lifty^{\alpha}\right)^{1/\alpha}\leq 	\left(\sum_{s \in \Sen}\PS(s)\Lambda(y)^{\alpha}\right)^{1/\alpha}  =\Lambda(y)\leq \e^{\epsu}.
            \end{align*}
                
    \end{enumerate}   
%=================================================================
%=================================================================
\section{}\label{apx:proof of lif-base-ALIP}
    If $\displaystyle s_{y}=\argmax_{s \in \Sen}\lifty$, then we have $\Lambda(y)=l(s_y,y)$. 
    Recall that $\displaystyle \bar{y}=\argmax_{y \in \Y}[\Lambda(y)]$.
    
    \begin{enumerate}
    
        \item   When  $\displaystyle \max_{y \in \Y}\Lambda_{\ell_{1}}(y) \leq {\eps}$,  for all $y \in \Y$, we have
                \begin{align*}
                     \Lambda_{\ell_{1}}(y)&=\sum_{s \in \Sen}\PS(s) |\lifty-1| \leq {\eps}  \Rightarrow
                    \PS(s_{y})|l(s_{y},y)-1| =\PS(s_{y})(\Lambda(y)-1) \leq {\eps},
                \end{align*}
                which results in $$ \max_{y \in \Y}\Lambda(y) \leq \frac{{\eps}}{\PS({s_{\bar{y}}})}+1.$$

        \item   When  $\displaystyle \max_{y \in \Y}\Lambda_{\chi^2}(y) \leq {\eps}$,  for all $y \in \Y$, we have
                \begin{align*}
                     \Lambda_{\chi^2}(y) &=\sum_{s \in \Sen}\PS(s) \left(\lifty-1\right)^2 \leq {\eps}  \Rightarrow
                    \PS({s_{{y}}})\left(l(s_{y},y)-1\right)^2 =\PS(s_{y})(\Lambda(y)-1)^2 \leq {\eps},
                \end{align*}
                which results in $$ \max_{y \in \Y}\Lambda(y) \leq \sqrt{\frac{\eps}{\PS({s_{\bar{y}}})}}+1.$$

        \item   When  $\displaystyle \max_{y \in \Y}\Lambda_{\alpha}^{S}(y) \leq {\eps}$, for all $y \in \Y$, we have
                \begin{align*}
                    \Lambda_{\alpha}^{S}(y) &=\left(\sum_{s \in \Sen}\PS(s)\lifty^{\alpha}\right)^{1/\alpha}\leq {\eps} \Rightarrow
                    \PS(s_{y})l(s_{y},y)^{\alpha} =\PS(s_{y})\Lambda(y)^{\alpha}\leq {\eps^\alpha},
                \end{align*}
                which results in $$ \max_{y \in \Y}\Lambda(y) \leq \frac{\eps}{\PS({s_{\bar{y}}})^{\frac{1}{\alpha}}}.$$

    \end{enumerate}       
%=================================================================
%=================================================================
\section{}\label{apx:proof of ALIP-lift-inverse}
    Since $(\epsl,\epsu)$-ALIP is satisfied, for all $y \in \Y$, we have $\e^{-\epsu}\leq\frac{1}{\lifty}\leq\e^{\epsl}$ and $\displaystyle \max_{s}\left(\frac{1}{\lifty}\right)=\frac{1}{\Psi(y)}$.
    \begin{enumerate}
    
        \item For $\ell_{1}$-lift-inverse, we have
        \begin{align*}
            \Psi_{\ell_{1}}(y) =\sum_{s \in \Sen}\PS(s)\left|\frac{1}{\lifty}-1\right|\leq\sum_{s \in \Sen}\PS(s)\left|\frac{1}{\Psi(y)}-1\right|
            =\frac{1}{\Psi(y)}-1 \leq \e^{\epsl}-1 .
        \end{align*}
        
        \item For $\chi^2$-lift-inverse, we have
        \begin{align*}
            \Psi_{\chi^{2}}(y) &=\sum_{s \in \Sen}\PS(s)\left(\frac{1}{\lifty}-1\right)^{2} \leq\sum_{s \in \Sen}\PS(s)\left(\frac{1}{\Psi(y)}-1\right)^{2} =\left(\frac{1}{\Psi(y)}-1\right)^2 \leq \left(\e^{\epsl}-1\right)^2.
        \end{align*}
        
        \item For $\alpha$-lift-inverse, we have
        \begin{align*}
            \Psi_{\alpha}^{S}(y)=\left(\sum_{s \in \Sen}\PS(s)\left(\frac{1}{\lifty}\right)^{\alpha}\right)^{\frac{1}{\alpha}}
            \leq\left(\sum_{s \in \Sen}\PS(s)\left(\frac{1}{\Psi(y)}\right)^{\alpha}\right)^{\frac{1}{\alpha}}
            =\frac{1}{\Psi(y)}\leq \e^{\epsl}.
        \end{align*}
        
    \end{enumerate}
%=================================================================
%=================================================================
\section{}\label{apx:proof of lift-inverse-ALIP}

    If $\displaystyle s_{y}=\argmin_{s}\lifty$, then we have $\Psi(y)=l(s_{y},y)$. Recall that $\displaystyle \underline{y}=\argmin_{y}[\Psi(y)].$
    \begin{enumerate}
        \item When  $\displaystyle \max_{y \in \Y}\Psi_{\ell_{1}}(y)\leq{\eps}$, for all $y \in \Y$, we have
        \begin{align*}
            \Psi_{\ell_{1}}(y)&=\sum_{s \in \Sen}\PS(s) \left|\frac{1}{\lifty}-1\right| \leq {\eps}  \Rightarrow
            \PS(s_{y})\left|\frac{1}{l(s_{y},y)}-1\right|=\PS(s_{y})\left(\frac{1}{\Psi(y)}-1\right) \leq {\eps},
        \end{align*}
        which results in  $$ \min_{y \in \Y}\Psi(y) \geq \frac{ \PS(s_{\underline{y}}) }{\eps+{\PS(s_{\underline{y}})}}.$$
        
        \item  When  $\displaystyle \max_{y \in \Y}\Psi_{\chi^2}(y) \leq {\eps}$, for all $y \in \Y$, we have
        \begin{align*}
            \Psi_{\chi^2}(y) &=\sum_{s \in \Sen}\PS(s) \left(\frac{1}{\lifty}-1\right)^2 \leq {\eps}  \Rightarrow
            \PS(s_{y})\left(\frac{1}{l(s_{y},y)}-1\right)^2 =\PS(s_{y})\left(\frac{1}{\Psi(y)}-1\right)^2 \leq {\eps},
        \end{align*}
        which results in $$ \min_{y \in \Y} \Psi(y) \geq \frac{  \sqrt{ \PS(s_{\underline{y}}) }  }{  \sqrt{\eps} + \sqrt{\PS(s_{\underline{y}})}}.$$
        
        \item  When  $\displaystyle \max_{y \in \Y}\Psi_{\alpha}^{S}(y) \leq {\eps}$, for all $y \in \Y$, we have
        \begin{align*}
            \Psi_{\alpha}^{S}(y)&= \left(\sum_{s \in \Sen}\PS(s)\left(\frac{1}{l(s,y)}\right)^{\alpha}\right)^{1/\alpha} \leq {\eps}
            \Rightarrow \PS(s_{y})\left(\frac{1}{l(s_{y},y)}\right)^{\alpha}=\PS(s_{y})\left(\frac{1}{\Psi(y)}\right)^{\alpha}\leq {\eps^{\alpha}},
        \end{align*}
        which results in  $$ \min_{y \in \Y}\Psi(y) \geq {\eps ^{-1}}\PS(s_{\underline{y}})^{\frac{1}{\alpha}}. $$
    \end{enumerate}
%=================================================================
%=================================================================
\bibliographystyle{IEEEtran}
\bibliography{bare_jrnl}

\end{document}